\documentclass[11pt,letterpaper]{article}
\usepackage[T1]{fontenc}
\usepackage{lmodern}
\usepackage[margin=1in]{geometry}
\usepackage{graphicx}

\usepackage{amsmath,amssymb,amsthm,mathtools}
\usepackage{aliascnt}
\usepackage{booktabs}
\usepackage{enumitem}
\usepackage[ruled,vlined,linesnumbered]{algorithm2e}
\SetAlFnt{\footnotesize}
\SetAlCapFnt{\small}
\SetAlCapNameFnt{\small}
\SetAlgoNlRelativeSize{0}
\SetNlSty{scriptsize}{}{}

\usepackage{placeins}
\usepackage{microtype}
\usepackage{xurl}
\usepackage{xcolor}
\usepackage{tikz}
\definecolor{mydarkblue}{rgb}{0,0.08,0.45}
\usepackage[numbers,sort&compress]{natbib}
\usepackage[colorlinks=true,linkcolor=mydarkblue,citecolor=mydarkblue,urlcolor=mydarkblue]{hyperref}
\usepackage[capitalize]{cleveref}
\crefname{algocf}{algorithm}{algorithms}
\Crefname{algocf}{Algorithm}{Algorithms}

\newtheorem{theorem}{Theorem}[section]
\newaliascnt{lemma}{theorem}
\newtheorem{lemma}[lemma]{Lemma}
\aliascntresetthe{lemma}
\newaliascnt{corollary}{theorem}
\newtheorem{corollary}[corollary]{Corollary}
\aliascntresetthe{corollary}
\newaliascnt{proposition}{theorem}
\newtheorem{proposition}[proposition]{Proposition}
\aliascntresetthe{proposition}
\newaliascnt{conjecture}{theorem}

\aliascntresetthe{conjecture}
\newtheorem*{question}{Question}
\newaliascnt{claim}{theorem}

\aliascntresetthe{claim}
\crefname{claim}{claim}{claims}
\Crefname{claim}{Claim}{Claims}
\newif\ifgainrestate
\newcommand{\gainlabel}[1]{\ifgainrestate\tag{\ref{#1}}\else\label{#1}\fi}
\theoremstyle{definition}
\newaliascnt{definition}{theorem}
\newtheorem{definition}[definition]{Definition}
\aliascntresetthe{definition}
\newaliascnt{problem}{theorem}

\aliascntresetthe{problem}
\crefname{problem}{problem}{problems}
\Crefname{problem}{Problem}{Problems}
\newaliascnt{remark}{theorem}
\newtheorem{remark}[remark]{Remark}
\aliascntresetthe{remark}
\usepackage{thmtools,thm-restate}

\newcommand{\ind}[1]{\mathbf{1}_{#1}}
\DeclareMathOperator*{\E}{\mathsf{E}}
\let\Pr\relax
\DeclareMathOperator*{\Pr}{\mathsf{Pr}}
\DeclareMathOperator{\Var}{\mathsf{Var}}
\DeclareMathOperator{\Cov}{Cov}
\DeclareMathOperator{\rank}{rank}

\DeclareMathOperator{\im}{im}

\newcommand{\loss}{\mathcal{L}}
\newcommand{\R}{\mathbb R}

\newcommand{\xcol}{A_h}   \newcommand{\Bmat}{A_M}   \newcommand{\Sig}{\Sigma}

\newcommand{\Info}{\mathcal I}
\newcommand{\Vinfo}{\mathcal V}
\newcommand{\ahat}{a}

\newcommand{\oracle}{\psi}

\DeclareMathOperator{\tr}{tr}
\DeclareMathOperator{\diag}{diag}

\title{\bfseries No Escape from the Quadratic Barrier:\\
Adaptive Attacks on Linear $\ell_2$ Sketches}
\author{
  Edith Cohen\\
  \small Google Research\\
  \small Tel Aviv University\\
  \small \href{mailto:edith@cohenwang.com}{\texttt{edith@cohenwang.com}}
  \and
  Uri Stemmer\\
  \small Tel Aviv University\\
  \small Google Research\\
  \small \href{mailto:u@uri.co.il}{\texttt{u@uri.co.il}}
}
\date{}
\hypersetup{
  pdftitle={No Escape from the Quadratic Barrier: Adaptive Attacks on Linear l2 Sketches},
  pdfauthor={Edith Cohen and Uri Stemmer}
}

\begin{document}
\maketitle

\begin{abstract}
Linear sketches estimate properties of high-dimensional inputs from a fixed low-dimensional projection. At constant accuracy, randomized Euclidean-norm sketches, notably the Johnson–Lindenstrauss transform, support exponentially many \emph{nonadaptive} queries in sketch dimension $k$. For adaptive black-box queries, known privacy-based defenses guarantee $\widetilde\Omega(k^2)$ and it remained open whether stronger guarantees were possible. Prior attacks used more queries, restricted the responder, or provided weaker failure guarantees.
 We show that the quadratic guarantee is optimal up to logarithmic factors, even with a computationally unbounded responder that knows the matrix, attack algorithm, history, and query distributions and may change its randomized rule after every query. For every fixed hidden $A\in\mathbb{R}^{k\times n}$ with $n>40k$ and target accuracy $\alpha \in (0,1]$, our black-box attack uses $O(k+\alpha^2k^2\log(2+\alpha^2k))$ queries in $\widetilde O(1+\alpha^2k)$ adaptive rounds to force, with constant probability, a round in which a constant fraction of answers have $\Omega(\alpha)$ relative error. Our results clarify the limits of defending a fixed compressive linear representation by adapting the prediction rule.

\end{abstract}

\section{Introduction}
\label{sec:introduction}

A system that is reliable on independently chosen inputs may fail under
adaptive interaction, where new queries are formed after observing answers that may reveal internal parameters.
In machine learning, a model can have small error on its test distribution
while searches guided by its outputs find failures
\citep{PapernotMGJCS17,IlyasEAL18}. Defenses against language-model attacks can be bypassed through adaptive interaction
\citep{NasrEtAl25}.  The robustness of a system is studied by quantifying how many interactions
are needed to find a failure, and by how much a responder can delay this
process by controlling what its answers reveal.

We study this question for the fundamental task of estimating the Euclidean norm of a
vector from a \emph{linear sketch}.  A fixed matrix $A\in\R^{k\times n}$
compresses each input $v\in\R^n$ to $k$ linear measurements $Av$.
A \emph{responder} uses this sketch  to estimate $\|v\|_2$.
The \emph{black-box} attacker chooses inputs and sees the returned estimates, but never sees
$A$ or the sketches.
Note that when $k<n$, inputs with different norms can have the same sketch, so errors exist. The aim of the attacker is to force failures \emph{efficiently}.

\paragraph{The quadratic benchmark.}
At constant relative accuracy, classical randomized constructions,
including the Johnson--Lindenstrauss transform and the AMS sketch
\citep{JohnsonLindenstrauss84,AlonMS99}, support exponentially many \emph{nonadaptive}
queries in the sketch dimension $k$.  Under adaptivity, the immediate guarantee is $\widetilde{\Omega}(k)$, by using a fresh sketch for each query. Defenses
based on differential privacy or subsampling support $\widetilde\Omega(k^2)$ queries
\citep{HassidimKMMS20,Blanc23}, by using a randomized estimator with JL or AMS sketches.
These defenses draw on adaptive data analysis: limiting what each answer reveals allows repeated use of the same randomness \citep{DworkFHPRR15,BassilyNSSSU16}.
The quadratic guarantee is attainable but its optimality remained open:
\begin{question}
Can a different choice of a linear sketch and responder support substantially more adaptive queries at the same sketch dimension?
\end{question}
Prior attacks either required polynomially more queries, targeted particular sketching matrices or estimators, restricted the responder, or established weaker failure guarantees
\citep{HardtWoodruff13,BenEliezerJWY20,CohenNSS23,AhmadianCS25}.

\subsection{Our results: attacks}

We show that the quadratic guarantee is also a limit. We design an attack algorithm that applies to any matrix with $n>40k$ and forces errors even from powerful unrestricted responders.

\paragraph{Informed responder model.}
The responder has unbounded computation, memory, and private randomness. It knows $A$, the attack algorithm,
and the attacker's complete execution history and current state. Before each query, it knows the input distribution. It may change its randomized rule adaptively after every query. Its only information restriction is that it observes each fresh input through the sketch. After it returns a response, the
input itself and the rest of that query's execution are revealed. Note that a stronger responder model only strengthens the attack guarantee. \Cref{def:interaction} formalizes this model.

\paragraph{Norm gap task.}
Our attacks work even when the responder only has to distinguish
$\|v\|_2\le L$ from $\|v\|_2\ge U$, for fixed public thresholds
with $U/L=1+\Theta(\alpha)$; either binary answer is permitted
between the thresholds. This task, also considered by \citep{HardtWoodruff13,AhmadianCS25}, is easier than norm estimation with relative error $\Omega(\alpha)$: thresholding a sufficiently accurate norm estimate gives a correct gap answer. Thus errors on the norm-gap task transfer to errors in norm estimation.

\begin{theorem}[A universal attack]
\label{thm:norm-operational}
There is an attack algorithm and a positive absolute constant
$c_{\rm frac}$ such that, for every $A\in\R^{k\times n}$ with $n>40k$,
every $\alpha\in(0,1]$, and every informed responder, the attack uses
\[
  O\bigl(k+\alpha^2k^2\log(2+\alpha^2k)\bigr)
\]
queries and returns, with probability at least $0.9$, a round $B$ of
$\Theta(k)$ consecutive queries generated independently from one announced
distribution. For binary norm gap with $U/L=1+\Theta(\alpha)$, at least
$c_{\rm frac}|B|$ answers are incorrect; for norm estimation, at least
$c_{\rm frac}|B|$ estimates have relative error $\Omega(\alpha)$.

All queries are supported on a common set of $40k+1$ coordinates, which
may be prescribed arbitrarily before the interaction. The attack uses $O(1+\alpha^2k\log(2+\alpha^2k))$ adaptive rounds; working
memory is $\widetilde O(k+\alpha^2k^2)$ words and arithmetic cost is
$\widetilde O(k^2+\alpha^4k^4)$.
\end{theorem}

The linear support size is optimal up to a constant factor: if $A$ is such that every $k$ columns are linearly independent, every input supported on a known set of at most $k$ coordinates can be recovered exactly from its sketch, allowing exact norm answers that do not reveal anything about $A$.

\paragraph{Accuracy vs. robustness}
The attack is more effective when higher accuracy is required.
At constant accuracy, our bounds match known defenses  of $\tilde{\Omega}(k^2)$ up to logarithmic factors.  At high accuracy $\alpha=\Theta(k^{-1/2})$, the attack uses $O(k)$ queries
in $O(1)$ rounds.
At very high accuracy ($\alpha < C k^{-1/2}$), the initial distribution already forces a constant fraction of errors.

\begin{figure}[t]
\centering
\includegraphics[width=0.64\linewidth]{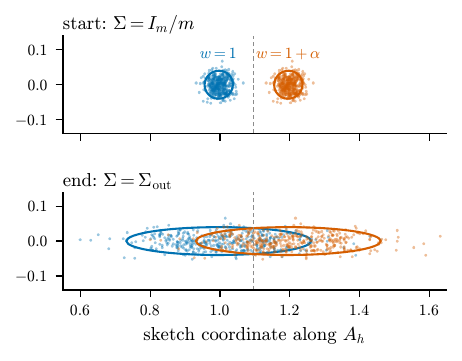}
\caption{The covariance attack in sketch space for queries with signal $w=1$ (blue) and $w=1+\alpha$
(orange). Sketches are projected on the signal column $\xcol$ and an orthogonal direction. Illustration shows $2\sigma$ ellipses of the sketch noise. The dashed line is the threshold $1+\alpha/2$.  At the start the
input noise is isotropic, and the sketch noise may separate the two signals (if there is overlap we are done).  At the
end, if the responder does not err earlier on an easy distribution,
the sketch noise is stretched along $\xcol$
and the two signals overlap. }
\label{fig:start-end}
\end{figure}

\paragraph{Attack structure.}
Each query combines a scalar signal with Gaussian noise.
The signal distribution remains fixed, while the attacker changes
the noise covariance between rounds. The attack keeps the noise norm
concentrated, so accurate norm answers distinguish low from high
signal values. Our analysis proceeds through this auxiliary
\emph{signal-gap} task.

The noise covariance determines how accurately the signal can be
estimated from the sketch. Higher signal-estimation variance
makes the two signal regimes harder to distinguish, even while the
input noise norm remains concentrated.
Starting from isotropic covariance, the attack seeks a
\emph{universally hard} query distribution\footnote{We call a
distribution \emph{universally hard} for a task if every response
rule has expected loss at least a fixed positive constant, even
knowing $A$ and the distribution.}
by changing the noise covariance (see \cref{fig:start-end}).

Each round also checks accuracy. If its empirical signal-gap loss
exceeds the validation threshold, the attack returns the round and
its query distribution. Otherwise, correlations between the answers
and the sampled noise guide a covariance update.
The analysis bounds the number of rounds that can pass validation,
forcing an observed round with many errors.
The returned distribution need not be universally hard, since the
responder can err on an easier distribution.
But if the responder maintains sufficiently small error whenever
a smaller constant error is achievable, the returned distribution
is universally hard.

\paragraph{Sufficient-statistic responders.}
For each fixed distribution, a scalar
\emph{sufficient statistic} contains all the information about the signal available in the sketch.
A \emph{sufficient-statistic responder} is restricted so that, conditional on the disclosed history, its mean binary answer depends on the current sketch only through this statistic (formalized in \Cref{sec:sufficient-model}). The class includes the Bayes-optimal signal-gap predictor, as well as predictors computed from the statistic followed by output noise, rounding, or clipping. It provides an idealized model of repeatedly fitting a predictor to the current query distribution while keeping the compressive map fixed. In this interpretation, accurate predictions supply feedback that guides the attacker's search for a hard noise distribution, even as the predictor is adapted to each new distribution.
\Cref{app:sufficient-readouts} develops the learning interpretation and its implications for output defenses.
We obtain a more efficient attack for informed responders with this additional restriction.

\begin{theorem}[Attack on sufficient-statistic responders; informal]
\label{thm:sufficient-informal}
There is an attack algorithm such that, for every
$A\in\R^{k\times n}$ with $n>40k$, every
$\alpha\in(0,1]$, and every sufficient-statistic norm responder, with probability at least $0.9$ the attack returns a round of $\Omega(k)$ consecutive queries, generated independently from one distribution, for which at least a constant fraction of norm estimates have relative error $\Omega(\alpha)$. The attack uses
$O(k+\alpha^2k^2)$ queries, $O(\log(2+\alpha^2k))$ adaptive rounds, $O(k\log(2+\alpha^2k))$ words of working memory,
and $\widetilde O(k^2+\alpha^2k^3)$ arithmetic operations.
All queries use any prescribed set of $40k+1$ coordinates.
On the same success event, if the threshold is reached
before a round is rejected, the final query distribution is universally hard for norm tasks.
\end{theorem}

\paragraph{Analyzing attack progress.}
\label{sec:techniques}
We establish the existence of a comparator covariance whose
signal-estimation variance makes the query distribution universally hard.
The key property exploited by both attacks is that accurate answers
force a component of the mean signed noise---the expected product
of the generated noise vector and the binary answer---along the
hidden signal-estimation direction; we extend the correlation
bound of \citet{AhmadianCS25} to general covariances.
For sufficient-statistic responders, the mean signed noise is aligned with the hidden direction, allowing direct multiplicative growth of the signal-estimation variance.
For general responders, the challenge is that other components
of the mean signed noise can obscure the useful direction.
We extract measurements that separate the current covariance
from the hard comparator and use them in matrix
multiplicative-weights updates, with progress measured by a
relative-entropy potential.

\begin{table}[t]
\centering
\footnotesize
\setlength{\tabcolsep}{2pt}
\renewcommand{\arraystretch}{1.2}
\begin{tabular}{@{}p{.13\linewidth}p{.135\linewidth}p{.225\linewidth}cccc@{}}
\toprule
Result & Responder & Guaranteed outcome & Queries & Rounds & Support
& Operations \\
\midrule
\raggedright Hardt--Woodruff \citep{HardtWoodruff13}
 & \raggedright Fixed within round
 & \raggedright Constant error rate (on queried fixed estimator)
 & $\mathrm{poly}(k)$ & $\widetilde O(k)$ & $O(k)$ & $\mathrm{poly}(k)$ \\[3pt]
\raggedright ACS \citep{AhmadianCS25}
 & \raggedright Informed
 & \raggedright Constant error rate OR failure of dist-optimal est
 & $\widetilde O(k^2)$ & $1$ & $\widetilde O(k^2)$ & $\widetilde O(k^4)$ \\[3pt]
\raggedright This work: general
 & \raggedright Informed
 & \raggedright Constant error rate (in round)
 & $O(k^2\log k)$ & $O(k\log k)$ & $\Theta(k)$ & $\widetilde O(k^4)$ \\[3pt]
\raggedright sufficient statistic
 & \raggedright Informed; scalar statistic
 & \raggedright Constant error rate (in round)
 & $O(k^2)$ & $O(\log k)$ & $\Theta(k)$ & $\widetilde O(k^3)$ \\
\midrule
\multicolumn{7}{@{}l}{\emph{Robustness benchmark:} AMS/JL with robustified estimator}\\[2pt]
\raggedright  \citep{HassidimKMMS20,Blanc23}
 & \raggedright Fixed
 & \raggedright Accurate answers
 & $\widetilde\Omega(k^2)$ & --- & --- & --- \\
\bottomrule
\end{tabular}
\caption{Comparison at constant accuracy and success probability. Sketch
 dimension $k$. Rounds are of attack adaptivity. Support is the size of a common
 coordinate set. Arithmetic-operation counts assume unit-cost sampling.}
\label{tab:comparison}
\end{table}

\subsection{Lower bound for informed responders}
Our attacks use $\tilde{O}(k+\alpha^2k^2)$ queries but known defenses \citep{HassidimKMMS20,Blanc23} guarantee only $\widetilde\Omega(\alpha^4k^2)$ adaptive queries, leaving a gap for higher accuracy. The distinction is that the known defenses use responders without knowledge of the query distributions. We show that
for a broad family of signal--noise query distributions, a Gaussian JL sketch can withstand $\Omega(\alpha^2k^2)$ adaptive queries when paired with a responder that is provided with the noise mean and covariance.

\begin{theorem}[Matching lower bound; informal]
\label{thm:covariance-lower-bound-informal}
Fix a target error fraction $\delta\in(0,1)$.
For a sufficiently large constant $C$ and $C/\sqrt{k}\le\alpha\le1$, consider attacks using the signal--noise query family specified in
\Cref{app:covariance-lower-bound}, with sufficiently concentrated noise energy.
For either norm task, any such attack that forces, with probability at least $0.9$, a $\delta$ fraction of errors in an observed interval of $\Omega(k)$ consecutive queries against every sketching matrix and every informed responder requires $\Omega(\alpha^2k^2)$ queries.
The constants may depend on $\delta$.

Subject to the family's moment and concentration conditions, the noise law is arbitrary and may depend on the signal; the distribution may change after every answer.
The lower bound is witnessed by a Gaussian JL matrix sampled once and kept fixed, together with a responder that beyond the matrix, current sketch, and task parameters, uses only the signal direction and noise mean and covariance. \end{theorem}

Taking $\delta$ to be the constant error fraction guaranteed by
our attacks, this establishes optimal query complexity for the
sufficient-statistic attack within this family and accuracy regime.
The general attack is optimal up to its logarithmic factor.

The responder adds noise to a distribution-informed linear estimator. For the analysis, we use KL divergence to bound what its answers reveal about the hidden Gaussian matrix, showing that finding a query law with large signal-estimation variance requires $\Omega(\alpha^2k^2)$ queries.
\Cref{app:covariance-lower-bound} gives the precise statement and
proof.

\subsection{Comparison with prior work}
\label{sec:comparison}
\Cref{tab:comparison} summarizes a comparison of our results with directly related prior work.
A broader related work discussion is included in \Cref{app:related-work}.

The problem was first formalized by
\citet*{HardtWoodruff13}. They gave polynomial-query attacks for arbitrary sketching matrices $A$,
using $\widetilde O(k)$ adaptive rounds and common query support of size $O(k)$ at constant relative accuracy. Their responder uses a randomized estimator
that may change between rounds, but remains fixed throughout each round.
Their attack returns a distribution
on which \emph{that estimator} has a constant error rate.  Our attack forces errors from a much stronger informed responder that additionally knows each
query distribution and may change its rule after every query.

\citet*{AhmadianCS25} (ACS) obtained a $\widetilde O(k^2)$-query attack
against arbitrary $A$ (with $n=\widetilde\Omega(k^2)$) and informed responders.
It samples all queries i.i.d.\ from one distribution in a single nonadaptive round. The attack guarantees that either the responder makes errors on a constant fraction of queries, or the
attack constructs an input on which the estimator optimal for the query
distribution fails (but the latter need not be an error by the actual
responder).  They left open the existence of a quadratic-query attack with a constant error rate and with support reduced from
$\widetilde O(k^2)$ to $\widetilde O(k)$ coordinates. As observed by \citep{AhmadianCS25}, at constant accuracy,
$\widetilde\Omega(k)$ adaptive rounds are necessary to force
errors against arbitrary sketches and responders using polynomially
many queries.\footnote{Indeed, a dimension-$k$ sketch can contain $\widetilde\Omega(k)$ independent JL or AMS blocks, each with $O(\log k)$ rows. Answering each round using a fresh block preserves the nonadaptive accuracy guarantee for all queries in that round.}

\paragraph{Organization.}
\Cref{sec:interaction-model} defines the interaction model and attack objective.
\Cref{sec:general-main} presents the general attack.
\Cref{sec:general-preview} previews the analysis.
\Cref{sec:sufficient-summary} gives the more efficient attack for
sufficient-statistic responders. \Cref{sec:conclusion} concludes.
\Cref{app:related-work} places our work in a broader context of related works.
The full analyses appear in
\Cref{sec:general-analysis,sec:sufficient-analysis}; the lower bound in \Cref{app:covariance-lower-bound}; subsequent appendices
give supporting proofs and details.

\section{Interaction model and attack objectives}
\label{sec:interaction-model}

Fix a matrix $A\in\R^{k\times n}$ and a publicly specified task with response space $\mathcal Y$. The task determines which answers are correct for each input.

The interaction is between an \emph{attacker} and a \emph{responder}. In query $i$, the attacker chooses an input $v_i\in\R^n$, the responder receives the \emph{sketch} $Av_i$ and returns an answer $a_i\in\mathcal Y$ that the attacker observes.

The attacker never observes $A$ or the sketches; it observes only its own execution and the returned answers from the responder.
Its objective is
to cause the responder to return incorrect answers.

\subsection{The informed responder model}
\label{sec:informed-interaction}

Our attacks are analyzed against an \emph{informed responder}, which has unbounded computation, memory, and randomness. It knows $A$ and the attacker's algorithm and its parameters, and the complete execution history.
The only information temporarily withheld from the responder at the beginning of query $i$ is the fresh randomness that the attacker uses to sample its current input $v_i$. This information is shared with the responder, but only after the answer $a_i$ is determined.

Let $\mathsf{Tr}_i$ denote the complete execution of both parties through the completion of query $i$, including their random choices, inputs, and responses. This transcript is available to the responder.
Our attack queries follow a particular structure:
\begin{definition}[Adaptive interaction query]
\label{def:interaction}

Query $i$ proceeds as follows.
\begin{enumerate}[label=(\roman*),leftmargin=1.7em]
\item The attacker announces a Borel distribution $\mathcal D_i$ on
$\R^n$, chosen measurably from its own execution transcript so far and fresh randomness.  It discloses $\mathcal D_i$ and the execution history of this query.
\item Using $A$, $\mathsf{Tr}_{i-1}$, and the newly disclosed information, the responder privately selects a measurable map
$\oracle_i:\R^k\to\mathcal Y$, possibly at random.
\item The attacker draws $v_i\sim\mathcal D_i$ using fresh randomness.
The responder receives the sketch $y_i=Av_i$ and returns an answer $a_i=\oracle_i(y_i)$, which the attacker observes. After the answer, the attacker discloses the remainder of its
execution for this query. \end{enumerate}
The attacker may stop and return an output based on its information and fresh randomness.
\end{definition}

\subsection{Attack structure and objective}
 Our attacks proceed in \emph{rounds}, each containing $\Omega(k)$ queries. At the start of each round, the attacker fixes a query distribution and the number of queries, then draws the queries independently from that distribution. Thus, conditional on these choices and the preceding history, the queries in a round are i.i.d. The responder may still change its rule after every query. We use $i$ for the global query index, $t$ for the round index, and $j$ for the query index within a round.

The attack's goal is to return a round, together with its query distribution, in which a constant fraction of the answers were incorrect for the task. We seek this guarantee with constant probability using few queries and rounds.

\subsection{Norm tasks}
\label{sec:attack-success}

\paragraph{Norm estimation.}
The responder returns an estimate
$\widehat r_i\in\mathcal Y=[0,\infty)$ of $\|v_i\|_2$.
For an accuracy parameter $\alpha>0$ and a fixed absolute
constant $c_{\rm err}>0$, the loss is
\begin{equation}\label{eq:relative-norm-loss}
  e_i=\ind{|\widehat r_i-\|v_i\|_2|>
       c_{\rm err}\alpha\|v_i\|_2}.
\end{equation}
Thus an answer is correct when its relative error is at most
$c_{\rm err}\alpha$.

\paragraph{Binary norm gap.}
Fix public thresholds $0<L<U$. The responder returns an answer
in $\mathcal Y=\{-1,+1\}$: the correct answer is $-1$ when
$\|v\|_2\le L$ and $+1$ when $\|v\|_2\ge U$, with either answer
permitted between the thresholds. The loss is
\begin{equation}\label{eq:norm-gap-loss}
 e_{\rm ng}(b,v)
 :=\ind{b=+1,\ \|v\|_2\le L}
   +\ind{b=-1,\ \|v\|_2\ge U}.
\end{equation}
Our thresholds satisfy $U/L=1+\Theta(\alpha)$.
When gap responses are derived from norm estimates by thresholding at $\sqrt{(L^2+U^2)/2}$, every incorrect
gap answer therefore certifies $\Omega(\alpha)$ relative error
in the norm estimate on the same input.

\begin{definition}[Realized loss and empirical accuracy]
\label{def:target-loss}
For a fixed task, let $e_i\in[0,1]$ denote the realized loss on
query $i$, a measurable function of its input and response.
For $\delta\in(0,1)$, the responder is
\emph{$\delta$-accurate} on a round with query index set $B$ if
\[
  \frac1{|B|}\sum_{i\in B}e_i\le\delta.
\]
\end{definition}

For both norm tasks, the loss is binary, so empirical accuracy means that at most a $\delta$ fraction of the answers in the round are incorrect.

\section{The covariance attack}
\label{sec:general-main}
\label{sec:model-main}
\label{sec:general-attack}
\label{sec:general-guarantee}
\label{sec:general-search-space}

This section presents the general attack, given in
\Cref{alg:general-responder-discovery}. \Cref{sec:general-preview}
previews its analysis; the full proof is in \Cref{sec:general-analysis}.
The interaction follows the disclosure
protocol of \Cref{def:interaction}. The algorithm uses up to $T_{\max}$ rounds with $r$ queries each.
In each round, the attacker draws queries from a fixed signal--noise
distribution and converts the norm-task responses into signal-gap answers
$b_j\in[-1,1]$. It accumulates the signed noise $S_t$ and the loss sum
$E_t$, using a fixed validation level $\delta_0$.
If $E_t>r\delta_0/2$, it returns the round; otherwise it uses $S_t$
to update the noise covariance, which is initially isotropic. The following
subsections specify the query law, response conversion, and covariance
update. The updates are parametrized by the cap $\lambda$.

The parameters have the following scales:
\[
 \begin{gathered}
 \lambda=\Theta(\alpha^2),\qquad r=\Theta(k),\\
 T_{\max}=\Theta\bigl(1+\alpha^2k\log(2+\alpha^2k)\bigr).
 \end{gathered}
\]

\paragraph{Accuracy regimes.}
Fix absolute constants $C_{\min}\ge1$ and sufficiently small
$\alpha_0\in(0,1]$. First replace $\alpha$ by
$\min\{\alpha,\alpha_0\}$; this changes only absolute constants in the
guarantees. If the resulting parameter is below $C_{\min}/\sqrt{k}$,
both attacks return their first isotropic round of $\Theta(k)$ queries
(\Cref{cor:isotropic-small-accuracy}). The pseudocode and
covariance-update analysis below concern the remaining regime
$\alpha\in [C_{\min}/\sqrt{k},\alpha_0]$.

\subsection{Query distribution}
\label{sec:gaussian-query-family}

Fix any prescribed set $S$ of $m+1\le n$ coordinates, where $m=40k$.  Choose the signal coordinate
$h$ uniformly from $S$ and let $M=S\setminus\{h\}$ carry the noise.
Both $h$ and $M$ remain fixed throughout the attack.  Query entries are $0$
outside $S$, so we may discard the unused columns of $A$ and, for the
analysis, relabel the coordinates as $n=m+1$ and $M=[n]\setminus\{h\}$.
\Cref{sec:random-signal-coordinate}
justifies this initialization.

We also fix a noise multiplier $c>0$ (a universal constant) and a signal law $W$ supported near
$1$, specified in \cref{def:signal-density}.  For any
$\Sig\succeq0$ with $\tr\Sig=1$, define the query law $\mathcal D_\Sig$ by
\begin{equation}\label{eq:query-model}
  w\sim W,\qquad u\sim N(0,\Sig),\qquad
  v=w e_h+c\,\iota_Mu ,
\end{equation}
where $w$ and $u$ are independent and
$\iota_M:\R^m\to\R^n$ embeds a vector on the coordinates $M$.
Writing $\xcol=A_{\bullet h}$ and $\Bmat=A_{\bullet M}$, the sketch takes
the form
\begin{equation}\label{eq:query-sketch}
  y=Av=w\xcol+c\Bmat u .
\end{equation}

After initialization, the attack changes only the covariance of the query distribution.
In round $t$ it announces $\mathcal D_{\Sig_t}$ and uses fresh,
independent draws of $(w,u)$ for every query in that round.

\subsection{The signal-gap task}
\label{sec:signal-gap-task}
The auxiliary task is to distinguish $w\le1$ from $w\ge1+\alpha$.
A response $b\in[-1,1]$ incurs loss
\begin{equation}\label{eq:signal-gap-loss-function}
  \loss(b,w)
  =\frac{1+b}{2}\ind{w\le1}
   +\frac{1-b}{2}\ind{w\ge1+\alpha}.
\end{equation}
The answers $-1$ and $+1$ correspond to the low and high regimes;
intermediate values hedge between them.  The gap $1<w<1+\alpha$ is
unlabeled and incurs no loss.  The use of $\mathcal Y = [-1,1]$ facilitates averaging of randomized and changing discrete response rules.

\begin{definition}[Signal distribution]\label{def:signal-density}
Put $\ell=10\alpha/c$.  The signal law $W=W_{\alpha,c}$ is sampled as
\[
 \begin{gathered}
 w=1-\ell+U_1+U_2,\\
 U_1\sim\operatorname{Unif}[0,\ell],\quad
 U_2\sim\operatorname{Unif}[0,\ell+\alpha],
 \end{gathered}
\]
with $U_1,U_2$ independent.
\end{definition}
Its density is trapezoidal: the two ramps occupy the labeled regimes,
and the plateau occupies the unlabeled gap.
\Cref{app:gain-proof} gives the density and explains its role
in extracting correlation from accurate answers.
We take $c>0$ to be a small absolute constant and $0<\alpha\le c/20$,
so the signal support stays at or above $1/2$.

\subsection{From norm answers to signal-gap answers} \label{sec:norm2signal}
The norm and the signal relate as follows:
\begin{equation}\label{eq:query-norm}
 \|v\|_2^2=w^2+c^2\|u\|_2^2
 =w^2+c^2+c^2(\|u\|_2^2-1).
\end{equation}
Since $\E\|u\|_2^2=\tr\Sig=1$, concentration of the noise norm makes
the squared query norm close to $w^2+c^2$.  We call a trace-one
covariance \emph{norm-admissible} at budget $\rho$ when
$\tr(\Sig^2)\le\rho$.

For fixed $c$, we use a sufficiently small constant multiple of
$\alpha^2/c^4$ as the admissibility budget.  For binary norm-gap queries,
fix the public thresholds
\begin{equation}\label{eq:norm-gap-thresholds}
 L=\sqrt{1+c^2+\alpha/2},\qquad
 U=\sqrt{1+c^2+3\alpha/2}
\end{equation}
and use the returned bit as the signal-gap response.  These thresholds
have $U/L=1+\Theta(\alpha)$.  For norm estimates $\widehat r$, use
\begin{equation}\label{eq:threshold-response}
 b:=
 \begin{cases}
  +1, & \widehat r^{\,2}\ge1+c^2+\alpha,\\
  -1, & \text{otherwise,}
 \end{cases}
\end{equation}
In either case, the signal-gap loss is at most the error indicator
for the chosen norm task plus the indicator of
$c^2|\|u\|_2^2-1|>\alpha/2$.
\Cref{sec:norm-reduction} proves this pointwise reduction and bounds
the concentration-failure probability.  Consequently, a round with large
signal-gap loss also has many norm-task errors, except for these rare
concentration failures.  The attacker computes the signal-gap loss from
its own $w$ and the returned or thresholded answer.

\subsection{Covariance constraints}

For a cap parameter $\lambda\in(1/m,1/2]$, the attack keeps the
covariance in
\begin{equation}\label{eq:capped-covariance-body}
 \mathcal K_\lambda
 :=\left\{\Sig:\tr\Sig=1,\quad
 \frac{1-\lambda}{m-1}I_m\preceq\Sig
 \right\}.
\end{equation}
The initial covariance $\Sig_0=I_m/m$ belongs to
$\mathcal K_\lambda$.
The eigenvalue floor ensures invertibility and, together with
$\tr\Sig=1$, implies the cap
\[
 \lambda_{\max}(\Sig)
 \le 1-(m-1)\frac{1-\lambda}{m-1}=\lambda.
\]
The same constraints give
\begin{equation}\label{eq:universal-spike-admissibility}
 \tr(\Sig^2)
 \le \lambda^2+\frac{(1-\lambda)^2}{m-1}
 \le \lambda^2+\frac1{m-1}.
\end{equation}
This follows since concentrating all eigenvalue mass above the floor
in one eigenvalue maximizes the sum of squares, giving one
eigenvalue $\lambda$ and the remaining $m-1$ eigenvalues
equal to $(1-\lambda)/(m-1)$.
Thus every $\Sig\in\mathcal K_\lambda$ is norm-admissible
with budget $\rho=\lambda^2+1/(m-1)$. The parameter choices make this
budget small enough for the norm reduction (\Cref{sec:end-to-end}).

\newcommand{\algcommentfont}[1]{\footnotesize\itshape\textcolor{mydarkblue}{#1}}
\SetCommentSty{algcommentfont}
\SetFuncSty{texttt}
\SetKwFunction{RunRound}{Round}
\SetKwProg{RoundProcedure}{Procedure}{:}{}

\begin{algorithm}[t]
\footnotesize
\SetKwInOut{Input}{Input}
\SetKwInOut{Output}{Output}
\DontPrintSemicolon
\caption{Covariance attack}
\label{alg:general-responder-discovery}
\Input{norm task; cap $\lambda$; round size $r$; validation level $\delta_0=1/20$;
round budget $T_{\max}$}
\Output{a rejected round and its covariance, or $\bot$}
$\Sig_0\leftarrow I_m/m$; $q_0\leftarrow(1-\lambda)/(m-1)$\;
\For{$t=0,\ldots,T_{\max}-1$}{
  $(S_t,E_t)\leftarrow\RunRound(\Sig_t,r)$\;
  \If(\tcp*[h]{check round accuracy}){$E_t>r\delta_0/2$}{\Return{round $t$ and $\Sig_t$}\;}
  $\widehat q_t\leftarrow S_t/\|S_t\|_2$
    \textup{(any unit vector if $S_t=0$)}\;
  Form $Z_t$ using \eqref{eq:matrix-cut-definition} \tcp{rank-one measurement}
  Update $\Sig_{t+1}$ by \eqref{eq:matrix-mw-unprojected}--\eqref{eq:matrix-mw-projection}
    with $\eta=1/2$\;
}
\Return{$\bot$}\;
\BlankLine
\RoundProcedure{\RunRound{$\Sig,r$}}{
  Announce $\mathcal D_\Sig$ (\S~\ref{sec:gaussian-query-family}) and round size $r$\;
  $R\leftarrow\Sig^{1/2}$; $S\leftarrow0$; $E\leftarrow0$\;
  \For(\tcp*[h]{queries}){$j=1,\ldots,r$}{
    Draw $w_j\sim W_{\alpha,c}$,
      $z_j\sim N(0,I_m)$ \;
    Query $v_j=w_je_h+c\,\iota_MRz_j$\tcp*{$v_j\sim\mathcal D_\Sig$ (\S~\ref{sec:gaussian-query-family})}
    Obtain $b_j$ \tcp*{signal-gap answer $b_j$ from norm task response (\S~\ref{sec:norm2signal})}
    $S\leftarrow S+b_jz_j$ \tcp*{accumulate signed noise}
    $E\leftarrow E+\loss(b_j,w_j)$
      \tcp*{accumulate signal-gap loss \eqref{eq:signal-gap-loss-function}}
  }
  \Return{$(S,E)$}\;
}
\BlankLine
\end{algorithm}

\subsection{Covariance update} \label{sec:cov-update}
Let $(S_t,E_t)=\texttt{Round}(\Sig_t,r)$ and $R_t=\Sig_t^{1/2}$.
If the round passes validation, the attacker forms
\begin{equation}\label{eq:matrix-cut-definition}
\begin{gathered}
 \widehat q_t=\frac{S_t}{\|S_t\|_2},\\
 a_t=R_t^{-\top}\widehat q_t,\qquad Z_t=q_0a_ta_t^\top,\\
 \text{where } q_0 = \frac{1-\lambda}{m-1}.
 \end{gathered}
\end{equation}
If $S_t=0$, choose any unit vector for $\widehat q_t$.

For positive definite trace-one matrices, define the relative entropy
\begin{equation}\label{eq:quantum-relative-entropy}
 D(P\|Q)=\tr\bigl(P(\log P-\log Q)\bigr).
\end{equation}
We use a matrix multiplicative-weights update
\citep{TsudaRW05,AroraKale16} with step size $\eta=1/2$,
followed by a relative-entropy projection:
\begin{align}
 \widehat\Sig_{t+1}
 &=\frac{\exp(\log\Sig_t+\eta Z_t)}
         {\tr\exp(\log\Sig_t+\eta Z_t)},
 \label{eq:matrix-mw-unprojected}\\
 \Sig_{t+1}
 &=\operatorname*{argmin}_{\Sig'\in\mathcal K_{\lambda}}
       D(\Sig'\|\widehat\Sig_{t+1}).
 \label{eq:matrix-mw-projection}
\end{align}
The exponential step \eqref{eq:matrix-mw-unprojected} increases weight according to $Z_t$. The projection \eqref{eq:matrix-mw-projection} ensures that the covariance remains in $\mathcal K_\lambda$. The projection is computed by rescaling the eigenvalues of $\widehat\Sig_{t+1}$ and clipping them from below at $q_0$, with the scale chosen so that they sum to one (\Cref{sec:general-computation}).

\section{Analysis preview}
\label{sec:general-preview}

We show that, with high constant probability, some round must fail
validation before the round budget is exhausted.

We first express the signal-estimation variance in terms of the query covariance $\Sig$. We then show that successful initialization guarantees the existence of a \emph{hard comparator covariance} $\Sig_\star$. We measure progress by the relative entropy $D(\Sig_\star\|\Sig_t)$ between this comparator and the current covariance. Accurate rounds produce sufficiently many potential-decreasing updates. Its bounded initial value and nonnegativity, together with a concentration argument, bound the number of rounds that can pass validation.

\subsection{The scalar channel}
\label{sec:scalar-channel-main}

Recall that the sketch has the form $y=w\xcol+c\Bmat u$, where
$u\sim N(0,\Sig)$ is independent of $w$. Put
\begin{equation}\label{eq:Ksigma}
 K_\Sig=\Bmat\Sig\Bmat^\top.
\end{equation}
Suppose $\xcol\in\im(K_\Sig)\setminus\{0\}$. Define the
\emph{normalized signal-estimation variance} and the scalar statistic by
\begin{equation}\label{eq:scalar-statistic}
 \begin{aligned}
  \Vinfo(\Sig)&=(\xcol^\top K_\Sig^\dagger\xcol)^{-1},\\
  T_\Sig(y)&=\Vinfo(\Sig)\xcol^\top K_\Sig^\dagger y,
 \end{aligned}
\end{equation}
where $\dagger$ denotes the Moore--Penrose pseudoinverse.
The statistic is sufficient for $w$: conditional on $T_\Sig(y)$,
the remaining sketch contains no additional information about the
signal. Moreover,
\begin{equation}\label{eq:estimator-noise}
 T_\Sig(y)=w+X,\qquad
 X\sim N(0,c^2\Vinfo(\Sig)),
\end{equation}
with $X$ independent of $w$.
Thus $c^2\Vinfo(\Sig)$ is the actual signal-estimation variance, and a larger value makes the two signal regimes harder to distinguish.
In particular, if $c^2\Vinfo(\Sig)\ge\alpha^2$, every responder incurs a constant expected signal-gap loss (\Cref{lem:large-variance-hardness}) and the distribution
is \emph{universally hard} for the signal-gap task.

If $\xcol=0$, then the sketch is independent of $w$ and $\Vinfo(\Sig)=+\infty$. If $\xcol\notin\im(K_\Sig)$, then the sketch contains a noiseless signal component and $\Vinfo(\Sig)=0$.

\subsection{Successful initialization}
With high constant probability, random-coordinate initialization (\Cref{sec:gaussian-query-family}) gives
either $\xcol=0$, in which case the sketch carries no signal information
and the isotropic query distribution is already hard for the signal-gap task, or $\xcol\in\im(\Bmat)\setminus\{0\}$ and
$\Vinfo(I_m/m)\ge1/(40k)$  (\Cref{cor:random}).

\subsection{Existence of a hard comparator}
We measure progress relative to a fixed
\emph{comparator covariance} $\Sig_\star\in\mathcal K_\lambda$,
constructed in the analysis using the hidden matrix $A$.
It has eigenvalue $\lambda$ along a suitable hidden direction
and eigenvalue $q_0$ in every orthogonal direction.
This covariance spike amplifies the isotropic variance by
$m\lambda=\Theta(\alpha^2k)$:
\[
  \Vinfo(\Sig_\star)=m\lambda\,\Vinfo(I_m/m).
\]
Together with the initialization bound, a sufficiently large
constant in $\lambda=\Theta(\alpha^2)$ makes
$c^2\Vinfo(\Sig_\star)$ a sufficiently large multiple of
$\alpha^2$, so the comparator is universally hard for the
signal-gap task. It is also norm-admissible because it belongs
to $\mathcal K_\lambda$
(\Cref{eq:universal-spike-admissibility,sec:end-to-end}).

\subsection{Accuracy forces correlation}
\label{sec:gain-main}
An accurate response must correlate with the signal-estimation error $T_\Sig(y)-w$. The following lemma extends the isotropic gain bound of \citet[Lemma~4.9]{AhmadianCS25} to general covariances;
its proof appears in \Cref{app:gain-proof}.
Write $\loss_\Sig(\oracle)=\E[\loss(\oracle(y),w)]$ for the expected
signal-gap loss under $\mathcal D_\Sig$.

\begin{restatable}[General-covariance gain lemma]{lemma}{gainlemma}
\label{lem:gain}
Let $\alpha>0$, $0<c\le1$, $\delta_0=1/20$, and $c_g=1/22$.
For the query law $\mathcal D_\Sig$ with signal law $W_{\alpha,c}$, let
$\Sig\succeq0$ satisfy $\tr\Sig=1$ and $0<\Vinfo(\Sig)<\infty$.
Every measurable response rule $\oracle:\R^k\to[-1,1]$ with
$\loss_\Sig(\oracle)\le\delta_0$ satisfies
\begin{equation}\gainlabel{eq:scalar-gain}
 \begin{aligned}
  \E[\oracle(y)(T_\Sig(y)-w)]&\ge c\Vinfo(\Sig)g,\\
  g&:=c_g\frac{c^2}{\alpha}.
 \end{aligned}
\end{equation}
\end{restatable}
Since the estimation error is a linear function of the generated
Gaussian noise, the gain bound guarantees a useful component of
the mean signed noise along the hidden signal-estimation direction.
The attacker accumulates this correlation through
$S_t=\sum_j b_jz_j$.

If the conditional expected losses are small on average over the
round, the useful component accumulates linearly with the number
of queries, while concentration controls the sampling fluctuations.
This remains true when the responder changes its rule after every
query: we analyze each contribution conditional on the preceding
history. Other components of the mean signed noise may also
accumulate, so $S_t$ need not approximate the hidden direction.
The next step shows that it nevertheless yields a separating
measurement.

\subsection{Correlation yields separating measurements}
The rank-one measurement $Z_t$ computed in
\eqref{eq:matrix-cut-definition} always satisfies
\[
  0\preceq Z_t\preceq I_m,
  \qquad
  \langle\Sig_t,Z_t\rangle=q_0,
\]
where $\langle P,Q\rangle=\tr(PQ)$.
We show that, conditional on the history before the round, with
probability bounded below by an absolute constant, either validation
rejects the round or
\[
  \langle\Sig_\star,Z_t\rangle\ge\kappa q_0
\]
for a sufficiently large absolute constant $\kappa$.
Since $Z_t=q_0a_ta_t^\top$, this means that the comparator assigns
at least $\kappa$ times as much noise variance as the current
covariance along the direction $a_t$.

On a high-probability validation event, passing the test implies
small average conditional loss, so the correlation guarantee
from the preceding subsection applies. A geometric bound converts
the useful component of the accumulated mean signed noise into
separation from the comparator, even when other components are
large. With $r=\Theta(k)$ queries, martingale bounds ensure that
sampling fluctuations preserve this separation with constant
probability. Thus the observed sum $S_t$ supplies the measurement
needed for the covariance update
(\Cref{sec:general-phase}).

\subsection{Entropy bounds the number of rounds}
The potential $D(\Sig_\star\|\Sig_t)$ is nonnegative and initially
at most $\lambda\log(m\lambda)$.
After an accepted round, a separating measurement makes the
covariance update decrease this potential by a sufficiently large
constant multiple of $q_0$. Any other update can increase it by
at most a smaller constant multiple of $q_0$.

The preceding separation guarantee, together with a concentration
bound, implies that with high constant probability, if all rounds
passed validation, a constant fraction of their measurements
would be separating. The constants ensure that their decreases
outweigh the possible increases from the remaining updates,
giving a total decrease of $\Omega(Tq_0)$ over $T$ rounds.
Since the potential cannot become negative, this rules out more than
\[
 O\!\left(1+\frac{\lambda}{q_0}\log(m\lambda)\right)
 =O\bigl(1+\alpha^2k\log(2+\alpha^2k)\bigr)
\]
accepted rounds. Thus some round must fail validation within the
stated round budget (\Cref{sec:general-potential}).

Note that this argument requires cumulative progress in relative entropy.
The signal-estimation variance may still decrease on individual updates, and the attack need not reach the comparator.

\section{Sufficient-statistic attack}
\label{sec:sufficient-summary}
\label{sec:sufficient-model}

We define the sufficient-statistic responder class using
$T_\Sig$ from \Cref{sec:scalar-channel-main}.
As in the general analysis, consider a successful initialization with
$\xcol\in\im(\Bmat)\setminus\{0\}$ and a positive definite covariance.
In the zero-signal case, the isotropic query distribution is already
hard for the signal-gap task, without any restriction on the responder.

Let $\mathcal H_{j-1}$ denote the disclosed information before query $j$
of round $t$, including the current query law but excluding the
responder's private execution.
The averaged response rule is
\[
 \bar\oracle_{t,j}(y)
 =\E[\oracle_{t,j}(y)\mid\mathcal H_{j-1}],
\]
where the expectation averages over private responder state and
randomization. The conditioning is formalized in
\Cref{sec:probability-conventions}.

\begin{definition}[Sufficient-statistic responder]
\label{def:phasewise-sufficient}
The responder is \emph{sufficient-statistic} if, in every round $t$
and query $j$,
\[
 \bar\oracle_{t,j}(y)
 =\varphi_{t,j}\bigl(T_{\Sig_t}(y)\bigr)
\]
for an $\mathcal H_{j-1}$-measurable function
$\varphi_{t,j}:\R\to[-1,1]$, almost surely under the current sketch law.
The responder may change its randomized rule after every query;
only the dependence of its conditional mean on the current sketch
is restricted.
\end{definition}

For norm estimation, this condition applies to the thresholded answer
in \eqref{eq:threshold-response}. For the binary norm-gap task, it
applies to the returned bit.

The class includes the Bayes-optimal signal-gap predictor
(\Cref{prop:sufficient-bayes}) and any responder that computes its
answer from $T_\Sig$ and fresh randomness.

\subsection{The attack}
\label{sec:sufficient-attack}
\label{sec:sufficient-guarantee}
\label{sec:sufficient-preview}

The attack, described in \Cref{alg:sufficient-statistic-discovery}, uses the same query family and signal-gap validation test
as \Cref{alg:general-responder-discovery}, starting at $\Sig_0=I_m/m$.
For $\alpha\in [C_{\min}/\sqrt{k},\alpha_0]$, the attack uses an
eigenvalue stopping threshold $\bar\rho$. It uses at most $T_{\max}$
covariance-update rounds, followed by one final validation round when
the top eigenvalue reaches or exceeds this threshold.
The parameters have scales
\[
 \begin{gathered}
 \bar\rho=\Theta(\alpha^2),\qquad
 T_{\max}=O(\log(2+\alpha^2k)),\\
 r_t=\Theta\bigl(k+\alpha^2k^2\vartheta^t\bigr)
 \; (0\le t<T_{\max}),\qquad
 r_{\rm fin}=\Theta(k),
 \end{gathered}
\]
where $\vartheta\in(0,1)$ is a fixed absolute constant.
The first round uses $\Theta(k+\alpha^2k^2)$ queries;
the geometrically decreasing term allows smaller later rounds,
while every round retains $\Omega(k)$ queries.
The total query budget is
\[
 \sum_{t=0}^{T_{\max}-1}r_t+r_{\rm fin}
 =O(k+\alpha^2k^2).
\]
The precise schedule is given in
\eqref{eq:sufficient-phase-length}, with its norm-task
specialization in \Cref{sec:sufficient-norm}.

The rounds are as in \Cref{alg:general-responder-discovery}, with $(S_t,E_t)=\texttt{Round}(\Sig_t,r_t)$ at round $t$.
If $E_t>r_t\delta_0/2$, the attack returns the round.
Otherwise it normalizes the signed sum to
$\widehat h_t=S_t/\|S_t\|_2$ and updates the covariance by
\begin{equation}\label{eq:natural-update}
 \mathcal N(\Sig,\widehat h)
 =\frac{\Sig+pp^\top}{1+\|p\|_2^2},
 \qquad p=\Sig^{1/2}\widehat h.
\end{equation}
If $S_t=0$, any unit vector may be used for $\widehat h_t$.
The vector $\widehat h_t$ is a direction in whitened coordinates;
$p$ maps it back to the noise coordinates.
The rank-one addition increases noise along $p$ before the denominator
renormalizes the covariance to trace one.
The update preserves positive definiteness and trace one.

After each update, the algorithm checks whether the largest
eigenvalue has reached the fixed threshold
$\bar\rho\in(1/m,1/2]$.
Initially, $\Sig_0=I_m/m$, so
$\lambda_{\max}(\Sig_0)=1/m<\bar\rho$.
Each update satisfies
$\lambda_{\max}(\Sig_{t+1})\le2\lambda_{\max}(\Sig_t)$:
the rank-one addition can at most double the largest eigenvalue,
and trace normalization can only reduce it.
When an update first reaches the threshold, the attack stops
updating and runs one final validation round at the resulting
covariance $\Sig_{\rm out}$. Since the preceding covariance was
below the threshold,
\begin{equation}\label{eq:sufficient-threshold-crossing}
 \bar\rho\le\lambda_{\max}(\Sig_{\rm out})<2\bar\rho.
\end{equation}

\begin{algorithm}[t]
\footnotesize
\SetKwInOut{Input}{Input}
\SetKwInOut{Output}{Output}
\DontPrintSemicolon
\caption{Sufficient-statistic attack}
\label{alg:sufficient-statistic-discovery}
\Input{norm task; threshold $\bar\rho$;
round lengths $r_0,\ldots,r_{T_{\max}-1},r_{\rm fin}$;
validation level $\delta_0=1/20$}
\Output{a rejected round and its covariance, or $\bot$}
$\Sig_0\leftarrow I_m/m$\;
\For{$t=0,\ldots,T_{\max}-1$}{
  $(S_t,E_t)\leftarrow\RunRound(\Sig_t,r_t)$ \tcp*{in \cref{alg:general-responder-discovery}}
  \If{$E_t>r_t\delta_0/2$}{
    \Return{round $t$ and $\Sig_t$}\;
  }
  $\widehat h_t\leftarrow S_t/\|S_t\|_2$
  \textup{(any unit vector if $S_t=0$)}\;
  $\Sig_{t+1}\leftarrow\mathcal N(\Sig_t,\widehat h_t)$ \tcp*{see \eqref{eq:natural-update}}
  \If(\tcp*[h]{threshold reached}){$\lambda_{\max}(\Sig_{t+1})\ge\bar\rho$}{
    $\Sig_{\rm out}\leftarrow\Sig_{t+1}$\;
    $(S,E)\leftarrow\RunRound(\Sig_{\rm out},r_{\rm fin})$ \tcp*{in \cref{alg:general-responder-discovery}}
    \If{$E>r_{\rm fin}\delta_0/2$}{
      \Return{the final round and $\Sig_{\rm out}$}\;
    }
    \Return{$\bot$}\;
  }
}
\Return{$\bot$}\;
\end{algorithm}

\subsection{Guarantee and analysis preview}

With probability at least $0.9$, the attack returns a round containing
a constant fraction of errors for the original norm task, as claimed
in \Cref{thm:sufficient-informal}.
On the same success event, if the eigenvalue threshold is reached,
the final distribution is hard for every responder for that norm task.

\paragraph{Alignment permits direction estimation.}
For a sufficient-statistic responder, the conditional mean signed
noise is parallel to the signal-estimation direction.
Gaussian noise orthogonal to this direction is independent of the
statistic, so its correlation with the answer vanishes.
Combining this alignment with the gain lemma, validation, and
concentration shows that normalizing $S_t$ estimates the direction
accurately enough for the covariance update
(\Cref{sec:sufficient-direction}).

\paragraph{Variance growth justifies the stopping test.}
On the success event, each covariance update increases
$\Vinfo(\Sig_t)$ by a fixed factor greater than one.
Controlling the accumulated direction-estimation errors also keeps
the variance within a constant factor of the largest possible
at the current top eigenvalue:
\[
 \Vinfo(\Sig_t)
 =\Theta\!\left(
   m\lambda_{\max}(\Sig_t)\Vinfo(I_m/m)
 \right).
\]
Before the threshold is reached, the upper bound limits variance
amplification to $O(m\bar\rho)$, allowing only
$O(\log(2+\alpha^2k))$ accepted updates.
At the threshold, the lower bound ensures that the variance is
large enough for universal signal-gap hardness.
Thus the observable eigenvalue test certifies hardness even though
the attacker cannot compute the signal-estimation variance
(\Cref{sec:sufficient-step}).

\paragraph{Query cost and norm admissibility.}
As the variance grows, the gain lemma guarantees a stronger
correlation, so fewer queries suffice to estimate the direction.
The resulting round schedule has total cost
$O(k+\alpha^2k^2)$, with the direction-estimation cost dominated
by the early rounds.
Control of the accumulated direction errors also gives
\[
 \tr(\Sig_t^2)\le 1/m+O(\bar\rho^2)
\]
throughout the successful execution, including the covariance
that crosses the threshold.
This maintains the norm admissibility needed for the norm-task
guarantees (\Cref{sec:sufficient-norm}).
\Cref{sec:sufficient-analysis} gives the full proof and an extension
to approximate alignment.

\FloatBarrier

\section{Conclusion}
\label{sec:conclusion}

At constant accuracy, we establish a near-quadratic limit on adaptive
reuse of any fixed linear norm sketch, even with an informed responder that changes its rule
after every query. The vulnerability that the attack exploits is that accurate answers expose useful correlations. We also present more efficient attacks for sufficient-statistic responses, which have aligned mean signed noise.

Our lower bound makes the accuracy dependence tight within a broad family of announced query distributions, using a responder that is given the signal direction and noise mean and covariance.
The main open question is where the optimal adaptive-query guarantee lies between $\alpha^4k^2$ and $\alpha^2k^2$, up to logarithmic factors, when query distributions are unrestricted or the responder lacks distributional information.
A further question is whether the logarithmic gap between our query bounds for general and sufficient-statistic responders is necessary.
\Cref{app:related-work} gives context and discusses implications for defenses.

\clearpage

\section*{Acknowledgments}

\paragraph{Edith Cohen:} Partially supported by the Israel Science Foundation (grant 1156/23).

\paragraph{Uri Stemmer:} Partially supported by the Israel Science Foundation (grant 1419/24), and the Blavatnik Research Foundation.

\section*{AI use statement}

This paper was developed from an initial proof sketch by the authors, through many iterations with the help of ChatGPT,
Claude, and Gemini.  The tools were used for formalizing statements, writing
proofs and text, improving proofs and bounds, and looking up relevant literature.  The authors are responsible for the content.

\renewcommand{\bibsection}{\section*{References}}
\bibliographystyle{plainnat}
\bibliography{highvariance_discovery}

@inproceedings{DworkFHPRR15,
  author    = {Cynthia Dwork and Vitaly Feldman and Moritz Hardt and Toniann Pitassi and Omer Reingold and Aaron Roth},
  title     = {Preserving statistical validity in adaptive data analysis},
  booktitle = {Proceedings of the 47th ACM Symposium on Theory of Computing (STOC)},
  pages     = {117--126},
  year      = {2015},
  url       = {https://doi.org/10.1145/2746539.2746580}
}

@inproceedings{BassilyNSSSU16,
  author    = {Raef Bassily and Kobbi Nissim and Adam Smith and Thomas Steinke and Uri Stemmer and Jonathan Ullman},
  title     = {Algorithmic stability for adaptive data analysis},
  booktitle = {Proceedings of the 48th ACM Symposium on Theory of Computing (STOC)},
  pages     = {1046--1059},
  year      = {2016},
  note      = {Full version in SIAM Journal on Computing, 50(3), 2021},
  url       = {https://doi.org/10.1145/2897518.2897566}
}

@inproceedings{HardtUllman14,
  author    = {Moritz Hardt and Jonathan Ullman},
  title     = {Preventing false discovery in interactive data analysis is hard},
  booktitle = {Proceedings of the 55th IEEE Symposium on Foundations of Computer Science (FOCS)},
  pages     = {454--463},
  year      = {2014},
  url       = {https://doi.org/10.1109/FOCS.2014.55}
}

@inproceedings{SteinkeUllman15,
  author    = {Thomas Steinke and Jonathan Ullman},
  title     = {Interactive fingerprinting codes and the hardness of preventing false discovery},
  booktitle = {Proceedings of the 28th Conference on Learning Theory (COLT)},
  series    = {PMLR},
  volume    = {40},
  pages     = {1588--1628},
  year      = {2015},
  url       = {https://proceedings.mlr.press/v40/Steinke15.html}
}

@inproceedings{Blanc23,
  author    = {Guy Blanc},
  title     = {Subsampling suffices for adaptive data analysis},
  booktitle = {Proceedings of the 55th ACM Symposium on Theory of Computing (STOC)},
  pages     = {999--1012},
  year      = {2023},
  url       = {https://doi.org/10.1145/3564246.3585226}
}

@inproceedings{HassidimKMMS20,
  author    = {Avinatan Hassidim and Haim Kaplan and Yishay Mansour and Yossi Matias and Uri Stemmer},
  title     = {Adversarially robust streaming algorithms via differential privacy},
  booktitle = {Advances in Neural Information Processing Systems 33 (NeurIPS)},
  pages     = {147--158},
  year      = {2020},
  note      = {Full version in Journal of the ACM, 69(6):42:1--42:14, 2022},
  url       = {https://arxiv.org/abs/2004.05975}
}

@inproceedings{BenEliezerJWY20,
  author    = {Omri Ben-Eliezer and Rajesh Jayaram and David P. Woodruff and Eylon Yogev},
  title     = {A framework for adversarially robust streaming algorithms},
  booktitle = {Proceedings of the 39th ACM Symposium on Principles of Database Systems (PODS)},
  pages     = {63--80},
  year      = {2020},
  note      = {Full version in Journal of the ACM, 69(2):17:1--17:33, 2022},
  url       = {https://arxiv.org/abs/2003.14265}
}

@inproceedings{WoodruffZhou21,
  author    = {David P. Woodruff and Samson Zhou},
  title     = {Tight bounds for adversarially robust streams and sliding windows via difference estimators},
  booktitle = {Proceedings of the 62nd IEEE Symposium on Foundations of Computer Science (FOCS)},
  pages     = {1183--1196},
  year      = {2021},
  url       = {https://arxiv.org/abs/2011.07471}
}

@inproceedings{AttiasCSS23,
  author    = {Idan Attias and Edith Cohen and Moshe Shechner and Uri Stemmer},
  title     = {A framework for adversarial streaming via differential privacy and difference estimators},
  booktitle = {14th Innovations in Theoretical Computer Science Conference (ITCS)},
  pages     = {8:1--8:19},
  year      = {2023},
  series    = {LIPIcs},
  volume    = {251},
  url       = {https://drops.dagstuhl.de/entities/document/10.4230/LIPIcs.ITCS.2023.8}
}

@inproceedings{CohenSS25,
  author    = {Edith Cohen and Mihir Singhal and Uri Stemmer},
  title     = {Breaking the quadratic barrier: Robust cardinality sketches for adaptive queries},
  booktitle = {Proceedings of the 42nd International Conference on Machine Learning (ICML)},
  series    = {PMLR},
  volume    = {267},
  pages     = {11185--11199},
  year      = {2025},
  url       = {https://proceedings.mlr.press/v267/cohen25c.html}
}

@inproceedings{MironovNS08,
  author    = {Ilya Mironov and Moni Naor and Gil Segev},
  title     = {Sketching in adversarial environments},
  booktitle = {Proceedings of the 40th ACM Symposium on Theory of Computing (STOC)},
  pages     = {651--660},
  year      = {2008},
  url       = {https://doi.org/10.1145/1374376.1374471}
}

@inproceedings{HardtWoodruff13,
  author    = {Moritz Hardt and David P. Woodruff},
  title     = {How robust are linear sketches to adaptive inputs?},
  booktitle = {Proceedings of the 45th ACM Symposium on Theory of Computing (STOC)},
  pages     = {121--130},
  year      = {2013},
  url       = {https://doi.org/10.1145/2488608.2488624}
}

@inproceedings{CherapanamjeriNelson20,
  author    = {Yeshwanth Cherapanamjeri and Jelani Nelson},
  title     = {On adaptive distance estimation},
  booktitle = {Advances in Neural Information Processing Systems 33 (NeurIPS)},
  pages     = {11178--11190},
  year      = {2020},
  url       = {https://arxiv.org/abs/2010.11252}
}

@inproceedings{CohenLNSSS22,
  author    = {Edith Cohen and Xin Lyu and Jelani Nelson and Tam{\'a}s Sarl{\'o}s and Moshe Shechner and Uri Stemmer},
  title     = {On the robustness of {CountSketch} to adaptive inputs},
  booktitle = {Proceedings of the 39th International Conference on Machine Learning (ICML)},
  series    = {PMLR},
  volume    = {162},
  pages     = {4112--4140},
  year      = {2022},
  url       = {https://proceedings.mlr.press/v162/cohen22a.html}
}

@article{CohenNSS23,
  author    = {Edith Cohen and Jelani Nelson and Tam{\'a}s Sarl{\'o}s and Uri Stemmer},
  title     = {Tricking the hashing trick: A tight lower bound on the robustness of {CountSketch} to adaptive inputs},
  journal   = {Proceedings of the AAAI Conference on Artificial Intelligence},
  volume    = {37},
  number    = {6},
  pages     = {7235--7243},
  year      = {2023},
  url       = {https://doi.org/10.1609/aaai.v37i6.25882}
}

@inproceedings{AhmadianCohen24,
  author    = {Sara Ahmadian and Edith Cohen},
  title     = {Unmasking vulnerabilities: Cardinality sketches under adaptive inputs},
  booktitle = {Proceedings of the 41st International Conference on Machine Learning (ICML)},
  series    = {PMLR},
  volume    = {235},
  pages     = {550--576},
  year      = {2024},
  url       = {https://proceedings.mlr.press/v235/ahmadian24a.html}
}

@inproceedings{GribelyukLWYZ24,
  author    = {Elena Gribelyuk and Honghao Lin and David P. Woodruff and Huacheng Yu and Samson Zhou},
  title     = {A strong separation for adversarially robust $\ell_0$ estimation for linear sketches},
  booktitle = {Proceedings of the 65th IEEE Symposium on Foundations of Computer Science (FOCS)},
  year      = {2024},
  pages     = {2318--2343},
  url       = {https://arxiv.org/abs/2409.16153}
}

@inproceedings{GribelyukLWYZ25,
  author    = {Elena Gribelyuk and Honghao Lin and David P. Woodruff and Huacheng Yu and Samson Zhou},
  title     = {Lifting linear sketches: Optimal bounds and adversarial robustness},
  booktitle = {Proceedings of the 57th ACM Symposium on Theory of Computing (STOC)},
  pages     = {395--406},
  year      = {2025},
  url       = {https://doi.org/10.1145/3717823.3718227}
}

@inproceedings{GribelyukLWYZ26a,
  author    = {Elena Gribelyuk and Honghao Lin and David P. Woodruff and Huacheng Yu and Samson Zhou},
  title     = {Adversarial Robustness on Insertion-Deletion Streams},
  booktitle = {Proceedings of the 58th ACM Symposium on Theory of Computing (STOC)},
  pages     = {2278--2289},
  year      = {2026},
  publisher = {ACM},
  url       = {https://arxiv.org/abs/2602.20854}
}

@inproceedings{GribelyukLWYZ26b,
  author    = {Elena Gribelyuk and Honghao Lin and David P. Woodruff and Huacheng Yu and Samson Zhou},
  title     = {Adversarial Robustness for Small Frequency Moments and a Weak Equivalence Theorem for Turnstile Streams},
  booktitle = {Proceedings of the IEEE Symposium on Foundations of Computer Science (FOCS)},
  year      = {2026},
  url       = {https://arxiv.org/abs/2607.06312}
}

@inproceedings{CohenNSSS26,
  author    = {Edith Cohen and Jelani Nelson and Tam{\'a}s Sarl{\'o}s and Mihir Singhal and Uri Stemmer},
  title     = {One attack to rule them all: Tight quadratic bounds for adaptive queries on cardinality sketches},
  booktitle = {Proceedings of the ACM-SIAM Symposium on Discrete Algorithms (SODA)},
  pages     = {3162--3200},
  year      = {2026},
  url       = {https://arxiv.org/abs/2411.06370}
}

@inproceedings{AhmadianCS25,
  author    = {Sara Ahmadian and Edith Cohen and Uri Stemmer},
  title     = {The cost of compression: Tight quadratic black-box attacks on sketches for $\ell_2$ norm estimation},
  booktitle = {Advances in Neural Information Processing Systems 38 (NeurIPS)},
  year      = {2025},
  note      = {arXiv:2507.16345},
  pages     = {25633--25659},
  url       = {https://proceedings.neurips.cc/paper_files/paper/2025/hash/249f73e01f0a2bb6c8d971b565f159a7-Abstract-Conference.html}
}

@incollection{JohnsonLindenstrauss84,
  author    = {William B. Johnson and Joram Lindenstrauss},
  title     = {Extensions of {L}ipschitz mappings into a {H}ilbert space},
  volume    = {26},
  pages     = {189--206},
  year      = {1984},
  booktitle = {Conference in Modern Analysis and Probability},
  series    = {Contemporary Mathematics},
  publisher = {American Mathematical Society},
  url       = {https://doi.org/10.1090/conm/026/737400}
}

@article{AlonMS99,
  author    = {Noga Alon and Yossi Matias and Mario Szegedy},
  title     = {The space complexity of approximating the frequency moments},
  journal   = {Journal of Computer and System Sciences},
  volume    = {58},
  number    = {1},
  pages     = {137--147},
  year      = {1999},
  url       = {https://doi.org/10.1006/jcss.1997.1545}
}

@book{LehmannCasella98,
  author    = {Erich L. Lehmann and George Casella},
  title     = {Theory of Point Estimation},
  edition   = {2nd},
  publisher = {Springer},
  year      = {1998},
  note      = {ISBN 978-0-387-98502-2}
}

@article{TsudaRW05,
  author    = {Koji Tsuda and Gunnar R{\"a}tsch and Manfred K. Warmuth},
  title     = {Matrix exponentiated gradient updates for on-line learning and {B}regman projection},
  journal   = {Journal of Machine Learning Research},
  volume    = {6},
  pages     = {995--1018},
  year      = {2005},
  url       = {https://jmlr.org/papers/v6/tsuda05a.html}
}

@article{WarmuthKuzmin08,
  author    = {Manfred K. Warmuth and Dima Kuzmin},
  title     = {Randomized online {PCA} algorithms with regret bounds that are logarithmic in the dimension},
  journal   = {Journal of Machine Learning Research},
  volume    = {9},
  pages     = {2287--2320},
  year      = {2008},
  url       = {https://jmlr.org/papers/v9/warmuth08a.html}
}

@article{AroraKale16,
  author    = {Sanjeev Arora and Satyen Kale},
  title     = {A combinatorial, primal-dual approach to semidefinite programs},
  journal   = {Journal of the ACM},
  volume    = {63},
  number    = {2},
  pages     = {12:1--12:35},
  year      = {2016},
  note      = {Preliminary version in STOC 2007, pages 227--236},
  doi       = {10.1145/2837020},
  url       = {https://doi.org/10.1145/2837020}
}

@inproceedings{SzegedyZSBEGF13,
  author    = {Christian Szegedy and Wojciech Zaremba and Ilya Sutskever and Joan Bruna and Dumitru Erhan and Ian Goodfellow and Rob Fergus},
  title     = {Intriguing properties of neural networks},
  year      = {2014},
  booktitle = {International Conference on Learning Representations (ICLR)},
  url       = {https://arxiv.org/abs/1312.6199}
}

@inproceedings{GoodfellowSS15,
  author    = {Ian J. Goodfellow and Jonathon Shlens and Christian Szegedy},
  title     = {Explaining and harnessing adversarial examples},
  booktitle = {International Conference on Learning Representations (ICLR)},
  year      = {2015},
  url       = {https://arxiv.org/abs/1412.6572}
}

@inproceedings{PapernotMGJCS17,
  author    = {Nicolas Papernot and Patrick McDaniel and Ian Goodfellow and Somesh Jha and Z. Berkay Celik and Ananthram Swami},
  title     = {Practical black-box attacks against machine learning},
  booktitle = {Proceedings of the 2017 ACM Asia Conference on Computer and Communications Security (ASIA CCS)},
  pages     = {506--519},
  year      = {2017},
  url       = {https://arxiv.org/abs/1602.02697}
}

@inproceedings{FawziMF16,
  author    = {Alhussein Fawzi and Seyed-Mohsen Moosavi-Dezfooli and Pascal Frossard},
  title     = {Robustness of classifiers: from adversarial to random noise},
  booktitle = {Advances in Neural Information Processing Systems 29 (NIPS)},
  pages     = {1632--1640},
  year      = {2016},
  url       = {https://arxiv.org/abs/1608.08967}
}

@inproceedings{MoosaviDezfooliFFF17,
  author    = {Seyed-Mohsen Moosavi-Dezfooli and Alhussein Fawzi and Omar Fawzi and Pascal Frossard},
  title     = {Universal adversarial perturbations},
  booktitle = {IEEE Conference on Computer Vision and Pattern Recognition (CVPR)},
  pages     = {86--94},
  year      = {2017},
  url       = {https://arxiv.org/abs/1610.08401}
}

@misc{TramerPGBM17,
  author    = {Florian Tram{\`e}r and Nicolas Papernot and Ian Goodfellow and Dan Boneh and Patrick McDaniel},
  title     = {The space of transferable adversarial examples},
  howpublished = {arXiv:1704.03453},
  year      = {2017},
  url       = {https://arxiv.org/abs/1704.03453}
}

@inproceedings{IlyasEAL18,
  author    = {Andrew Ilyas and Logan Engstrom and Anish Athalye and Jessy Lin},
  title     = {Black-box adversarial attacks with limited queries and information},
  booktitle = {Proceedings of the 35th International Conference on Machine Learning (ICML)},
  year      = {2018},
  series    = {PMLR},
  volume    = {80},
  pages     = {2137--2146},
  url       = {https://proceedings.mlr.press/v80/ilyas18a.html}
}

@inproceedings{ChenJW20,
  author    = {Jianbo Chen and Michael I. Jordan and Martin J. Wainwright},
  title     = {{HopSkipJumpAttack}: A query-efficient decision-based attack},
  booktitle = {IEEE Symposium on Security and Privacy (SP)},
  pages     = {1277--1294},
  year      = {2020},
  url       = {https://arxiv.org/abs/1904.02144}
}

@inproceedings{IlyasSTETM19,
  author    = {Andrew Ilyas and Shibani Santurkar and Dimitris Tsipras and Logan Engstrom and Brandon Tran and Aleksander Madry},
  title     = {Adversarial examples are not bugs, they are features},
  booktitle = {Advances in Neural Information Processing Systems 32 (NeurIPS)},
  year      = {2019},
  url       = {https://arxiv.org/abs/1905.02175}
}

@inproceedings{TsiprasSETM19,
  author    = {Dimitris Tsipras and Shibani Santurkar and Logan Engstrom and Alexander Turner and Aleksander Madry},
  title     = {Robustness may be at odds with accuracy},
  booktitle = {International Conference on Learning Representations (ICLR)},
  year      = {2019},
  url       = {https://arxiv.org/abs/1805.12152}
}

@inproceedings{DanielyShacham20,
  author    = {Amit Daniely and Hadas Shacham},
  title     = {Most {ReLU} networks suffer from $\ell^2$ adversarial perturbations},
  booktitle = {Advances in Neural Information Processing Systems 33 (NeurIPS)},
  year      = {2020},
  pages     = {6629--6636},
  url       = {https://papers.nips.cc/paper/2020/hash/497476fe61816251905e8baafdf54c23-Abstract.html}
}

@inproceedings{BubeckCGT21,
  author    = {S{\'e}bastien Bubeck and Yeshwanth Cherapanamjeri and Gauthier Gidel and R{\'e}mi Tachet des Combes},
  title     = {A single gradient step finds adversarial examples on random two-layers neural networks},
  booktitle = {Advances in Neural Information Processing Systems 34 (NeurIPS)},
  pages     = {10081--10091},
  year      = {2021},
  url       = {https://arxiv.org/abs/2104.03863}
}

@inproceedings{BartlettBC21,
  author    = {Peter L. Bartlett and S{\'e}bastien Bubeck and Yeshwanth Cherapanamjeri},
  title     = {Adversarial examples in multi-layer random {ReLU} networks},
  booktitle = {Advances in Neural Information Processing Systems 34 (NeurIPS)},
  pages     = {9241--9252},
  year      = {2021},
  url       = {https://arxiv.org/abs/2106.12611}
}

@inproceedings{LecuyerAGHJ19,
  author    = {Mathias Lecuyer and Vaggelis Atlidakis and Roxana Geambasu and Daniel Hsu and Suman Jana},
  title     = {Certified robustness to adversarial examples with differential privacy},
  booktitle = {IEEE Symposium on Security and Privacy (SP)},
  year      = {2019},
  pages     = {656--672},
  url       = {https://arxiv.org/abs/1802.03471}
}

@inproceedings{CohenRK19,
  author    = {Jeremy M. Cohen and Elan Rosenfeld and J. Zico Kolter},
  title     = {Certified adversarial robustness via randomized smoothing},
  booktitle = {Proceedings of the 36th International Conference on Machine Learning (ICML)},
  year      = {2019},
  series    = {PMLR},
  volume    = {97},
  pages     = {1310--1320},
  url       = {https://proceedings.mlr.press/v97/cohen19c.html}
}

@inproceedings{BubeckSellke21,
  author    = {S{\'e}bastien Bubeck and Mark Sellke},
  title     = {A Universal Law of Robustness via Isoperimetry},
  booktitle = {Advances in Neural Information Processing Systems 34 (NeurIPS)},
  year      = {2021},
  pages     = {28811--28822},
  url       = {https://arxiv.org/abs/2105.12806}
}

@article{Freedman83,
  author    = {David A. Freedman},
  title     = {A Note on Screening Regression Equations},
  journal   = {The American Statistician},
  volume    = {37},
  number    = {2},
  pages     = {152--155},
  year      = {1983},
  doi       = {10.1080/00031305.1983.10482729},
  url       = {https://doi.org/10.1080/00031305.1983.10482729}
}

@article{Ioannidis05,
  author    = {John P. A. Ioannidis},
  title     = {Why Most Published Research Findings Are False},
  journal   = {PLoS Medicine},
  volume    = {2},
  number    = {8},
  pages     = {e124},
  year      = {2005},
  doi       = {10.1371/journal.pmed.0020124},
  url       = {https://doi.org/10.1371/journal.pmed.0020124}
}

@article{LukacsBA10,
  author    = {Paul M. Lukacs and Kenneth P. Burnham and David R. Anderson},
  title     = {Model Selection Bias and {Freedman's} Paradox},
  journal   = {Annals of the Institute of Statistical Mathematics},
  volume    = {62},
  number    = {1},
  pages     = {117--125},
  year      = {2010},
  doi       = {10.1007/s10463-009-0234-4},
  url       = {https://doi.org/10.1007/s10463-009-0234-4}
}

@article{BenDavidBKTW94,
  author    = {Shai Ben-David and Allan Borodin and Richard M. Karp and G{\'a}bor Tardos and Avi Wigderson},
  title     = {On the Power of Randomization in On-Line Algorithms},
  journal   = {Algorithmica},
  volume    = {11},
  number    = {1},
  pages     = {2--14},
  year      = {1994},
  doi       = {10.1007/BF01294260},
  url       = {https://doi.org/10.1007/BF01294260}
}

@article{ShiloachEven81,
  author    = {Shimon Even and Yossi Shiloach},
  title     = {An On-Line Edge-Deletion Problem},
  journal   = {Journal of the ACM},
  volume    = {28},
  number    = {1},
  pages     = {1--4},
  year      = {1981},
  doi       = {10.1145/322234.322235},
  url       = {https://doi.org/10.1145/322234.322235}
}

@inproceedings{AhnGM12,
  author    = {Kook Jin Ahn and Sudipto Guha and Andrew McGregor},
  title     = {Analyzing Graph Structure via Linear Measurements},
  booktitle = {Proceedings of the 23rd Annual ACM-SIAM Symposium on Discrete Algorithms (SODA)},
  pages     = {459--467},
  year      = {2012},
  doi       = {10.1137/1.9781611973099.40},
  url       = {https://doi.org/10.1137/1.9781611973099.40}
}

@inproceedings{GawrychowskiMW20,
  author    = {Pawe{\l} Gawrychowski and Shay Mozes and Oren Weimann},
  title     = {Minimum Cut in {$O(m\log^2 n)$} Time},
  booktitle = {47th International Colloquium on Automata, Languages, and Programming (ICALP)},
  series    = {LIPIcs},
  volume    = {168},
  pages     = {57:1--57:15},
  year      = {2020},
  url       = {https://doi.org/10.4230/LIPIcs.ICALP.2020.57}
}

@inproceedings{GutenbergWN20,
  author    = {Maximilian Probst Gutenberg and Christian Wulff-Nilsen},
  title     = {Decremental {SSSP} in Weighted Digraphs: Faster and Against an Adaptive Adversary},
  booktitle = {Proceedings of the 31st Annual ACM-SIAM Symposium on Discrete Algorithms (SODA)},
  pages     = {2542--2561},
  year      = {2020},
  url       = {https://arxiv.org/abs/2001.10821}
}

@inproceedings{Wajc20,
  author    = {David Wajc},
  title     = {Rounding Dynamic Matchings Against an Adaptive Adversary},
  booktitle = {Proceedings of the 52nd Annual ACM Symposium on Theory of Computing (STOC)},
  year      = {2020},
  doi       = {10.1145/3357713.3384258},
  pages     = {194--207},
  url       = {https://doi.org/10.1145/3357713.3384258}
}

@inproceedings{BeimelKMNSS22,
  author    = {Amos Beimel and Haim Kaplan and Yishay Mansour and Kobbi Nissim and Thatchaphol Saranurak and Uri Stemmer},
  title     = {Dynamic Algorithms Against an Adaptive Adversary: Generic Constructions and Lower Bounds},
  booktitle = {Proceedings of the 54th Annual ACM Symposium on Theory of Computing (STOC)},
  pages     = {1671--1684},
  year      = {2022},
  doi       = {10.1145/3519935.3520064},
  url       = {https://doi.org/10.1145/3519935.3520064}
}

@inproceedings{AthalyeEIK18,
  author    = {Anish Athalye and Logan Engstrom and Andrew Ilyas and Kevin Kwok},
  title     = {Synthesizing Robust Adversarial Examples},
  booktitle = {Proceedings of the 35th International Conference on Machine Learning (ICML)},
  series    = {Proceedings of Machine Learning Research},
  volume    = {80},
  pages     = {284--293},
  year      = {2018},
  url       = {https://proceedings.mlr.press/v80/athalye18b.html}
}

@inproceedings{BenEliezerEO21,
  author       = {Omri Ben-Eliezer and Talya Eden and Krzysztof Onak},
  title        = {Adversarially Robust Streaming via Dense-Sparse Trade-offs},
  year         = {2022},
  booktitle = {Symposium on Simplicity in Algorithms (SOSA)},
  pages     = {214--227},
  url          = {https://doi.org/10.1137/1.9781611977066.15}
}

@inproceedings{CohenGNS26,
  author    = {Edith Cohen and Elena Gribelyuk and Jelani Nelson and Uri Stemmer},
  title     = {Adaptively Robust Resettable Streaming},
  booktitle = {Proceedings of the 43rd International Conference on Machine Learning (ICML)},
  series    = {Proceedings of Machine Learning Research},
  volume    = {306},
  pages     = {21046--21087},
  year      = {2026},
  publisher = {PMLR},
  url       = {https://proceedings.mlr.press/v306/cohen26a.html}
}

@misc{CohenGMS26,
  author       = {Edith Cohen and Elena Gribelyuk and Pasin Manurangsi and Uri Stemmer},
  title        = {Tight Bounds on the Cost of Adaptivity for the {Meyerson} Sketch},
  howpublished = {arXiv:2609.06669},
  year         = {2026},
  url          = {https://arxiv.org/abs/2609.06669}
}

@misc{NasrEtAl25,
  author       = {Milad Nasr and Nicholas Carlini and Chawin Sitawarin and
                  Sander V. Schulhoff and Jamie Hayes and Michael Ilie and
                  Juliette Pluto and Shuang Song and Harsh Chaudhari and
                  Ilia Shumailov and Abhradeep Thakurta and Kai Yuanqing Xiao and
                  Andreas Terzis and Florian Tram{\`e}r},
  title        = {The Attacker Moves Second: Stronger Adaptive Attacks Bypass
                  Defenses Against {LLM} Jailbreaks and Prompt Injections},
  howpublished = {arXiv:2510.09023},
  year         = {2025},
  url          = {https://arxiv.org/abs/2510.09023}
}

@inproceedings{HeZRS16,
  author    = {Kaiming He and Xiangyu Zhang and Shaoqing Ren and Jian Sun},
  title     = {Deep Residual Learning for Image Recognition},
  booktitle = {Proceedings of the IEEE Conference on Computer Vision and Pattern Recognition (CVPR)},
  pages     = {770--778},
  year      = {2016},
  url       = {https://openaccess.thecvf.com/content_cvpr_2016/html/He_Deep_Residual_Learning_CVPR_2016_paper.html}
}

@inproceedings{HuEtAl21LoRA,
  author = {Edward J. Hu and Yelong Shen and Phillip Wallis and
            Zeyuan Allen-Zhu and Yuanzhi Li and Shean Wang and
            Lu Wang and Weizhu Chen},
  title = {{LoRA}: Low-Rank Adaptation of Large Language Models},
  year = {2022},
  booktitle = {International Conference on Learning Representations (ICLR)},
  url = {https://openreview.net/forum?id=nZeVKeeFYf9}
}

@article{Hager89,
  author  = {William W. Hager},
  title   = {Updating the Inverse of a Matrix},
  journal = {SIAM Review},
  volume  = {31},
  number  = {2},
  pages   = {221--239},
  year    = {1989},
  doi     = {10.1137/1031049},
  url     = {https://doi.org/10.1137/1031049}
}

\clearpage
\appendix

\section{Related work and discussion}
\label{app:related-work}
\label{sec:discussion}

\subsection{Related work}

\paragraph{Adaptive inputs and robustness.}
Classical guarantees for randomized algorithms and structures often
assume that inputs are independent of the algorithm's random choices.
This assumption can fail when outputs influence subsequent inputs,
whether through an adversarial attack or ordinary use within a
larger computation. Repeated interaction may reveal enough about
the hidden randomness to identify inputs on which the algorithm
fails.

This issue has been studied across statistics and adaptive data
analysis
\citep{Freedman83,Ioannidis05,LukacsBA10,HardtUllman14,
DworkFHPRR15,BassilyNSSSU16,SteinkeUllman15},
randomized online and dynamic algorithms
\citep{BenDavidBKTW94,ShiloachEven81,AhnGM12,GawrychowskiMW20,
GutenbergWN20,Wajc20,BeimelKMNSS22,CohenGMS26},
sketching and streaming
\citep{MironovNS08,HardtWoodruff13,BenEliezerJWY20,HassidimKMMS20,
WoodruffZhou21,AttiasCSS23,BenEliezerEO21},
and machine learning
\citep{SzegedyZSBEGF13,GoodfellowSS15,PapernotMGJCS17,
AthalyeEIK18,NasrEtAl25}.
Sketching provides a clean setting for quantifying how much
interaction is needed to compromise a fixed randomized representation.

\paragraph{Defenses and the quadratic guarantee.}
\label{app:quadratic-benchmark}
A simple defense uses independent copies of a randomized structure,
answering each query with a fresh copy. This supports a number of
adaptive queries roughly linear in the number of copies.
Methods developed in adaptive data analysis improve this guarantee
to quadratic by limiting what the answers reveal about the reused
randomness
\citep{DworkFHPRR15,BassilyNSSSU16,HassidimKMMS20,Blanc23}.

For JL or AMS sketches of total dimension $k$, these methods support
$\widetilde\Omega(\alpha^4k^2)$ adaptive queries.
Each independent copy uses $\Theta(\alpha^{-2})$ measurements,
leaving $\Theta(\alpha^2k)$ copies for quadratic reuse.
The methods use a fixed randomized response rule with fresh randomness suffices.

\paragraph{Quadratic vulnerabilities of sketches.}
Attacks on sketches differ in the structures they cover, the
responders they allow, and the failures they produce
\citep{HardtWoodruff13,CherapanamjeriNelson20,BenEliezerJWY20,
CohenLNSSS22,GribelyukLWYZ24,GribelyukLWYZ25}.
An attack on a standard estimator may be much cheaper than on one designed to obscure internal parameters. 

A line of work establishes quadratic-query vulnerabilities for
AMS and CountSketch, MinHash cardinality sketches, monotone
composable and linear cardinality sketches, and general linear
sketches for norm estimation
\citep{CohenNSS23,AhmadianCohen24,CohenNSSS26,AhmadianCS25}.
These results show that the quadratic guarantee is often a barrier,
even when the responder can choose how to use the sketch.
Our work establishes this vulnerability for Euclidean-norm estimation, for any linear sketch and with unrestricted responders.
\Cref{sec:comparison,tab:comparison} give the detailed comparison
with the closest prior attacks.

\paragraph{Beyond the quadratic barrier.}
Stronger guarantees are possible when the representation or the
allowed interactions change.
Recent algorithms for insertion--deletion and resettable streams
obtain stronger robustness by maintaining and refreshing
representations as the stream evolves
\citep{GribelyukLWYZ26a,GribelyukLWYZ26b,CohenGNS26}.
They process individual stream updates and maintain state tied to
the execution of a single stream. These guarantees do not provide
a common fixed, composable representation for independently
maintained inputs.
A different approach retains a fixed cardinality sketch but limits
how often each element participates in the queries: the quadratic
restriction then applies to each element's participation rather
than to the total number of queries \citep{CohenSS25}.
These results highlight the importance of both the representation
and the way it is reused.

\subsection{Implications for prediction and defenses}
\label{app:sufficient-readouts}

\paragraph{Connections to adversarial prediction.}
A predictor can be accurate on independently drawn inputs while
interaction with its outputs helps locate failures.
This gap has motivated extensive empirical and theoretical work
on adversarial examples
\citep{FawziMF16,MoosaviDezfooliFFF17,TramerPGBM17,
DanielyShacham20,BubeckCGT21,BartlettBC21,BubeckSellke21,
IlyasSTETM19,TsiprasSETM19}.
Our attacks share a simple mechanism with black-box attacks on
learned models: responses are used to weight sampled perturbations,
and their average identifies directions useful to the attack
\citep{PapernotMGJCS17,IlyasEAL18,ChenJW20}.
The sketching setting lets us quantify this mechanism and the
extent to which a responder can delay it.

\paragraph{Accuracy and disclosure.}
More accurate answers enable cheaper attacks, reaching $O(k)$
queries at $\alpha=\Theta(k^{-1/2})$.
Our lower bound also shows a benefit of deliberately
limiting response precision.
The added noise masks the matrix-dependent estimation error while
preserving the required accuracy.
Randomization also provides provable protection against bounded
input perturbations in learned predictors
\citep{LecuyerAGHJ19,CohenRK19}, making its role in controlling
adaptive disclosure a useful direction for further study.

\paragraph{Learning after a fixed compression.}
Consider a predictor built on a compressive linear layer, as used
in neural architectures \citep{HeZRS16}.
Suppose the compression remains fixed while the downstream
predictor is repeatedly fitted to the current input distribution.
This is an idealized setting for adaptation through limited
trainable components, including LoRA when applied downstream
of the fixed compression \citep{HuEtAl21LoRA}.

For our Gaussian query family, the Bayes-optimal signal-gap
predictor depends only on the scalar statistic $T_\Sig$
(\Cref{prop:sufficient-bayes}).
This motivates sufficient-statistic responders as a model of
prediction fitted to each distribution.
Postprocessing that uses only this prediction and fresh randomness,
including noise, rounding, or clipping, preserves the restriction.
The more efficient attack therefore applies to these output defenses.

The interaction can be viewed as a search for a distribution on
which the fixed representation loses essential information.
Accurate predictions guide this search, even as the predictor
is adapted to each new distribution.
If the responder remains sufficiently accurate whenever a
stricter constant tolerance is feasible, the returned distribution
is universally hard (\Cref{rem:forced-error}).
For the sufficient-statistic attack, reaching the eigenvalue
threshold also certifies this hardness.

\paragraph{Responses and representations.}
Our results suggest two complementary approaches to robustness
when a fixed compressive component forms an information bottleneck
in a larger model. One is to limit what the responses reveal about
that component. Fresh output noise can mask information useful
to the attacker; dependence on additional sketch information may
introduce ``decoy'' correlations that obscure the useful direction.
Our results quantify a benefit of the former and leave open how
much additional protection the latter can provide.

Another approach is to periodically replace the compressive
component while preserving predictive accuracy on the intended
input distribution. Changing which information the representation
retains could make directions learned from earlier responses less
useful to the attacker. This suggests studying defenses that
refresh the representation itself, alongside defenses that modify
the responses produced from a fixed representation.

\subsection{Open problems}
\label{app:open-problems}

\paragraph{Accuracy and distributional information.}
Our attacks use $\widetilde O(k+\alpha^2k^2)$ queries.
For signal--noise attacks with sufficiently concentrated noise
energy, we show that a responder given the noise mean and
covariance can withstand $\Omega(\alpha^2k^2)$ queries.
Known defenses without distributional information guarantee
$\widetilde\Omega(\alpha^4k^2)$ queries.
How much of this gap can be closed for unrestricted query
distributions, or without distributional information?

\paragraph{More efficient general attacks.}
Our sufficient-statistic attack saves a logarithmic factor in
queries and uses less memory and computation than our general
attack. Can general attacks match its query bound and achieve
comparable or better memory and computational costs?
General responders can use sketch information irrelevant to
signal prediction to create correlations that obscure the
useful direction. Do these correlations inherently increase
the cost of attack?

The gap in adaptive rounds is necessary: at constant accuracy,
general responders can require $\widetilde\Omega(k)$ rounds
for attacks using polynomially many queries
(\Cref{sec:comparison}).

\paragraph{Attacks on learned responders.}
When does fitting a signal-gap predictor to each query
distribution, while keeping the compression fixed, leave it
vulnerable to the more efficient attack?
In particular, how does excess prediction risk control
response correlations orthogonal to the signal-estimation
direction, and when does this imply the approximate alignment
required in \Cref{rem:sufficient-computation}?

\paragraph{Other tasks.}
Can attacks with similarly sharp query bounds be obtained
for other tasks, such as estimating other $\ell_p$ norms?
More broadly, can response correlations guide attacks that
expose information lost by nonlinear representations?

\section{Analysis of the general attack}
\label{sec:general-analysis}

We prove \Cref{thm:norm-operational}, then establish the supporting
results used in the proof.

\subsection{Parameters and proof of \texorpdfstring{\Cref{thm:norm-operational}}{Theorem 1.1}}
\label{sec:end-to-end}

\paragraph{Parameter choice.}
The following parameter choices suffice for the covariance-update
regime; we make no attempt to optimize the constants.
Fix a small absolute constant $0<c\le1$ and set $\delta=0.04$.
Use $\varepsilon_{\rm nc}=\delta_0/8$ and
$c_\rho=\varepsilon_{\rm nc}/8=\delta_0/64$ in the norm reduction
of \Cref{lem:norm-to-gap}. Choose $c_{\rm err}\le1/16$ small enough
for that lemma. Choose any prescribed set $S\subseteq[n]$ of
$m+1=40k+1$ coordinates, choose $h$ uniformly from $S$, and put
$M=S\setminus\{h\}$.

Set
\[
 C_\lambda=\frac{C_{\rm marg}}{c_g^2c^4},\qquad
 C_{\rm marg}=256,\qquad
 \lambda=C_\lambda\alpha^2,\qquad
 q_0=\frac{1-\lambda}{m-1}.
\]
The constant $C_{\rm marg}$ is the separation margin established in
\Cref{sec:general-phase}. Choose absolute constants $C_{\min}\ge1$
and $\alpha_0>0$ so that
$C_{\min}/\sqrt k\le\alpha\le\alpha_0$ implies
\begin{equation}\label{eq:alpha-range-conditions}
 \alpha\le\frac{c}{20},\qquad
 \frac1m<\lambda\le\frac12,\qquad
 \lambda^2+\frac1{m-1}\le c_\rho\frac{\alpha^2}{c^4}.
\end{equation}
Specifically, take
\[
 \begin{aligned}
 \alpha_0^2&=\min\left\{\frac{c^2}{400},\frac1{2C_\lambda},
                            \frac{c_\rho}{2c^4C_\lambda^2}\right\},\\
 C_{\min}^2&=\max\left\{1,\frac1{20C_\lambda},
                              \frac{c^4}{19c_\rho}\right\}.
 \end{aligned}
\]
The lower endpoint controls the isotropic contribution to the noise
norm fluctuations; the upper endpoint controls the covariance spike.

Use the following step size, round budget, and common round length:
\begin{equation}\label{eq:general-parameters}
 \begin{aligned}
 \eta&=\frac12,\qquad
 T_{\max}=\left\lceil\frac{\lambda}{q_0}\log(m\lambda)
                  +16\log\frac2\delta\right\rceil+1,\\
 r&=\max\left\{48m,\
       \left\lceil\frac{32}{\delta_0^2}
           \log\frac{4T_{\max}}{\delta}\right\rceil\right\}.
 \end{aligned}
\end{equation}

The same formulas are used with general $0<\delta<1$ in
\Cref{sec:general-potential}. By
\eqref{eq:universal-spike-admissibility} and
\eqref{eq:alpha-range-conditions}, every covariance in the search
space satisfies
\begin{equation}\label{eq:K-admissible}
 \tr(\Sig^2)\le c_\rho\frac{\alpha^2}{c^4},\qquad
 \|\Sig\|_{\rm op}\le\lambda=O(\alpha^2).
\end{equation}
\begin{proof}[Proof of \Cref{thm:norm-operational}]
First consider $C_{\min}/\sqrt{k}\le\alpha\le\alpha_0$.
Convert the norm-task answers as in \Cref{sec:norm2signal}.
The resulting answers form an informed signal-gap responder.

With probability at least $0.95$ over the choice of signal
coordinate, either $\xcol=0$ or
$1/(40k)\le\Vinfo(I_m/m)<\infty$
(\Cref{cor:random}). In the latter case, the spike comparator of
\eqref{eq:universal-spike-family} satisfies
\begin{equation}\label{eq:general-alpha-margin}
 \Vinfo(\Sig_\star)=m\lambda\Vinfo(I_m/m)\ge\lambda,
 \qquad g^2\Vinfo(\Sig_\star)\ge C_{\rm marg}=256.
\end{equation}
This is the margin needed for the entropy argument:
\Cref{thm:general-main} shows that a round must fail validation
within $T_{\max}$ rounds with probability at least
$1-\delta=0.96$, conditional on this initialization.
If $\xcol=0$, the sketch carries no information about the signal,
so every response rule has signal-gap loss
$5/(10+c)>\delta_0$; the validation bound
(\Cref{lem:adaptive-phase-validation}(b)) gives the same rejection
probability in the first round.
Thus the attack returns a rejected round with probability at
least $0.95\cdot0.96$.

Each round draws $r$ independent queries from its announced
distribution, and the stopping rule returns a round with
empirical signal-gap loss $E_t/r>\delta_0/2$.
Every covariance used by the attack is norm-admissible by
\eqref{eq:K-admissible}.
For the round length in \eqref{eq:general-parameters},
\Cref{cor:empirical-norm-transfer} with $N=T_{\max}$ and
$\varepsilon=0.01$ gives, with probability at least $0.99$,
the simultaneous bounds
\[
 \frac1r\sum_{j=1}^r e^{\rm norm}_{t,j}
 \ge \frac{E_t}{r}-\frac{\delta_0}{4}
\]
for all executed rounds, where $e^{\rm norm}_{t,j}$ is the loss
for the original norm task.
Consequently, with probability at least
$0.95\cdot0.96-0.01>0.9$, the returned round contains more than
a $\delta_0/4$ fraction of norm-task errors.

Since $\lambda\le1/2$,
$\lambda/q_0\le2m\lambda=80C_\lambda\alpha^2k$.
Also $\alpha^2k\ge1$, so
$\log(m\lambda)=O(\log(2+\alpha^2k))$. Consequently,
\[
 T_{\max}=O\bigl(1+\alpha^2k\log(2+\alpha^2k)\bigr),
 \qquad r=\Theta(k),
\]
where the second bound uses $T_{\max}\le k^{O(1)}$ and fixed
$\delta$. The total query count is
\begin{equation}\label{eq:end-to-end-query-bound}
 Q=rT_{\max}=O\bigl(k+\alpha^2k^2\log(2+\alpha^2k)\bigr).
\end{equation}
All queries are supported on $S$, and
\Cref{sec:general-computation} gives the stated working-memory and
arithmetic bounds.

For $0<\alpha<C_{\min}/\sqrt{k}$ and $\alpha\le\alpha_0$,
\Cref{cor:isotropic-small-accuracy} with $C=C_{\min}$ gives the
guarantee for either norm task in one round of $\Theta(k)$ queries,
using $O(k)$ working memory and $O(k^2)$ arithmetic operations.
Its norm-estimation error threshold $\alpha/16$ is at least
$c_{\rm err}\alpha$. Take
\[
 c_{\rm frac}=\min\left\{\frac{\delta_0}{4},
                \frac{\delta_{\rm iso}(c,C_{\min})}{2}\right\}
\]
to cover both regimes.
Finally, for arbitrary $\alpha\in(0,1]$, run at
$\beta=\min\{\alpha,\alpha_0\}$. Since
$\alpha_0\alpha\le\beta\le\alpha$, the forced relative error
remains $\Omega(\alpha)$ and the norm gap remains
$1+\Theta(\alpha)$, while the resource bounds at $\beta$ are at
most those stated at $\alpha$.
\end{proof}

\subsection{Gaussian geometry and initialization}
\label{sec:random-signal-coordinate}

\subsubsection{Gaussian identities and degenerate cases}
\label{sec:scalar-channel}

We collect the scalar-channel definitions and Gaussian identities used
below, including the degenerate cases. With $K_\Sig$ as in
\eqref{eq:Ksigma}, define the normalized \emph{Fisher information} for the scalar signal,
\begin{equation}\label{eq:fisher-information}
 \Info(\Sig)=
 \begin{cases}
  \xcol^\top K_\Sig^\dagger\xcol,
       & \xcol\in\im(K_\Sig),\\
  +\infty, & \xcol\notin\im(K_\Sig),
 \end{cases}
\end{equation}
and the \emph{normalized signal-estimation variance}
\begin{equation}\label{eq:signal-estimation-variance}
 \Vinfo(\Sig)=\Info(\Sig)^{-1},
\end{equation}
with the conventions $1/(+\infty)=0$ and $1/0=+\infty$.  The actual Fisher
information in the observation $y=w\xcol+c\Bmat u$ is
$\Info(\Sig)/c^2$, and the actual signal-estimation variance is
$c^2\Vinfo(\Sig)$.

For $0<\Info(\Sig)<\infty$, the statistic in
\eqref{eq:scalar-statistic} is the minimum-variance unbiased estimator
in this Gaussian location family \citep{LehmannCasella98}. It satisfies
\begin{equation}\label{eq:gls-estimator}
 T_\Sigma(y)=w+c\Vinfo(\Sig)\ahat_\Sig^\top u,
 \qquad
 \ahat_\Sig=\Bmat^\top K_\Sig^\dagger\xcol.
\end{equation}
In particular, the estimation error in \eqref{eq:estimator-noise} is
a linear function of $u$, independent of $w$.

To verify sufficiency directly, put
$d=\Vinfo(\Sig)K_\Sig^\dagger\xcol$, $\epsilon=c\Bmat u$, and
$T=T_\Sig(y)=d^\top y$. Since $d^\top\xcol=1$,
\[
 T=w+d^\top\epsilon,\qquad
 r=y-\xcol T=(I-\xcol d^\top)\epsilon.
\]
Moreover,
\[
 \Cov(r,T-w)
 =c^2\bigl(K_\Sig d-\xcol\,d^\top K_\Sig d\bigr)=0.
\]
Gaussianity makes $r$ independent of $T-w$. Since $w$ is independent
of the Gaussian noise, $r$ is independent of $(w,T)$. The identity
$y=\xcol T+r$ therefore proves that $T$ is sufficient for $w$.

\subsubsection{Variational geometry}

The scalar-channel variance has a useful geometric description.  If
$\xcol\notin\im(\Bmat)$, the sketch has a noiseless signal component and
$\Vinfo(\Sig)=0$ for every covariance.  The successful initialization event
in \Cref{cor:random} ensures $\xcol\in\im(\Bmat)$, which is why we may
assume this condition in the subsequent optimization.
When it holds, let $z_*$ be the minimum-norm solution of
$\Bmat z=\xcol$: a noise vector that the sketch maps exactly onto the
signal column.

\begin{lemma}[Variational formula]\label{lem:variational}
For every positive semidefinite $\Sig$,
\begin{equation}\label{eq:variational-info}
  \Info(\Sig)
  =\min\left\{z^\top\Sig^\dagger z:\ \Bmat z=\xcol,\ z\in\im(\Sig)\right\},
\end{equation}
where the minimum is $+\infty$ if the constraint set is empty.  Consequently,
if $\xcol\in\im(\Bmat)\setminus\{0\}$, then
\begin{equation}\label{eq:isotropic-beta}
  \Vinfo(I_m/m)=\frac1{m\|z_*\|_2^2},
  \qquad
  \Vinfo(\Sig)\le\frac{\lambda_{\max}(\Sig)}{\|z_*\|_2^2}
  =m\lambda_{\max}(\Sig)\,\Vinfo(I_m/m)
\end{equation}
for every $\Sig$, with equality in the second relation when $z_*$ is an
eigenvector of $\Sig$ for its largest eigenvalue.
\end{lemma}

\begin{proof}
We first prove the variational formula, then derive the two variance
bounds.  Put $L=\Sig^{1/2}$, so
$K_\Sig=(\Bmat L)(\Bmat L)^\top$ and
$\im(K_\Sig)=\im(\Bmat L)$.  If $\xcol$ is outside this image, then
$\Info(\Sig)=+\infty$ by definition and the constraint set in
\eqref{eq:variational-info} is empty.

Otherwise, the minimum-norm solution of a feasible system $Dq=x$ has
squared norm $x^\top(DD^\top)^\dagger x$.  Apply this identity first to
$D=\Bmat L$, then to $D=L$, and group the feasible vectors $q$ according
to $z=Lq$.  This gives
\[
 \Info(\Sig)
 =\min_{\Bmat Lq=\xcol}\|q\|_2^2
 =\min_{\substack{\Bmat z=\xcol\\z\in\im(\Sig)}}
       \ \min_{Lq=z}\|q\|_2^2
 =\min_{\substack{\Bmat z=\xcol\\z\in\im(\Sig)}}
       z^\top\Sig^\dagger z.
\]

Now assume $\xcol\in\im(\Bmat)\setminus\{0\}$.  For the isotropic
covariance, the formula immediately gives
$\Info(I_m/m)=m\|z_*\|_2^2$.
For any nonzero covariance, every feasible $z$ satisfies
\[
 z^\top\Sig^\dagger z
 \ge\frac{\|z\|_2^2}{\lambda_{\max}(\Sig)}
 \ge\frac{\|z_*\|_2^2}{\lambda_{\max}(\Sig)}.
\]
Minimizing and taking reciprocals proves the upper bound on
$\Vinfo(\Sig)$; it also holds when the constraint set is empty, since
then $\Vinfo(\Sig)=0$.  If $z_*$ is a top eigenvector, it is feasible and
attains both inequalities, proving equality.  Finally, for $\Sig=0$,
both $\Vinfo(\Sig)$ and the stated upper bound are zero.
\end{proof}

\subsubsection{Isotropic initialization on any fixed support}
\label{sec:isotropic-start}

The random-coordinate initialization is the same as that of
\citet[Lemma~4.5]{AhmadianCS25}.
Their bound uses $m=\Omega(k\log^2 k)$ noise coordinates and guarantees
variance $\Omega(1/(k\log k))$ with high constant probability.
The leverage-score argument below removes the logarithmic factors:
$m=\Theta(k)$ coordinates suffice for variance $\Omega(1/k)$.

\begin{corollary}[Random-signal-coordinate baseline]
\label{cor:random}
Fix any set $S\subseteq[n]$ of $m+1$ coordinates, with
$m+1\ge2k/\omega$ for some $\omega\in(0,1)$.  Choose $h$ uniformly from
$S$ and put $M=S\setminus\{h\}$.  Then with probability at least
$1-\omega$ over $h$,
either $\xcol=0$ or
\begin{equation}\label{eq:random-iso-high-var}
  \frac{\omega}{2k}\le\Vinfo(I_m/m)<\infty .
\end{equation}
In particular, for $m\ge40k$ this holds with probability at least $0.95$ and
lower bound $1/(40k)$.
\end{corollary}
If $\xcol=0$, the sketch is independent of the signal, $\Vinfo(\Sig)=+\infty$
for every $\Sig$, and every query law is already universally hard for
the signal-gap task. In the covariance-update regime, admissibility
transfers this signal-gap hardness to the norm tasks
(\Cref{lem:norm-to-gap}).
In the nonzero case, the successful event also ensures
$\xcol\in\im(\Bmat)$.  This is a guarantee over the random choice of $h$;
the attacker need not test it.

The proof uses the following identity, a standard consequence of the
Sherman--Morrison formula \citep[Eq.~(2)]{Hager89} applied on $\im(C)$.
We include a direct projector proof that also covers rank-deficient $C$.

\begin{lemma}[Leverage-score representation identity]
\label{lem:leverage-representation}
Let $C\in\R^{k\times s}$ have columns $a_1,\ldots,a_s$, and let
\[
  \tau_j=a_j^\top(CC^\top)^\dagger a_j.
\]
Let $C_{-j}$ denote the matrix obtained from $C$ by deleting its $j$th
column.  If $\tau_j<1$, then
\begin{equation}\label{eq:beta-leverage-bound}
  \min\{\|z\|_2^2:C_{-j}z=a_j\}
  = \frac{\tau_j}{1-\tau_j}.
\end{equation}
Consequently, if $\omega s>k$, all but an $\omega$ fraction of the columns
satisfy $a_j\in\im(C_{-j})$ and
\[
  \min\{\|z\|_2^2:C_{-j}z=a_j\}
  \le \frac{k}{\omega s-k}.
\]
\end{lemma}
\begin{proof}
Let $P=C^\top(CC^\top)^\dagger C$ be the orthogonal projector onto the
row space of $C$, so $P_{jj}=\tau_j$.
For $\tau_j<1$, the vector $r=(I-P)e_j$ lies in $\ker(C)$ and satisfies
$r_j=\|r\|_2^2=1-\tau_j$.  Thus
\[
 q=e_j-\frac{r}{1-\tau_j}
 \quad\text{satisfies}\quad
 q_j=0,\qquad Cq=a_j,\qquad
 \|q\|_2^2=\frac{\tau_j}{1-\tau_j}.
\]
Deleting the zero coordinate gives a feasible representation of $a_j$.
To see that it has minimum norm, take any other $q'$ with
$q'_j=0$ and $Cq'=a_j$.  Then $d=e_j-q'$ lies in $\ker(C)$ and has
$d_j=1$.  Since $r$ is the projection of $e_j$ onto this kernel,
\[
 1=\langle d,e_j\rangle=\langle d,r\rangle
 \le\|d\|_2\sqrt{1-\tau_j}.
\]
It follows that
$\|q'\|_2^2=\|d\|_2^2-1\ge\tau_j/(1-\tau_j)$, proving the identity.

Finally, $\sum_j\tau_j=\tr(P)=\rank(C)\le k$.
Hence all but an $\omega$ fraction of the columns have
$\tau_j\le k/(\omega s)<1$, and for these columns
$\tau_j/(1-\tau_j)\le k/(\omega s-k)$.
\end{proof}

\begin{proof}[Proof of \Cref{cor:random}]
Apply \Cref{lem:leverage-representation} to the restricted matrix
$C=A_{\bullet S}$ with $s=m+1$.  With probability at least $1-\omega$,
$\xcol\in\im(\Bmat)$ and
\[
 \|z_*\|_2^2
 \le\frac{k}{\omega(m+1)-k}
 \le\frac{2k}{\omega(m+1)},
\]
where the last inequality uses $\omega(m+1)\ge2k$.
If $\xcol=0$, the first alternative holds.  Otherwise
\Cref{lem:variational} gives
\[
 \Vinfo(I_m/m)=\frac1{m\|z_*\|_2^2}
 \ge\frac{\omega(m+1)}{2km}
 \ge\frac{\omega}{2k}.
\]
\end{proof}

\subsubsection{The hard comparator}
On the nonzero successful initialization event, $\xcol\in\im(\Bmat)\setminus\{0\}$.
Since every covariance in the search space is positive definite,
$0<\Vinfo(\Sig)<\infty$ throughout.  Recall that $z_*$ is the minimum-norm
solution of $\Bmat z=\xcol$, and set
$v_\star=z_*/\|z_*\|_2$ and $\Sig_\star=\Sig_{v_\star}$, where
\begin{equation}\label{eq:universal-spike-family}
 \Sig_v=q_0I_m+(\lambda-q_0)vv^\top,
 \qquad v\in\mathbb S^{m-1}.
\end{equation}
By \Cref{lem:variational}, this spike maximizes the signal-estimation
variance over $\mathcal K_{\lambda}$, with
\begin{equation}\label{eq:comparator-variance}
 \Vinfo(\Sig_\star)=m\lambda\,\Vinfo(I_m/m).
\end{equation}
The comparator is used only in the analysis; the attacker need not
compute it.  The initialization bound gives $\Vinfo(\Sig_\star)\ge\lambda$ when
$m=40k$.  The case $\xcol=0$ will be handled separately.

\paragraph{The eigenvalue-cap benchmark.}
The same benchmark holds without an eigenvalue floor, as needed for
the sufficient-statistic attack.  For $1/m\le\lambda\le1$,
\begin{equation}\label{eq:capped-offline-optimum}
 \max_{\substack{\Sig\succeq0,\ \tr\Sig=1\\
                  \lambda_{\max}(\Sig)\le\lambda}}
       \Vinfo(\Sig)
 =m\lambda\,\Vinfo(I_m/m).
\end{equation}
Indeed, \Cref{lem:variational} gives the upper bound, attained by
$\lambda v_*v_*^\top+(1-\lambda)(I_m-v_*v_*^\top)/(m-1)$.
The general attack uses this unknown spike as a comparator in its
entropy analysis.  The sufficient-statistic attack instead preserves a
constant fraction of this optimum at its attained top eigenvalue.
Norm admissibility additionally requires a bound on $\tr(\Sig^2)$,
which each attack establishes for the covariances it uses.
Neither attack needs to compute $z_*$.

\begin{remark}[Offline optimum under norm admissibility]
\label{rem:offline-exact}
For the same fixed $h,M$ and $1/m\le\rho\le1$,
\[
 \max_{\substack{\Sig\succeq0,\ \tr\Sig=1\\\tr(\Sig^2)\le\rho}}
 \Vinfo(\Sig)
 =\bigl(1+\sqrt{(m-1)(m\rho-1)}\bigr)\Vinfo(I_m/m).
\]
Indeed, $\tr(\Sig^2)\ge\lambda^2+(1-\lambda)^2/(m-1)$ for
$\lambda=\lambda_{\max}(\Sig)$; equality and the variance bound in
\Cref{lem:variational} are attained by a spike along $z_*$ with equal
eigenvalues on its orthogonal complement.
\end{remark}

\subsection{Correlation and validation in an adaptive round}
\label{sec:round-analysis}
\label{sec:gain}

\subsubsection{Correlation and hardness from the gain lemma}

We record the two consequences of \Cref{lem:gain} used in the attack
analyses: a response--noise correlation and a quantitative hardness
threshold. The signal density and the proof of the gain lemma are given
in \Cref{app:gain-proof}.

To express \Cref{lem:gain} in the noise coordinates, define
the response-weighted mean
\begin{equation}\label{eq:first-moment}
 m_\oracle(\Sig)=\E[\oracle(y)u].
\end{equation}

For $0<\Vinfo(\Sig)<\infty$, define the \emph{gain}
\begin{equation}\label{eq:gain-correlation}
 \Gamma_\oracle(\Sig)
 =\E[\oracle(y)(T_\Sigma(y)-w)]
 =c\Vinfo(\Sig)\langle\ahat_\Sig,m_\oracle(\Sig)\rangle.
\end{equation}
The identity follows from \eqref{eq:gls-estimator}.  Since the estimation
error has mean zero, the gain is its covariance with the response.

Dividing the scalar gain bound of \Cref{lem:gain} by
$c\Vinfo(\Sig)$ gives its vector formulation:
\begin{equation}\label{eq:general-gain}
 \loss_\Sig(\oracle)\le\delta_0
 \quad\Longrightarrow\quad
 \langle\ahat_\Sig,m_\oracle(\Sig)\rangle\ge g.
\end{equation}
The bound holds throughout the positive finite variance range.

The gain lemma also gives a quantitative hardness threshold
at the loss level $\delta_0$ used in the attack analysis.

\begin{corollary}[Accuracy bounds the variance]
\label{cor:gain-hardness}
Under the hypotheses of \Cref{lem:gain}, any response rule with
signal-gap loss at most $\delta_0$ satisfies
\begin{equation}\label{eq:gain-hardness-main}
 g^2\Vinfo(\Sig)\le1.
\end{equation}
Consequently, if $g^2\Vinfo(\Sig)>1$, every signal-gap response rule
has loss greater than $\delta_0$.
\end{corollary}
\begin{proof}
Since $|\oracle|\le1$, Cauchy--Schwarz and
\eqref{eq:estimator-noise} bound the left side of
\eqref{eq:scalar-gain} by $c\sqrt{\Vinfo(\Sig)}$.
Combining the bounds gives $g^2\Vinfo(\Sig)\le1$.
\end{proof}

\subsubsection{Conditional losses and information sets}
\label{sec:probability-conventions}

The transcript $\mathsf{Tr}_{i-1}$ in \Cref{def:interaction} records
the complete execution of both parties through completed query $i-1$.
It is available to the responder, but its private part is not available
to the attacker.  We distinguish two information sets just before the
fresh draw of query $i$:
\begin{itemize}[leftmargin=1.7em]
\item $\mathcal H_{i-1}$ is generated by the attacker's disclosed
execution and the returned answers, including the current announcement
of $\mathcal D_i$ and all computation and randomness disclosed in step~(i).
\item $\mathcal F_{i-1}$ additionally contains the full transcript
$\mathsf{Tr}_{i-1}$ and the responder's private randomness selecting
its current map $\oracle_i$ in step~(ii).
\end{itemize}
Thus $\mathcal H_{i-1}\subseteq\mathcal F_{i-1}$, the map $\oracle_i$
is $\mathcal F_{i-1}$-measurable, and freshness of the query randomness gives
\[
  v_i\mid\mathcal F_{i-1}\sim\mathcal D_i,
  \qquad v_i\mid\mathcal H_{i-1}\sim\mathcal D_i.
\]
After the answer and disclosure, both histories contain the generated
input, the answer, and the attacker's sampling variables.  The current
distribution announcement is additional to $\mathsf{Tr}_{i-1}$;
the latter always ends at completed query $i-1$.

The realized loss $e_i$ is observed by the attacker.  Its conditional
expectation
\[
  \varepsilon_i:=\E[e_i\mid\mathcal F_{i-1}]
\]
is the \emph{predictable loss} used in the analysis.  Validation relates
the average realized loss in a round to the average predictable loss.

Conditioning on $\mathcal H_{i-1}$ instead averages over the
responder's private history as well as its current randomization.
The sufficient-statistic definition uses precisely this conditional mean;
enlarging its conditioning information to the full transcript would be
a stronger restriction on the responder.  Within a round, we use $j$ in place of $i$ with the same
conventions.

In the interaction, the realized loss is
$e_i=\loss(\oracle_i(y_i),w_i)$.  Conditionally on
$\mathcal F_{i-1}$, both the rule and the covariance are fixed, so
$\varepsilon_i=\loss_{\Sig_t}(\oracle_i)$ for query $i$ in round $t$.

\subsubsection{Drift and validation}

Fix one round with a positive definite trace-one covariance
$\Sig=RR^\top$, and drop the index $t$. For the drift bound assume
$0<\Vinfo(\Sig)<\infty$; the validation lemma does not require this
assumption. Neither lemma requires an eigenvalue floor or cap.
Both lemmas hold for either information set from
\Cref{sec:probability-conventions}. We first write the proofs with
$\mathcal F_{j-1}$ and explain the use of $\mathcal H_{j-1}$ below.
For $j=1,\dots,r$ use the enlarged information set
$\mathcal F_{j-1}$ of \Cref{sec:probability-conventions}, including the
responder's private randomness that selects the current response rule.  Thus
$\oracle_j$ is $\mathcal F_{j-1}$-measurable.  Put
\begin{equation}\label{eq:adaptive-phase-sample}
 X_j=\oracle_j(y_j)\,z_j,
 \qquad \mu_j=\E[X_j\mid\mathcal F_{j-1}],
 \qquad e_j=\loss(\oracle_j(y_j),w_j),
 \qquad \varepsilon_j=\E[e_j\mid\mathcal F_{j-1}],
\end{equation}
so that $\varepsilon_j=\loss_\Sig(\oracle_j)$ is the predictable loss.  The
signed noise sum of the round splits into a predictable drift and a
martingale,
\begin{equation}\label{eq:drift-martingale}
 S_r=\sum_{j=1}^rX_j=M_r+N_r,
 \qquad M_r=\sum_{j=1}^r\mu_j,
 \qquad N_r=\sum_{j=1}^r(X_j-\mu_j).
\end{equation}
The attacker observes $S_r$ and the loss sum $E=\sum_je_j$, but neither
$M_r$ nor the $\varepsilon_j$.  The signal-estimation direction in whitened coordinates is
\[
 h_R=\sqrt{\Vinfo(\Sig)}\,R^\top\ahat_\Sig,
\]
a unit vector, since
$\|h_R\|_2^2=\Vinfo(\Sig)\,\ahat_\Sig^\top\Sig\ahat_\Sig
=\Vinfo(\Sig)\,\xcol^\top K_\Sig^\dagger K_\Sig K_\Sig^\dagger\xcol
=\Vinfo(\Sig)\Info(\Sig)=1$.  Because $u=Rz$, the first moment
\eqref{eq:first-moment} is $m_\oracle(\Sig)=R\,\E[\oracle(y)z]$, so for every
response rule
\begin{equation}\label{eq:whitened-gain}
 \bigl\langle h_R,\E[\oracle(y)z]\bigr\rangle
 =\sqrt{\Vinfo(\Sig)}\,\bigl\langle\ahat_\Sig,m_\oracle(\Sig)\bigr\rangle.
\end{equation}

\begin{lemma}[Drift of an accurate round]\label{lem:adaptive-phase-gain}
Each $\mu_j$ lies in $\im(R^\top\Bmat^\top)$ and has $\|\mu_j\|_2\le1$, and
\begin{equation}\label{eq:adaptive-conditional-covariance}
 \E[(X_j-\mu_j)(X_j-\mu_j)^\top\mid\mathcal F_{j-1}]\preceq I_m.
\end{equation}
If $r^{-1}\sum_j\varepsilon_j\le\delta_0$, then
\begin{equation}\label{eq:adaptive-sum-score}
 \langle h_R,M_r\rangle
 \ge rg\sqrt{\Vinfo(\Sig)}.
\end{equation}
These conclusions hold for every realization of the predictable rules;
the drift bound is asserted on the event that their average loss is at
most $\delta_0$.
\end{lemma}

\begin{proof}
The query depends on $z_j$ only through $\Bmat Rz_j$, and Gaussian
conditioning gives $\E[z_j\mid\Bmat Rz_j]\in\im(R^\top\Bmat^\top)$; hence so
is $\mu_j=\E[\oracle_j(y_j)\E[z_j\mid w_j,\Bmat Rz_j]\mid\mathcal F_{j-1}]$.
For every unit vector $b$, Cauchy--Schwarz gives
$|\langle b,\mu_j\rangle|=|\E[\oracle_j(y_j)b^\top z_j\mid\mathcal F_{j-1}]|
\le1$, proving the norm bound.  Also,
\[
 \E[(X_j-\mu_j)(X_j-\mu_j)^\top\mid\mathcal F_{j-1}]
 \preceq\E[\oracle_j(y_j)^2z_jz_j^\top\mid\mathcal F_{j-1}]
 \preceq I_m.
\]
For the drift, fix a realization of the rules and let
$\bar\oracle=r^{-1}\sum_j\oracle_j$.  Evaluate these functions on a fresh
generic query from $\mathcal D_\Sig$.  Loss and first moment commute
with averaging, so
$\loss_\Sig(\bar\oracle)=r^{-1}\sum_j\varepsilon_j\le\delta_0$
and $\E[\bar\oracle(y)z]=r^{-1}M_r$. The vector form
\eqref{eq:general-gain} of \Cref{lem:gain} gives
$\langle\ahat_\Sig,m_{\bar\oracle}(\Sig)\rangle\ge g$, and
\eqref{eq:whitened-gain} turns this into \eqref{eq:adaptive-sum-score}.
\end{proof}

\begin{lemma}[Empirical validation]\label{lem:adaptive-phase-validation}
Call the round \emph{accepted} if
$E\le r\delta_0/2$ and \emph{rejected} otherwise.  If
$r\ge32\delta_0^{-2}\log(2/\delta_{\rm rd})$, then with probability at
least $1-\delta_{\rm rd}$ all three of the following hold:
\begin{enumerate}[label=\textup{(\alph*)},leftmargin=2.2em]
\item if the round is accepted, then
\begin{equation}\label{eq:validation-implication}
 \frac1r\sum_j\varepsilon_j<\delta_0 ;
\end{equation}
\item if every response rule has signal-gap loss at least $\delta_0$
at $\Sig$, then the round is rejected;
\item if the round is rejected, then
$r^{-1}\sum_j\varepsilon_j>\delta_0/4$.
\end{enumerate}
\end{lemma}

\begin{proof}
The differences $e_j-\varepsilon_j$ form a martingale difference sequence
bounded by one, so the Azuma--Hoeffding inequality gives
$|r^{-1}\sum_j(e_j-\varepsilon_j)|\le\delta_0/4$ with probability at
least $1-2e^{-r\delta_0^2/32}\ge1-\delta_{\rm rd}$.  On this event an
accepted round has $r^{-1}\sum_j\varepsilon_j\le3\delta_0/4$, which
is~(a); under the hypothesis of~(b) every $\varepsilon_j\ge\delta_0$,
so $E/r\ge3\delta_0/4>\delta_0/2$; and a rejected round has
$r^{-1}\sum_j\varepsilon_j>\delta_0/2-\delta_0/4$, which
is~(c).
\end{proof}

The same statements hold with $\mathcal H_{j-1}$ in place of
$\mathcal F_{j-1}$. To compute the conditional means, replace the
privately selected rule by
$\bar\oracle_j(y)=\E[\oracle_j(y)\mid\mathcal H_{j-1}]$.
Fresh query randomness is independent of this private rule, and both
loss and first moment commute with averaging over it. The drift proof
therefore applies to $\bar\oracle_j$. The conditional covariance bound
still uses the actual bounded answer, and the validation proof uses
the actual realized loss. This is the version used by the
sufficient-statistic attack.

\subsection{Separating measurements}
\label{sec:general-phase}

We now specialize to the general attack, with
$\Sig=RR^\top\in\mathcal K_\lambda$ and $1/m<\lambda\le1/2$.
Thus $q_0=(1-\lambda)/(m-1)\ge1/(2m)$.
Use the drift decomposition of \Cref{sec:round-analysis}, with
conditioning on $\mathcal F_{j-1}$.

The gain lemma controls the component along $h_R$, while the responder
can influence the other components.  The next inequality shows that these
other components cannot eliminate separation from the comparator.  Let $\Sig_\star\in\mathcal K_{\lambda}$ be the
comparator and set $X_\star=R^{-1}\Sig_\star R^{-\top}$.

\begin{lemma}[Comparator separation]\label{lem:comparator-separation}
For every $z\in\im(R^\top\Bmat^\top)$,
\begin{equation}\label{eq:comparator-rayleigh-homogeneous}
 z^\top X_\star z
 \ge \frac{\Vinfo(\Sig_\star)}{\Vinfo(\Sig)}
      \langle h_R,z\rangle^2.
\end{equation}
\end{lemma}

\begin{proof}
Write $z=R^\top\Bmat^\top\xi$ for some $\xi\in\R^k$.  Since
$\xcol\in\im(K_\Sig)$,
\[
 \langle h_R,z\rangle
 =\sqrt{\Vinfo(\Sig)}\,\ahat_\Sig^\top\Sig\Bmat^\top\xi
 =\sqrt{\Vinfo(\Sig)}\,\xcol^\top K_\Sig^\dagger K_\Sig\xi
 =\sqrt{\Vinfo(\Sig)}\,\xi^\top\xcol ,
 \qquad
 z^\top X_\star z=\xi^\top K_{\Sig_\star}\xi .
\]
Cauchy--Schwarz gives
\[
 (\xi^\top\xcol)^2
 =\bigl(\xi^\top K_{\Sig_\star}^{1/2}K_{\Sig_\star}^{\dagger1/2}
 \xcol\bigr)^2
 \le\bigl(\xi^\top K_{\Sig_\star}\xi\bigr)\Info(\Sig_\star)
 =\frac{z^\top X_\star z}{\Vinfo(\Sig_\star)},
\]
so $z^\top X_\star z\ge\Vinfo(\Sig_\star)(\xi^\top\xcol)^2
=\bigl(\Vinfo(\Sig_\star)/\Vinfo(\Sig)\bigr)\langle h_R,z\rangle^2$.
\end{proof}

It remains to transfer separation from the conditional means to the
observed sum $S_r=M_r+N_r$. The covariance floor controls the
sampling error in the comparator's quadratic form:
\begin{equation}\label{eq:comparator-trace-bound}
 \tr X_\star
 =\tr(\Sig_\star\Sig^{-1})
 \le\frac{\tr\Sig_\star}{q_0}
 =\frac1{q_0}\le2m.
\end{equation}
This bound lets us handle dependence between $M_r$ and $N_r$
using only martingale second moments.

\begin{lemma}[A measurement from the observed sum]
\label{lem:adaptive-comparator-cut}
Run one round of $r\ge48m$ queries at $\Sig=RR^\top$.
Let $\mathcal F_{\rm start}$ contain the full execution history and
the announced covariance before the round's first query, and let
$X_\star\succeq0$ be determined by this information with
$\tr X_\star\le2m$. Set
\begin{equation}\label{eq:adaptive-observed-direction}
 q=S_r/\|S_r\|_2,
\end{equation}
choosing any unit vector if $S_r=0$.
For every responder there is an event $\mathcal G$ such that, for
every $L\ge4$,
\begin{equation}\label{eq:adaptive-cut-probability}
 \Pr[\mathcal G\mid\mathcal F_{\rm start}]\ge p_{\rm cut}
 \qquad\text{and}\qquad
 q^\top X_\star q\ge c_{\rm cut}L
 \quad\text{on }\mathcal G\cap\{M_r^\top X_\star M_r\ge r^2L\},
\end{equation}
where one may take $p_{\rm cut}=1/2$ and $c_{\rm cut}=1/16$.
\end{lemma}

The probability is conditional on the information \emph{before}
the round. We do not condition on accuracy or acceptance, which may
be correlated with $\mathcal G$. The implication on the intersection
in \eqref{eq:adaptive-cut-probability} is what the entropy argument uses.

\begin{proof}
All expectations and probabilities are conditional on
$\mathcal F_{\rm start}$. By
\eqref{eq:adaptive-conditional-covariance}, the martingale differences
have conditional covariance at most $I_m$, so
\[
 \E\|N_r\|_2^2\le rm,
 \qquad
 \E[N_r^\top X_\star N_r]\le r\tr X_\star\le2rm.
\]
Define
\[
 \mathcal G=\{\|N_r\|_2^2+N_r^\top X_\star N_r\le r^2\}.
\]
Markov's inequality and $r\ge48m$ give
$\Pr[\mathcal G]\ge1-3m/r\ge15/16\ge p_{\rm cut}$.
Write $\|v\|_\star=(v^\top X_\star v)^{1/2}$.
On $\mathcal G$, both $\|N_r\|_2$ and $\|N_r\|_\star$ are at most $r$.
Since $\|M_r\|_2\le r$ by \Cref{lem:adaptive-phase-gain}, we have
$\|S_r\|_2\le2r$. Whenever $\|M_r\|_\star\ge r\sqrt L$ with $L\ge4$,
the triangle inequality gives
\[
 \|S_r\|_\star\ge r(\sqrt L-1)\ge\tfrac12r\sqrt L.
\]
In particular, $S_r\ne0$ on this intersection, and
\[
 q^\top X_\star q
 =\frac{\|S_r\|_\star^2}{\|S_r\|_2^2}
 \ge\frac{L}{16}.
\]
No independence between $M_r$ and $N_r$ is required.
\end{proof}

To obtain separation from the comparator, we require the
\emph{comparator margin}
\begin{equation}\label{eq:universal-comparator-margin}
 g^2\Vinfo(\Sig_\star)\ge
 C_{\rm marg}:=\frac{8}{p_{\rm cut}c_{\rm cut}}=256,
\end{equation}
where $g=c_gc^2/\alpha$ is defined in \Cref{lem:gain}, and
$p_{\rm cut}=1/2$ and $c_{\rm cut}=1/16$ are the constants
from \Cref{lem:adaptive-comparator-cut}.
By \Cref{cor:gain-hardness}, this margin also makes the comparator
universally hard for the signal-gap task at loss level $\delta_0$.
The attacker need not reach or identify this comparator.

We now express separation in terms of the rank-one measurement
used by the algorithm. For $\Sig=RR^\top$, write
$X_\star=R^{-1}\Sig_\star R^{-\top}$ and, for a unit vector $q$, set
\[
 a=R^{-\top}q,\qquad Z(q;\Sig)=q_0aa^\top,
\]
as in \eqref{eq:matrix-cut-definition}.

\begin{lemma}[Normalized rank-one separator]\label{lem:normalized-matrix-cut}
For every unit $q$ and every $\Sig\in\mathcal K_{\lambda}$,
\begin{equation}\label{eq:matrix-cut-current}
 0\preceq Z(q;\Sig)\preceq I_m,
 \qquad \langle\Sig,Z(q;\Sig)\rangle=q_0,
 \qquad \langle\Sig_\star,Z(q;\Sig)\rangle=q_0\,q^\top X_\star q .
\end{equation}
\end{lemma}

\begin{proof}
Since $a^\top\Sig a=q^\top q=1$ and $\Sig\succeq q_0I_m$,
$q_0\|a\|_2^2\le1$.  Hence $q_0aa^\top\preceq I_m$ and
$\langle\Sig,q_0aa^\top\rangle=q_0$.  Finally,
$a^\top\Sig_\star a=q^\top R^{-1}\Sig_\star R^{-\top}q=q^\top X_\star q$.
\end{proof}

Combining the drift, comparator separation, and sampling-error bounds
gives the separating measurement needed for the entropy argument.

\begin{lemma}[Separation from an accurate round]\label{lem:good-phase}
Let the comparator $\Sig_\star$ satisfy the margin
\eqref{eq:universal-comparator-margin}, and let $\mathcal G$ be the event of
\Cref{lem:adaptive-comparator-cut} for the round, with
$L=C_{\rm marg}$ and $X_\star=R^{-1}\Sig_\star R^{-\top}$;
\eqref{eq:comparator-trace-bound} verifies its trace assumption.
On the event $r^{-1}\sum_j\varepsilon_j\le\delta_0$, we have
\begin{equation}\label{eq:adaptive-comparator-lower}
 M_r^\top X_\star M_r
 \ge r^2g^2\Vinfo(\Sig_\star)\ge r^2C_{\rm marg}.
\end{equation}
On the intersection of this event with $\mathcal G$, the measurement direction satisfies
\begin{equation}\label{eq:good-phase}
 q^\top X_\star q\ge\kappa:=c_{\rm cut}C_{\rm marg}=\frac8{p_{\rm cut}}=16.
\end{equation}
\end{lemma}

\begin{proof}
By \Cref{lem:adaptive-phase-gain}, $M_r\in\im(R^\top\Bmat^\top)$ and
$\langle h_R,M_r\rangle\ge rg\sqrt{\Vinfo(\Sig)}$, so
\Cref{lem:comparator-separation} at $z=M_r$ gives
$M_r^\top X_\star M_r\ge\Vinfo(\Sig_\star)r^2g^2$, and the margin gives the
second inequality in \eqref{eq:adaptive-comparator-lower}.  This is the
hypothesis of \eqref{eq:adaptive-cut-probability} with $L=C_{\rm marg}$.
\end{proof}

Restoring the round index $t$, every measurement satisfies
\[
  0\preceq Z_t\preceq I_m,\qquad
  \langle\Sig_t,Z_t\rangle=q_0.
\]
On the validation event of \Cref{lem:adaptive-phase-validation},
an accepted round also satisfies
\[
  \langle\Sig_\star,Z_t\rangle\ge16q_0
\]
whenever its sampling event $\mathcal G_t$ occurs.
Conditional on the information before the round,
$\mathcal G_t$ has probability at least $1/2$.
These are the properties used in the entropy argument below.

\subsection{Entropy progress and termination}
\label{sec:general-potential}
\label{sec:general-proof}

We prove the rejection guarantee used in \Cref{sec:end-to-end}.
For any failure probability $0<\delta<1$, use the round parameters
\eqref{eq:general-parameters}.

\begin{proposition}[Rejection within the round budget]
\label{thm:general-main}
Suppose $1/m<\lambda\le1/2$,
$\xcol\in\im(\Bmat)\setminus\{0\}$, and the spike comparator
of \eqref{eq:universal-spike-family} satisfies
$g^2\Vinfo(\Sig_\star)\ge C_{\rm marg}=256$.
Then \Cref{alg:general-responder-discovery}, with the parameters
\eqref{eq:general-parameters}, rejects a round within $T_{\max}$
rounds with probability at least $1-\delta$ against every informed
signal-gap responder.
\end{proposition}

We now show how these measurements control the nonnegative
potential $D(\Sig_\star\|\Sig_t)$.

\begin{lemma}[One multiplicative-weights step]\label{lem:matrix-entropy-progress}
Let $\Sig=RR^\top$ and $\Sig_\star$ lie in $\mathcal K_{\lambda}$, let $q$
be a unit vector, and let $\Sig^+$ be the update along $Z=Z(q;\Sig)$ with
step $\eta>0$, with unprojected iterate $\widehat\Sig^+$.  Then
\begin{equation}\label{eq:entropy-potential-drop}
 D(\Sig_\star\|\Sig^+)
 \le D(\Sig_\star\|\Sig)-\eta q_0\,q^\top X_\star q+(e^\eta-1)q_0 .
\end{equation}
For $\eta=1/2$ the potential therefore increases by less than $\tfrac23q_0$
after any measurement, and decreases by more than $7q_0$ after a measurement with
$q^\top X_\star q\ge\kappa=16$.
\end{lemma}

\begin{proof}
Since $\Sig_\star\in\mathcal K_{\lambda}$, the Pythagorean inequality for
the relative-entropy projection \citep[Section~5]{WarmuthKuzmin08} gives
\[
 D(\Sig_\star\|\widehat\Sig^+)
 \ge D(\Sig_\star\|\Sig^+)+D(\Sig^+\|\widehat\Sig^+)
 \ge D(\Sig_\star\|\Sig^+).
\]
The first inequality follows directly from the first-order optimality
condition for \eqref{eq:matrix-mw-projection}.  Writing
$\mathcal Z=\tr\exp(\log\Sig+\eta Z)$ and using $\tr\Sig_\star=1$,
\[
 D(\Sig_\star\|\widehat\Sig^+)
 =D(\Sig_\star\|\Sig)
   -\eta\langle\Sig_\star,Z\rangle+\log\mathcal Z.
\]
Golden--Thompson and $0\preceq Z\preceq I_m$, which gives
$e^{\eta Z}\preceq I_m+(e^\eta-1)Z$, imply
\[
 \mathcal Z
 \le\tr(\Sig e^{\eta Z})
 \le1+(e^\eta-1)\langle\Sig,Z\rangle
 =1+(e^\eta-1)q_0 .
\]
Hence $\log\mathcal Z\le(e^\eta-1)q_0$, and
\Cref{lem:normalized-matrix-cut} gives
$\langle\Sig_\star,Z\rangle=q_0\,q^\top X_\star q$.  For $\eta=1/2$,
$e^\eta-1<0.65$ and $8-0.65>7$.
\end{proof}

The starting point of the potential is small.  For the spike comparator
$\Sig_\star=\Sig_{v_\star}$,
\begin{equation}\label{eq:spike-relative-entropy}
 D(\Sig_\star\|I_m/m)
 =\lambda\log(m\lambda)+(1-\lambda)\log(mq_0)
 \le\lambda\log(m\lambda),
\end{equation}
since $mq_0\le1$ when $\lambda\ge1/m$.

\begin{proof}[Proof of \Cref{thm:general-main}]
Write $T=T_{\max}$ and
$X_{\star,t}=R_t^{-1}\Sig_\star R_t^{-\top}$.
Let $\mathcal E_{\rm val}$ be the event that
\Cref{lem:adaptive-phase-validation} holds in every executed round.
Using $\delta_{\rm rd}=\delta/(2T)$ and taking a union bound gives
$\Pr[\mathcal E_{\rm val}]\ge1-\delta/2$.

For each round $t=0,\ldots,T-1$, let $\mathcal G_t$ be the event supplied
by \Cref{lem:adaptive-comparator-cut}, with
$L=C_{\rm marg}$ and comparator $X_{\star,t}$, whose trace is at most
$2m$ by \eqref{eq:comparator-trace-bound}.
The event $\mathcal G_t$ is defined from the round's sampling error,
whether or not validation accepts the round. After the attack stops,
set all later $\mathcal G_t$ to the whole sample space. Each indicator
then has conditional mean at least $p_{\rm cut}$ given the information
before its round.
The multiplicative Chernoff bound for adapted indicators gives
\[
 \Pr\!\left[\sum_{t=0}^{T-1}\ind{\mathcal G_t}
              <\frac{p_{\rm cut}T}{2}\right]
 \le e^{-p_{\rm cut}T/8}\le\frac\delta2,
\]
where the last inequality uses $T\ge16\log(2/\delta)$.
Denote the complementary event by $\mathcal E_{\rm cut}$.
This bound is unconditional on validation; we intersect the two good
events only after bounding their probabilities.

Suppose all $T$ rounds are accepted on
$\mathcal E_{\rm val}\cap\mathcal E_{\rm cut}$.
Their average predictable losses are below
$\delta_0$.  By \Cref{lem:good-phase},
$\widehat q_t^\top X_{\star,t}\widehat q_t\ge\kappa=16$ whenever $\mathcal G_t$ occurs.
There are at least $p_{\rm cut}T/2$ such rounds.  Summing
\eqref{eq:entropy-potential-drop}, and using nonnegativity of the
Rayleigh quotient in the other rounds, gives
\[
 \begin{aligned}
 0\le D(\Sig_\star\|\Sig_T)
 &\le D(\Sig_\star\|I_m/m)
       -\frac{q_0\kappa}{2}\frac{p_{\rm cut}T}{2}
       +0.65q_0T\\
 &<\lambda\log(m\lambda)-q_0T<0.
 \end{aligned}
\]
Here \eqref{eq:spike-relative-entropy} bounds the initial potential,
$\kappa p_{\rm cut}/4=2$, and
$T>\lambda\log(m\lambda)/q_0$.
This contradiction proves that a round is rejected with probability at
least $1-\delta$.
\end{proof}

\subsection{From signal-gap loss to norm errors}
\label{sec:norm-reduction}

We transfer signal-gap guarantees to norm estimation and binary norm-gap
queries by controlling the fluctuations of the noise norm.  By \eqref{eq:query-norm}, the squared query norm differs from
$w^2+c^2$ by $c^2(\|u\|_2^2-1)$.
Since $\tr\Sig=1$,
\[
 \E\|u\|_2^2=1,
 \qquad
 \Var\!\left(c^2(\|u\|_2^2-1)\right)=2c^4\tr(\Sig^2).
\]
The admissibility bound of \Cref{def:admissible} controls these
fluctuations relative to the signal gap.

\begin{definition}[Norm-admissible covariance]\label{def:admissible}
For $\rho\in[1/m,1]$, a positive semidefinite trace-one covariance
$\Sig$ is \emph{$\rho$-admissible} if $\tr(\Sig^2)\le\rho$.
\end{definition}

The following reduction controls both expected losses and the errors
observed in an interaction.  For fixed $c$ and concentration-failure
probability $\varepsilon_{\rm nc}$, it requires an admissibility bound that
is a sufficiently small multiple of $\alpha^2/c^4$.

\begin{lemma}[Norm-to-gap reduction]\label{lem:norm-to-gap}
There is a constant $c_{\rm err}=c_{\rm err}(c)>0$ and, for every
$\varepsilon_{\rm nc}\in(0,1)$, a constant $c_\rho=c_\rho(\varepsilon_{\rm nc})>0$ such
that the following holds.  Suppose $0<\alpha\le c/20$ and $\Sig$ is
$\rho$-admissible with
\begin{equation}\label{eq:norm-reduction-admissibility}
 \rho\le c_\rho\frac{\alpha^2}{c^4}.
\end{equation}
For a query $v=w e_h+c\,\iota_Mu$ and a norm estimate $\widehat r\ge0$, use the thresholded response $b$ from \eqref{eq:threshold-response} and the events
\[
 \mathcal E:=\bigl\{c^2\bigl|\|u\|_2^2-1\bigr|\le\alpha/2\bigr\},
 \qquad
 \mathcal A:=\bigl\{|\widehat r-\|v\|_2|\le c_{\rm err}\alpha\|v\|_2\bigr\}.
\]
Then the per-query signal-gap loss \eqref{eq:signal-gap-loss-function}
satisfies, pointwise,
\begin{equation}\label{eq:norm-to-gap-pointwise}
 \loss(b,w)\le\ind{\mathcal A^c}+\ind{\mathcal E^c},
 \qquad\text{and}\qquad
 \Pr_{u\sim N(0,\Sig)}[\mathcal E^c]\le\varepsilon_{\rm nc} .
\end{equation}
Consequently, for $v\sim\mathcal D_\Sig$ and any norm estimator
$\widehat r=\widehat r(Av)$, let $\oracle(y)$ denote its thresholded
response rule. Then
\begin{equation}\label{eq:norm-to-gap-loss}
 \loss_\Sig(\oracle)
 \le
 \Pr\!\left[
 |\widehat r(Av)-\|v\|_2|>
 c_{\rm err}\alpha\|v\|_2
 \right]
 +\varepsilon_{\rm nc} ,
\end{equation}
and if every signal-gap response rule has loss at least $\delta$ at $\Sig$,
then every norm estimator has relative error larger than $c_{\rm err}\alpha$
with probability at least $\delta-\varepsilon_{\rm nc}$.
\end{lemma}

\begin{corollary}[Binary norm-gap responses]
\label{cor:norm-gap-to-signal-gap}
Under the hypotheses of \Cref{lem:norm-to-gap}, use the fixed thresholds \eqref{eq:norm-gap-thresholds}.
These thresholds are independent of $\Sig$ and satisfy
$U/L=1+\Theta(\alpha)$ for fixed $c$.
For every binary answer $b\in\{-1,+1\}$, the signal-gap loss satisfies
\begin{equation}\label{eq:norm-gap-to-signal-gap}
 \loss(b,w)\le e_{\rm ng}(b,v)+\ind{\mathcal E^c},
 \qquad \Pr[\mathcal E^c]\le\varepsilon_{\rm nc},
\end{equation}
where $\mathcal E$ is the noise-concentration event in
\Cref{lem:norm-to-gap} and $e_{\rm ng}$ is defined using these thresholds.
Consequently, signal-gap loss at least $\delta$ for every response rule
implies norm-gap error probability at least $\delta-\varepsilon_{\rm nc}$
for every binary responder.
\end{corollary}

\begin{corollary}[Empirical transfer to norm errors]
\label{cor:empirical-norm-transfer}
Fix $0<\alpha\le c/20$, either norm task, and $0<\varepsilon<1$.
Use the losses and response conversion of
\Cref{lem:norm-to-gap,cor:norm-gap-to-signal-gap}.
Consider at most $N$ adaptive rounds, with queries drawn from
$\mathcal D_{\Sig_t}$ as in \texttt{Round}. Each covariance and round
length are fixed before that round's queries are drawn. Suppose
\[
 r_t\ge\frac{32}{\delta_0^2}\log\frac{N}{\varepsilon}
\]
in every round. Set $\varepsilon_{\rm nc}=\delta_0/8$ in
\Cref{lem:norm-to-gap}. With probability at least $1-\varepsilon$,
simultaneously in every executed round whose covariance satisfies
\eqref{eq:norm-reduction-admissibility},
\begin{equation}\label{eq:empirical-norm-transfer}
 \frac1{r_t}\sum_{j=1}^{r_t}e^{\rm norm}_{t,j}
 \ge \frac{E_t}{r_t}-\frac{\delta_0}{4},
\end{equation}
where $e^{\rm norm}_{t,j}$ is the loss for the chosen norm task and
$E_t$ is the realized signal-gap loss sum. In particular, a rejected
admissible round contains more than a $\delta_0/4$ fraction of
norm-task errors.
\end{corollary}

\begin{proof}
Let $F_{t,j}=\ind{\mathcal E_{t,j}^c}$ indicate failure of the
noise-concentration event in \Cref{lem:norm-to-gap}, and let $I_t$
indicate that the round's covariance satisfies
\eqref{eq:norm-reduction-admissibility}. The covariance is fixed
before the round, so $I_t$ is known before each query. Even when
admissibility holds only on a later success event, the variables
$I_tF_{t,j}$ always have conditional mean at most $\delta_0/8$.
Conditional Hoeffding gives
\[
 \Pr\!\left[\frac1{r_t}\sum_j I_tF_{t,j}>\frac{\delta_0}{4}
       \ \middle|\ \mathcal F_{\rm start}\right]
 \le e^{-r_t\delta_0^2/32}\le\frac{\varepsilon}{N}.
\]
A union bound over the executed rounds bounds the failure probability
by $\varepsilon$. In each admissible round $I_t=1$, so summing
\eqref{eq:norm-to-gap-pointwise} or
\eqref{eq:norm-gap-to-signal-gap} proves
\eqref{eq:empirical-norm-transfer}. Rejection means
$E_t/r_t>\delta_0/2$, giving the final claim.
\end{proof}

The pointwise reductions also transfer universal hardness by taking
expectations. The empirical corollary applies to both attacks:
the general attack keeps every covariance admissible, whereas the
sufficient-statistic attack establishes admissibility on its success
event. No conditioning on that event is needed in the corollary.

\paragraph{Proofs.}
\begin{proof}[Proof of \Cref{lem:norm-to-gap}]
Since $\Var(\|u\|_2^2)=2\tr(\Sig^2)$, Chebyshev's inequality gives
\[
 \Pr[\mathcal E^c]
 \le\frac{8c^4\tr(\Sig^2)}{\alpha^2}
 \le\frac{8c^4\rho}{\alpha^2}.
\]
We may therefore take
\begin{equation}\label{eq:rho-for-norm-reduction}
 c_\rho(\varepsilon_{\rm nc}):=\frac{\varepsilon_{\rm nc}}8.
\end{equation}
The admissibility bound \eqref{eq:norm-reduction-admissibility} then
ensures $\Pr[\mathcal E^c]\le\varepsilon_{\rm nc}$.
Since $\alpha\le c/20$, \Cref{def:signal-density} gives
$1/2\le w\le1+\alpha+10\alpha/c\le2+c$.
The nominal squared norms $w^2+c^2$ of
the two labeled regimes are separated by $(1+\alpha)^2-1\ge2\alpha$, so on
$\mathcal E$ we have $\|v\|_2^2\le1+c^2+\alpha/2$ whenever $w\le1$ and
$\|v\|_2^2\ge1+c^2+3\alpha/2$ whenever $w\ge1+\alpha$.

Choose $c_{\rm err}=c_{\rm err}(c)>0$ so that
$c_{\rm err}c/20\le1$ and
$3c_{\rm err}((2+c)^2+c^2+c)\le1/4$.
On $\mathcal A$, the first condition gives
$|\widehat r^{\,2}-\|v\|_2^2|\le3c_{\rm err}\alpha\|v\|_2^2$.
On $\mathcal E$, we have $\|v\|_2^2\le(2+c)^2+c^2+c$, using $\alpha\le c$.
The second condition therefore bounds the squared-norm error by
$\alpha/4$ on $\mathcal A\cap\mathcal E$.
Hence $\widehat r^{\,2}\le1+c^2+3\alpha/4$ when $w\le1$ and
$\widehat r^{\,2}\ge1+c^2+5\alpha/4$ when $w\ge1+\alpha$, so the thresholded
response is correct on both labeled regimes and its per-query loss
\eqref{eq:signal-gap-loss-function} vanishes.  Since the loss never exceeds
one, this proves \eqref{eq:norm-to-gap-pointwise}.  Taking expectations gives
\eqref{eq:norm-to-gap-loss}, and rearranging yields the final claim.
\end{proof}

\begin{proof}[Proof of \Cref{cor:norm-gap-to-signal-gap}]
On $\mathcal E$, $w\le1$ implies $\|v\|_2\le L$, whereas
$w\ge1+\alpha$ implies
$\|v\|_2^2\ge(1+\alpha)^2+c^2-\alpha/2\ge U^2$.
Thus a correct norm-gap answer is also a correct signal-gap answer on
$\mathcal E$.  The signal-gap loss is zero between the signal thresholds
and at most one elsewhere, proving \eqref{eq:norm-gap-to-signal-gap}.
Taking expectations proves the hardness claim.  Finally,
$U/L-1=\alpha/(L(U+L))=\Theta(\alpha)$, since $L$ and $U$ are bounded
above and away from zero for fixed $c$ and $0<\alpha\le c/20$.
\end{proof}

\begin{remark}[Accuracy when feasible yields a hard distribution]
\label{rem:forced-error}
Consider either covariance-update attack with
$C_{\min}/\sqrt{k}\le\alpha\le\alpha_0$, and fix either norm task.
Let $R^\star(\Sig)$ be the infimum expected loss on
$\mathcal D_\Sig$ over measurable response rules using the sketch,
even knowing $A$ and $\Sig$.
Let $\overline R_t$ be the responder's average conditional error
probability in round $t$:
\[
 \overline R_t=\frac1{r_t}\sum_{j=1}^{r_t}p_{t,j},
\]
where $p_{t,j}$ conditions on $\mathcal F_{j-1}$ for the general
attack and $\mathcal H_{j-1}$ for the sufficient-statistic attack.

Suppose the responder maintains error at most $\delta_0/8$
whenever error below $\delta_0/16$ is achievable:
\begin{equation}\label{eq:accuracy-when-feasible}
 R^\star(\Sig_t)<\delta_0/16
 \quad\Longrightarrow\quad
 \overline R_t\le\delta_0/8.
\end{equation}
Then, on the success event of the corresponding attack analysis,
the returned distribution satisfies
$R^\star(\Sig_t)\ge\delta_0/16$ and is therefore universally hard
for the chosen norm task.

Indeed, on this event a rejected round has average conditional
signal-gap loss greater than $\delta_0/4$.
Norm admissibility and
\Cref{lem:norm-to-gap,cor:norm-gap-to-signal-gap}
bound this loss by $\overline R_t+\delta_0/8$.
Thus $\overline R_t>\delta_0/8$, and
\eqref{eq:accuracy-when-feasible} implies the claimed hardness.

For $\alpha<C_{\min}/\sqrt{k}$, the isotropic construction of
\Cref{cor:isotropic-small-accuracy} already yields a universally
hard distribution on its success event, without any additional
accuracy condition on the responder.
\end{remark}

\section{Analysis of the sufficient-statistic attack}
\label{sec:sufficient-analysis}

We begin with the parameter choices and proof of
\Cref{thm:sufficient-informal}. The remaining subsections establish
alignment and direction estimation, the covariance invariants, and
termination with the stated query bound. We reuse the geometry,
round lemmas, and norm reduction from \Cref{sec:general-analysis};
the new ingredient is alignment of the mean signed noise with the
signal-estimation direction. The last subsection treats approximate
alignment.

\subsection{Parameters and proof of \texorpdfstring{\Cref{thm:sufficient-informal}}{Theorem 1.2}}
\label{sec:sufficient-parameters}
\label{sec:sufficient-norm}

\paragraph{Parameter choice.}
Use the constants $c,c_\rho,c_{\rm err},C_{\min},\alpha_0$ and
$\varepsilon_{\rm nc}=\delta_0/8$ from \Cref{sec:end-to-end}.
Choose any prescribed set $S\subseteq[n]$ of $40k+1$ coordinates,
choose $h$ uniformly from $S$, and put $M=S\setminus\{h\}$ and
$m=40k$. For $C_{\min}/\sqrt{k}\le\alpha\le\alpha_0$, set
\begin{equation}\label{eq:sufficient-rho-choice}
 \nu_0=\frac1{40k},\qquad c_R=\frac34,\qquad
 C_{\rm cap}=\frac{2}{c_Rc_g^2c^4},\qquad
 \bar\rho=C_{\rm cap}\alpha^2.
\end{equation}
The threshold $\bar\rho$ is fixed throughout the attack. The
constants chosen for the general attack also suffice here: since
$C_\lambda=96C_{\rm cap}$, the corresponding
$\lambda=C_\lambda\alpha^2$ gives
\begin{equation}\label{eq:common-admissibility-constants}
 \bar\rho=\frac{\lambda}{96},\qquad
 \frac1m+14\bar\rho^2
 \le\frac1{m-1}+\lambda^2
 \le c_\rho\frac{\alpha^2}{c^4}.
\end{equation}
Also $\bar\rho\le1/192<1/2$, while
$m\bar\rho\ge40C_{\rm cap}C_{\min}^2>1$. Thus no further change
to $C_{\min}$ or $\alpha_0$ is needed.

With $g=c_gc^2/\alpha$ and failure probability $\delta=0.04$, use
the following round schedule:
\begin{equation}\label{eq:sufficient-phase-length}
 \begin{aligned}
 &\gamma=\frac14,\qquad
 T_{\max}=1+\left\lceil\log_{1+\gamma}(m\bar\rho)\right\rceil,\qquad
 \ell_\delta=\log\frac{6(T_{\max}+1)}\delta,\\
 &\varepsilon_t=\frac1{32}\left(\frac78\right)^t,
 \qquad \underline V_t=(1+\gamma)^t\nu_0,\\
 &r_t=\left\lceil\max\left\{
       \frac{256(m+\ell_\delta)}{\varepsilon_tg^2\underline V_t},
       \frac{32\ell_\delta}{\delta_0^2}
                         \right\}\right\rceil,\\
 &r_{\rm fin}=\left\lceil\max\left\{m,
                 \frac{32\ell_\delta}{\delta_0^2}\right\}\right\rceil.
 \end{aligned}
\end{equation}
All round lengths are fixed in advance.  The angular-error targets
$\varepsilon_t$ sum to at most $1/4$ over the executed rounds, while the variance lower bounds grow
geometrically. The variance lower bound grows faster than the angular-error
target decreases, so their product increases geometrically. Thus the
direction-estimation term in $r_t$ decreases geometrically; the validation
term is fixed.

The rejection guarantee in \Cref{thm:sufficient-summary} requires
the hardness margin
\begin{equation}\label{eq:sufficient-hardness-margin}
 c_Rg^2m\bar\rho\nu_0>1.
\end{equation}
The choices in \eqref{eq:sufficient-rho-choice} make the left side
exactly $2$.

\begin{proof}[Proof of \Cref{thm:sufficient-informal}]
First consider $C_{\min}/\sqrt{k}\le\alpha\le\alpha_0$.

\emph{Initialization and rejection.}
By \Cref{cor:random}, with probability at least $0.95$ over $h$,
either $\xcol=0$ or
$\nu_0\le\Vinfo(I_m/m)<\infty$.
If $\xcol=0$, every signal-gap rule has loss
$p_c=5/(10+c)>\delta_0$, and validation rejects the first round
with probability at least $1-\delta$.
Otherwise, apply \Cref{thm:sufficient-summary} to the
sufficient-statistic answers (after thresholding for norm estimation).
The margin \eqref{eq:sufficient-hardness-margin} holds, so the
attack returns a round with empirical signal-gap loss greater than
$\delta_0/2$ with probability at least $1-\delta$, conditional on
successful initialization. On the same event every covariance
satisfies
$\tr(\Sig_t^2)\le1/m+14\bar\rho^2$ and is therefore admissible
by \eqref{eq:common-admissibility-constants}.
In the zero-signal case the isotropic covariance satisfies the same
bound. The joint rejection and admissibility event has probability
at least $0.95(1-\delta)>0.91$.

\emph{Observed errors and final-distribution hardness.}
The round schedule meets \Cref{cor:empirical-norm-transfer} with
$N=T_{\max}+1$ and $\varepsilon=\delta/6$.
This corollary applies even though admissibility was established
only on the attack's success event. Their intersection has
probability at least
$0.95(1-\delta)-\delta/6>0.9$, and the returned round contains
more than a $\delta_0/4$ fraction of norm-task errors.
If the eigenvalue threshold is reached, \Cref{thm:sufficient-summary}
also gives signal-gap loss greater than $\delta_0$ for every rule
on the final distribution. The expected-loss reductions in
\Cref{lem:norm-to-gap,cor:norm-gap-to-signal-gap} then give norm-task
error probability at least
$\delta_0-\varepsilon_{\rm nc}=7\delta_0/8$ for every responder.

\emph{Resources and remaining accuracy regimes.}
We have $m\bar\rho=\Theta(\alpha^2k)$, so
$T_{\max}=O(\log(2+\alpha^2k))$.
In \eqref{eq:sufficient-query-bound-main}, the direction-estimation
term is $O(\alpha^2k^2)$ and the validation term is
$O(\log k\log\log k)=O(k)$, with fixed constants understood.
Together with \eqref{eq:sufficient-space-bound-main}, this gives
\begin{equation}\label{eq:sufficient-specialized-query}
 \begin{aligned}
 &O(k+\alpha^2k^2)&&\text{queries},\\
 &O(\log(2+\alpha^2k))&&\text{rounds},\\
 &O(k\log(2+\alpha^2k))&&\text{words of working memory},\\
 &\widetilde O(k^2+\alpha^2k^3)&&\text{arithmetic operations}.
 \end{aligned}
\end{equation}
Every round has $\Omega(k)$ queries. Indeed, for $t<T_{\max}$,
$\underline V_t\le(1+\gamma)m\bar\rho\nu_0$, so
$g^2\underline V_t\le2(1+\gamma)/c_R=10/3$ and the direction term
in $r_t$ is $\Omega(m)$. Also $r_{\rm fin}\ge m$.
All queries are supported on $S$.

For $0<\alpha<C_{\min}/\sqrt{k}$ and $\alpha\le\alpha_0$,
\Cref{cor:isotropic-small-accuracy} gives the guarantee against
every responder in one round of $\Theta(k)$ queries; its
distribution is already universally hard on the success event.
Taking the smaller error-fraction constant from the two regimes
gives the stated guarantee. For $\alpha>\alpha_0$, run at
$\alpha_0$, as in the proof of \Cref{thm:norm-operational}; this
preserves all guarantees after adjusting absolute constants.
\end{proof}

\subsection{Alignment and direction estimation}
\label{sec:sufficient-direction}

For the supporting analysis, fix $h,M$ with
$\xcol\in\im(\Bmat)\setminus\{0\}$, so every positive definite
covariance has positive finite signal-estimation variance.
Fix one round, omit its index $t$, and write $R=\Sig^{1/2}$.
Use the shared-information version of \Cref{sec:round-analysis}:
\[
 X_j=\oracle_j(y_j)z_j,\qquad
 \mu_j=\E[X_j\mid\mathcal H_{j-1}],\qquad
 S_r=M_r+N_r,
\]
where $M_r=\sum_j\mu_j$ and $N_r=\sum_j(X_j-\mu_j)$.
Also write $\varepsilon_j=\E[e_j\mid\mathcal H_{j-1}]$ for the
predictable loss; these per-query losses are distinct from the angular
targets $\varepsilon_t$ indexed by round.

Define the unit vector
\[
 h_\Sig=\sqrt{\Vinfo(\Sig)}\,R\ahat_\Sig,
\]
which is $h_R$ from \Cref{sec:round-analysis} for this whitening.

\begin{proposition}[Exact alignment]
\label{prop:sufficient-phasewise-alignment}
For a sufficient-statistic responder, each conditional first moment is
collinear with $h_\Sig$:
\begin{equation}\label{eq:sufficient-conditional-alignment}
 \mu_j=\theta_jh_\Sig
\end{equation}
for an $\mathcal H_{j-1}$-measurable scalar $\theta_j$.
On the event $r^{-1}\sum_j\varepsilon_j\le\delta_0$,
\begin{equation}\label{eq:sufficient-average-gain}
 M_r=\sigma h_\Sig,
 \qquad \sigma=\sum_j\theta_j\ge rg\sqrt{\Vinfo(\Sig)}.
\end{equation}
\end{proposition}

\begin{proof}
Decompose the standard Gaussian as $z=\xi h_\Sig+z_\perp$, where
$\xi\sim N(0,1)$ and the centered vector $z_\perp$ is independent of
$(w,\xi)$.  The scalar channel is
\[
 T_\Sig(y)=w+c\sqrt{\Vinfo(\Sig)}\,\xi.
\]
Thus for every bounded measurable $\varphi$,
\begin{equation}\label{eq:sufficient-first-moment}
 \E[\varphi(T_\Sig(y))z]
   =\E[\varphi(T_\Sig(y))\xi]\,h_\Sig.
\end{equation}
Apply this identity to the averaged rule, conditionally on
$\mathcal H_{j-1}$.  The lower bound on $\sigma$ is the drift bound
of \Cref{lem:adaptive-phase-gain}.
\end{proof}

\begin{lemma}[Direction estimation in an adaptive round]
\label{lem:sufficient-adaptive-direction}
For every responder, every integer $r\ge1$, and every $b>0$, with
probability at least $1-e^{-b}$,
\begin{equation}\label{eq:direction-concentration}
 \|N_r\|_2\le8\sqrt{r(m+b)}.
\end{equation}
On this event, if $M_r=\sigma h$ for a unit vector $h$ with
$\sigma\ge rg\sqrt{\Vinfo(\Sig)}$ and
\begin{equation}\label{eq:sufficient-adaptive-sample-bound}
 r\ge\frac{256(m+b)}{\varepsilon g^2\Vinfo(\Sig)},
 \qquad 0<\varepsilon\le1,
\end{equation}
then $S_r\ne0$ and
$\langle S_r/\|S_r\|_2,h\rangle^2\ge1-\varepsilon$.
Both probability and implication are valid conditional on the
information before the round.
\end{lemma}

\begin{proof}
Fix a unit vector $q$.  Given $\mathcal H_{j-1}$, let
$Y=q^\top X_j$ and let $Y'$ be an independent copy of its conditional
law.  Since $|\oracle_j|\le1$, we have $|Y|\le|q^\top z_j|$, with
$q^\top z_j\sim N(0,1)$.  For every real $\lambda$, symmetrization
and convexity of $\cosh$ give
\[
 \begin{aligned}
 \E e^{\lambda(Y-\E Y)}
 &\le\E e^{\lambda(Y-Y')}
  =\E\cosh(\lambda(Y-Y'))\\
 &\le\tfrac12\E[\cosh(2\lambda Y)+\cosh(2\lambda Y')]
  \le e^{2\lambda^2}.
 \end{aligned}
\]
All expectations in this display are conditional.  Iterating the
bound over the round gives
$\Pr[q^\top N_r\ge s]\le e^{-s^2/(8r)}$ for every $s>0$.
Take a $1/2$-net of the unit sphere of size at most $5^m$ and use
$\|N_r\|_2\le2\max_q q^\top N_r$ over the net.  With
$s=4\sqrt{r(m+b)}$, the union bound gives a failure probability at
most $5^me^{-2(m+b)}\le e^{-b}$, proving
\eqref{eq:direction-concentration} without a lower bound on $r$.

Under \eqref{eq:sufficient-adaptive-sample-bound},
$\|N_r\|_2\le\sqrt\varepsilon\,\sigma/2$.
The component of $S_r$ along $h$ is at least
$\sigma(1-\sqrt\varepsilon/2)>0$, and its perpendicular component has
norm at most $\sqrt\varepsilon\,\sigma/2$.  Hence
\[
 1-\left\langle\frac{S_r}{\|S_r\|_2},h\right\rangle^2
 \le\frac{\varepsilon/4}{(1-\sqrt\varepsilon/2)^2}
 \le\varepsilon.
\]
\end{proof}

\subsection{Variance growth and covariance invariants}
\label{sec:sufficient-step}

We first bound the effect of one update using the variational formula
from \Cref{lem:variational}, then combine the bounds along the covariance
sequence.

\begin{lemma}[One covariance update]
\label{lem:natural-update-identity}
Let $\Sig\succ0$ have trace one, let $\widehat h$ be a unit vector, and
put $\Sig^+=\mathcal N(\Sig,\widehat h)$,
\[
 s=\widehat h^\top\Sig\widehat h,\qquad
 \chi=\langle\widehat h,h_\Sig\rangle^2.
\]
Then
\begin{equation}\label{eq:sufficient-step-bounds}
 \frac{\Vinfo(\Sig^+)}{\Vinfo(\Sig)}
 \ge\frac{2}{(1+s)(2-\chi)},\qquad
 \lambda_{\max}(\Sig^+)
 \le\frac{2\lambda_{\max}(\Sig)}{1+s}.
\end{equation}
\end{lemma}

\begin{proof}
Write $R=\Sig^{1/2}$ and let $z_\Sig=\Sig\ahat_\Sig$ be the minimizer
in \eqref{eq:variational-info}. Indeed, $\Bmat z_\Sig=\xcol$ and
$z_\Sig^\top\Sig^{-1}z_\Sig=\Info(\Sig)$.
The update and its inverse are
\[
 \Sig^+=\frac{R(I_m+\widehat h\widehat h^\top)R}{1+s},
 \qquad
 (\Sig^+)^{-1}
 =(1+s)R^{-1}
       \left(I_m-\frac{\widehat h\widehat h^\top}{2}\right)R^{-1}.
\]
Since $R^{-1}z_\Sig=\sqrt{\Info(\Sig)}\,h_\Sig$, evaluating the new
variational problem at the old minimizer gives
\[
 \Info(\Sig^+)
 \le z_\Sig^\top(\Sig^+)^{-1}z_\Sig
 =(1+s)\Info(\Sig)(1-\chi/2).
\]
Taking reciprocals proves the variance bound. Finally,
$pp^\top=R\widehat h\widehat h^\top R\preceq\Sig$, proving the
eigenvalue bound.
\end{proof}

\begin{lemma}[Progress with angular error]
\label{lem:multiplicative-progress}
Suppose $\lambda_{\max}(\Sig)\le\bar\rho$,
$\langle\widehat h,h_\Sig\rangle^2\ge1-\varepsilon$, with
$0\le\varepsilon\le1$. For
$\mathcal R(\Sig)=\Vinfo(\Sig)/\lambda_{\max}(\Sig)$,
\begin{equation}\label{eq:multiplicative-progress}
 \begin{aligned}
 \frac{\Vinfo(\Sig^+)}{\Vinfo(\Sig)}
 &\ge\frac{2}{(1+\bar\rho)(1+\varepsilon)},\\
 \frac{\mathcal R(\Sig^+)}{\mathcal R(\Sig)}
 &\ge\frac1{1+\varepsilon}\ge e^{-\varepsilon}.
 \end{aligned}
\end{equation}
If $\bar\rho\le1/2$ and $\varepsilon\le1/16$, then
$\Vinfo(\Sig^+)\ge(5/4)\Vinfo(\Sig)$.
\end{lemma}

\begin{proof}
Since $s\le\bar\rho$ and $2-\chi\le1+\varepsilon$,
\Cref{lem:natural-update-identity} gives the variance bound.
Dividing its variance bound by its eigenvalue bound gives
\[
 \frac{\mathcal R(\Sig^+)}{\mathcal R(\Sig)}
 \ge\frac1{2-\chi}\ge\frac1{1+\varepsilon}\ge e^{-\varepsilon}.
\]
Finally, $2/((3/2)(17/16))=64/51>5/4$.
\end{proof}

\begin{lemma}[Covariance invariants]
\label{lem:sufficient-floor-mass}
Consider a sequence of updates starting from $I_m/m$, with
squared angular errors at most $\varepsilon_t$, where
\[
 0\le\varepsilon_t\le1/16,\qquad \sum_t\varepsilon_t\le1/4.
\]
Suppose $\lambda_{\max}(\Sig_j)<\bar\rho\le1/2$ before each update.
Every iterate, including the one produced by the last update, satisfies
\[
 \mathcal R(\Sig_t)\ge c_Rm\Vinfo(I_m/m),\qquad
 \Sig_t=\beta_tI_m+P_t,\quad \beta_t>0,\quad P_t\succeq0,
 \quad \tr P_t\le7\bar\rho.
\]
Consequently, $\tr(\Sig_t^2)\le1/m+14\bar\rho^2$.
\end{lemma}

\begin{proof}
By \Cref{lem:multiplicative-progress} and $\sum_t\varepsilon_t\le1/4$,
\[
 \mathcal R(\Sig_t)\ge e^{-1/4}\mathcal R(I_m/m)
 \ge\tfrac34m\Vinfo(I_m/m).
\]
Define $s_j=\widehat h_j^\top\Sig_j\widehat h_j$ and
$\beta_t=m^{-1}\prod_{j<t}(1+s_j)^{-1}$.
The update gives $\Sig_t=\beta_tI_m+P_t$ with $P_t\succeq0$, and
\[
 \tr P_t=1-m\beta_t
 \le\sum_{j<t}s_j\le\sum_{j<t}\lambda_{\max}(\Sig_j).
\]
If there are no updates, the conclusion is immediate.  Otherwise let
$T$ be their number. Every update multiplies the variance
by at least $1+\gamma=5/4$, and
\Cref{lem:variational} bounds
$\Vinfo(\Sig_{T-1})\le m\bar\rho\,\Vinfo(I_m/m)$.
For each $j<T$, the variance growth and the lower bound on $\mathcal R$
therefore give
\[
 \lambda_{\max}(\Sig_j)
 \le\frac{\Vinfo(\Sig_j)}{c_Rm\Vinfo(I_m/m)}
 \le\frac{\bar\rho}{c_R}(1+\gamma)^{-(T-1-j)}.
\]
Summing this geometric bound, for every $t\le T$ we obtain
\[
 \tr P_t
 \le\frac{\bar\rho}{c_R}
        \sum_{i\ge0}(1+\gamma)^{-i}
 =\frac{20}3\bar\rho<7\bar\rho.
\]
Finally, put $d_t=\tr P_t$, so $m\beta_t+d_t=1$.
Every update starts below $\bar\rho$, so
\Cref{lem:natural-update-identity} gives
$\lambda_{\max}(P_t)\le\lambda_{\max}(\Sig_t)<2\bar\rho$,
including at the final iterate. Hence
\[
 \tr(\Sig_t^2)
 =\frac{1-d_t^2}{m}+\tr(P_t^2)
 \le\frac1m+2\bar\rho d_t
 \le\frac1m+14\bar\rho^2.
\]
\end{proof}

\subsection{Termination and query complexity}
\label{sec:sufficient-proof}

We combine direction estimation with the covariance invariants to
prove the signal-gap guarantee used in \Cref{sec:sufficient-norm}.
The following statement also records the generic dependence on
$m$, the initial variance lower bound $\nu_0$, and the failure
probability $\delta$.

\begin{proposition}[Rejection and threshold hardness]
\label{thm:sufficient-summary}
Fix $h,M$ with $m\ge2$ and
$\xcol\in\im(\Bmat)\setminus\{0\}$, and let $\alpha>0$ and $0<c\le1$.
Suppose $0<\nu_0\le\Vinfo(I_m/m)$ and $1/m<\bar\rho\le1/2$.
For $0<\delta<1$, use the round schedule
\eqref{eq:sufficient-phase-length}. Under the margin
\eqref{eq:sufficient-hardness-margin},
\Cref{alg:sufficient-statistic-discovery} returns a rejected round
with probability at least $1-\delta$ against every
sufficient-statistic responder.  The round has empirical signal-gap
loss greater than $\delta_0/2$.

On the same event, every covariance produced satisfies
\begin{equation}\label{eq:sufficient-admissible}
 \tr(\Sig^2)\le\frac1m+14\bar\rho^{\,2}.
\end{equation}
If the final round is reached, its covariance satisfies
\begin{equation}\label{eq:sufficient-theorem-output}
 \begin{aligned}
 &\bar\rho\le\lambda_{\max}(\Sig_{\rm out})<2\bar\rho,\\
 &\Vinfo(\Sig_{\rm out})
   \ge c_Rm\lambda_{\max}(\Sig_{\rm out})\Vinfo(I_m/m)
   \ge c_Rm\bar\rho\,\Vinfo(I_m/m).
 \end{aligned}
\end{equation}
Consequently, that final distribution is universally hard: every
signal-gap rule has loss greater than $\delta_0$ on it.
\end{proposition}

Positive definiteness, trace one, and the eigenvalue bound
$\lambda_{\max}(\Sig_t)<2\bar\rho$ hold on every run.
The sharper admissibility bound \eqref{eq:sufficient-admissible}
requires successful direction estimation and is asserted on the
theorem's high-probability event. By \Cref{lem:variational},
\eqref{eq:sufficient-theorem-output} attains at least three quarters
of the largest variance possible with top eigenvalue at most
$\lambda_{\max}(\Sig_{\rm out})$.

There are at most $T_{\max}$ rounds and one final round.  The total query
count is
\begin{equation}\label{eq:sufficient-query-bound-main}
 Q=O\!\left(
     \frac{m+\log(1/\delta)}{g^2\nu_0}
     +T_{\max}\log\frac{T_{\max}+1}{\delta}+m
       \right).
\end{equation}
The implementation uses
\begin{equation}\label{eq:sufficient-space-bound-main}
 O(mT_{\max})\ \text{words of working memory},\qquad
 O(QmT_{\max}+mT_{\max}^3)\ \text{arithmetic operations}.
\end{equation}

\begin{proof}[Proof of \Cref{thm:sufficient-summary}]
For each executed round, use the validation event of
\Cref{lem:adaptive-phase-validation} and the concentration event of
\Cref{lem:sufficient-adaptive-direction} with $b=\ell_\delta$.
Their conditional failure probabilities are at most
$2e^{-\ell_\delta}$ and $e^{-\ell_\delta}$, respectively.
Include the final validation event if that round is reached.
A union bound gives an event $\mathcal E$ of probability at least
\[
 1-(3T_{\max}+2)e^{-\ell_\delta}\ge1-\delta/2
\]
on which all these conclusions hold.  No conditioning on acceptance
is used to establish this probability bound.

\emph{Accepted rounds.}
Work on $\mathcal E$.  If round $t$ is accepted, its average
predictable loss is below $\delta_0$.
By induction over accepted rounds and
\Cref{prop:sufficient-phasewise-alignment,lem:sufficient-adaptive-direction},
\begin{equation}\label{eq:sufficient-induction}
 \Vinfo(\Sig_t)\ge(1+\gamma)^t\Vinfo(I_m/m)\ge\underline V_t,
 \qquad
 \langle\widehat h_t,h_{\Sig_t}\rangle^2\ge1-\varepsilon_t.
\end{equation}
Indeed, the variance lower bound holds initially and is preserved by
\Cref{lem:multiplicative-progress}.  Given that bound, exact alignment
and the round length yield the angular estimate.  The gain lemma is
valid throughout, so no separate variance regime is needed.

\emph{Termination and hardness.}
If all $T_{\max}$ rounds were accepted without reaching the threshold,
the variance growth would give
$\Vinfo(\Sig_{T_{\max}})\ge(1+\gamma)^{T_{\max}}\Vinfo(I_m/m)
>m\bar\rho\,\Vinfo(I_m/m)$.
This contradicts the stopping condition and \Cref{lem:variational}:
\begin{equation}\label{eq:sufficient-upper-before-cap}
 \Vinfo(\Sig)\le m\lambda_{\max}(\Sig)\Vinfo(I_m/m)
                \le m\bar\rho\,\Vinfo(I_m/m).
\end{equation}
Thus a round is rejected or an update reaches the threshold.
In the latter case, \Cref{lem:sufficient-floor-mass} gives
\eqref{eq:sufficient-theorem-output}, and the margin
\eqref{eq:sufficient-hardness-margin} implies
$g^2\Vinfo(\Sig_{\rm out})>1$.
By \Cref{cor:gain-hardness}, every response rule has loss greater
than $\delta_0$ there, so the final validation
round is rejected.  The same covariance-invariants lemma proves
\eqref{eq:sufficient-admissible} in either case, since only accepted
rounds produce updates and every update starts below the threshold.

\emph{Queries.}
The direction-estimation terms sum geometrically, because
\[
 \frac1{\varepsilon_t\underline V_t}
 =\frac{32}{\nu_0}\left(\frac{32}{35}\right)^t.
\]
Also $T_{\max}=O(1+\log(m\bar\rho))=O(\log m)$ and
$\ell_\delta=O(m+\log(1/\delta))$.
Summing the direction terms, validation terms, and ceilings, and
including the final round, proves
\eqref{eq:sufficient-query-bound-main}.  The implementation bounds are
proved in \Cref{sec:sufficient-implementation}.
\end{proof}

\subsection{Approximate alignment}

\begin{remark}[Allowing correlation with orthogonal noise]
\label{rem:sufficient-computation}
Exact dependence on the sufficient statistic is not necessary:
the attack also tolerates a small component of the mean signed
noise orthogonal to the signal-estimation direction.
For round $t$, write
\[
 M_t=\sum_{j=1}^{r_t}\E[b_jz_j\mid\mathcal H_{j-1}],
 \qquad
 \sigma_t=\langle M_t,h_{\Sig_t}\rangle,
 \qquad
 d_t=M_t-\sigma_t h_{\Sig_t}.
\]
Thus $\sigma_t$ measures the useful correlation and $d_t$ the
perpendicular correlation. For sufficient-statistic responders,
$d_t=0$.

More generally, suppose that every round with average predictable
signal-gap loss at most $\delta_0$ satisfies
\[
 \|d_t\|_2\le\frac14\sqrt{\varepsilon_t}\,\sigma_t.
\]
The gain lemma gives
$\sigma_t\ge r_tg\sqrt{\Vinfo(\Sig_t)}>0$.
Here $\varepsilon_t$ is the target squared angular error of the
estimated direction, not a tolerance for response errors.

Replace $256$ by $1024$ in the direction-estimation term of
\eqref{eq:sufficient-phase-length}.
Then, on the concentration event, the sampling error also has
norm at most $\sqrt{\varepsilon_t}\,\sigma_t/4$.
Consequently,
\[
 \|S_t-\sigma_t h_{\Sig_t}\|_2
 \le\frac12\sqrt{\varepsilon_t}\,\sigma_t.
\]
The same normalization argument as in
\Cref{lem:sufficient-adaptive-direction} gives squared angular
error at most $\varepsilon_t$. The remaining analysis is
unchanged.

\paragraph{A uniform allowance for perpendicular correlation.}
The geometric choice of $\varepsilon_t$ is convenient but not
essential. The covariance analysis only requires
$\varepsilon_t\le1/16$ and
$\sum_{t<T_{\max}}\varepsilon_t\le1/4$.
This permits a schedule that accommodates the same relative
perpendicular correlation in every accurate round.
Specifically, suppose
\[
 \frac{\|d_t\|_2^2}{\sigma_t^2}
 \le\varepsilon'
 \le\frac1{256T_{\max}}.
\]
Choose
\[
 \varepsilon_t
 =\max\left\{\frac1{16(t+1)^2},\,16\varepsilon'\right\}.
\]
Then the required drift bound holds, each
$\varepsilon_t\le1/16$, and
\[
 \sum_{t<T_{\max}}\varepsilon_t
 \le\frac{\pi^2}{96}+\frac1{16}<\frac14.
\]
Moreover, $1/\varepsilon_t\le16(t+1)^2$, while the variance
lower bound grows as $(5/4)^t$. Hence the direction-estimation
costs still have a bounded geometric-weighted sum, preserving
the same asymptotic query bound.

The extension therefore covers responders that use other
sketch components, provided their perpendicular correlations
satisfy the stated bound. Applying it to a learned predictor
requires establishing this correlation bound; small prediction
loss alone does not supply it.
\end{remark}

\section{Optimality of the covariance attacks}
\label{app:covariance-lower-bound}

We use a Gaussian Johnson--Lindenstrauss matrix, sampled once and kept fixed throughout the interaction, together with a constructed responder
that uses no interaction history and less distributional information
than our informed model allows. We show that  they withstand $\Omega(k+\alpha^2k^2)$ adaptive queries from a broad class of attacks.
Specifically, for any fixed target error fraction $\delta\in(0,1)$ and constant $\ell>0$, with high constant probability, no consecutive interval of at least $\ell k$ queries has error fraction at least $\delta$.

The attacks use signal--noise distributions whose noise energy is sufficiently concentrated. The signal direction and joint distribution may change after every answer, and the noise may depend on the signal.
The responder knows the matrix and task parameters and receives each input through its sketch. Its only additional per-query information is the signal direction and the first two moments of the noise. The constants in the query guarantee and the noise-concentration requirement may depend on $\delta$ and $\ell$.

The proof uses two facts. First, for any fixed query law, the estimation-error variance of the linear statistic used by our responder is $O(1/k)$ except with probability exponentially small in $k$ over the Gaussian matrix.
Second, while this variance remains $O(1/k)$, the responder can answer accurately while masking the matrix-dependent estimation error with output noise of scale $\alpha$. Each such query contributes only $O(1/(\alpha^2k))$ to a KL divergence bound relative to an ideal interaction that reveals nothing about $A$.
Comparing the interactions up to the first selection of a large-variance law shows that reaching such a law with constant probability requires $\Omega(\alpha^2k^2)$ queries.
Concentration then converts the responder's conditional accuracy into the observed-error guarantee.
We first state the model and guarantees, then give the proof.

\subsection{Attack and responder models}
\label{sec:lb-model}

\begin{definition}[Announced signal--noise attacks]
\label{def:covariance-attack}
Before query $i$, the attacker announces a unit vector $s_i\in\R^n$
and a joint law $\mathcal P_i$ of $(w_i,u_i)$ satisfying
\begin{equation}\label{eq:lb-law}
 w_i\in[1/2,2],\qquad u_i\in s_i^\perp\ \text{almost surely},\qquad
 \mu_i:=\E u_i,\qquad \Gamma_i:=\Cov(u_i),\qquad
 \tr\Gamma_i\le1.
\end{equation}
It then takes a fresh draw from this law and queries
$v_i=w_i s_i+u_i$. The announcement may depend on its complete previous
execution and private randomness. Conditional on that execution and the
announcement, the fresh draw is independent of the hidden matrix $A$.
The attacker may retain all sampling randomness. It observes the answers, but does not see $A$ or the sketches $Av_i$.
\end{definition}

The attacker may change the law after every query and there is no restriction on the
number of adaptive rounds. For norm tasks, we additionally require
\begin{equation}\label{eq:lb-energy}
 \Var(\|u_i\|_2^2)\le c_{\rm ad}\alpha^2
\end{equation}
for a sufficiently small constant $c_{\rm ad}>0$. This is a promise on
each announced law, not additional information required by the
responder. For centered Gaussian noise,
$\Var(\|u_i\|_2^2)=2\tr(\Gamma_i^2)$, so a sufficiently small
bound on $\tr(\Gamma_i^2)$ implies \eqref{eq:lb-energy}.
For noncentered Gaussian noise, the variance is
$2\tr(\Gamma_i^2)+4\mu_i^\top\Gamma_i\mu_i$.
For general noise, this variance also depends on fourth moments
and need not be controlled by the mean and covariance alone.

This model allows non-Gaussian noise with nonzero mean and signal
dependence. Our attacks use the special case
$s_i=e_h$ and $u_i=c\,\iota_M u$, where
$u\sim N(0,\Sig)$ is independent of the signal. Thus
$\mu_i=0$ and
$\Gamma_i=c^2\iota_M\Sig\iota_M^\top$:
the noise multiplier $c$ is absorbed into $u_i$ in this section.

\paragraph{The defending responder.}
Draw $A\in\R^{k\times n}$ with independent $N(0,1/k)$ entries and keep
it fixed. For the current announcement, suppress the query index and put
\[
 x=As,\qquad K=A\Gamma A^\top,\qquad
 \nu=
 \begin{cases}
  (x^\top K^\dagger x)^{-1},&x\in\im(K)\setminus\{0\},\\
  0,&x\notin\im(K),\\
  +\infty,&x=0.
 \end{cases}
\]
Here $K^\dagger$ denotes the Moore--Penrose inverse. When $\nu<\infty$,
it is the variance of the estimation error $T-w$ below; for our Gaussian
query family it equals $c^2\Vinfo(\Sig)$.

Given the sketch $y=Av$, form $\widetilde y=y-A\mu$. If
$0<\nu<\infty$, the responder computes
\begin{equation}\label{eq:lb-statistic}
 T=\nu x^\top K^\dagger\widetilde y.
\end{equation}
If $\nu=0$, let $P$ be the orthogonal projector onto $\im(K)^\perp$ and
use $T=\langle Px,\widetilde y\rangle/\|Px\|_2^2$. This recovers $w$
exactly: $A(u-\mu)\in\im(K)$ almost surely, since a centered random
vector lies in the image of its covariance almost surely. When $x=0$,
set $T=0$. For adaptive signal directions, we include this exceptional
case in the bad event controlled by the proof.

Conditional on $A$, the past, and any announcement with $\nu<\infty$,
\begin{equation}\label{eq:lb-error}
 T=w+\eta,\qquad
 \E[\eta]=0,\qquad \E[\eta^2]=\nu.
\end{equation}
Indeed, when $\nu>0$, $\eta=\nu x^\top K^\dagger A(u-\mu)$ and its
second moment is $\nu^2x^\top K^\dagger K K^\dagger x=\nu$; when
$\nu=0$, $\eta=0$. No independence between $w$ and $u$ is needed.

The responder adds fresh independent output noise,
\begin{equation}\label{eq:lb-noisy-response}
 o=T+\xi,\qquad \xi\sim N(0,\sigma^2),
\end{equation}
and answers norm queries using
\begin{equation}\label{eq:lb-norm-estimate}
 \widehat r=\sqrt{o^2+\tr\Gamma+\|\mu\|_2^2}.
\end{equation}
For norm estimation it returns $\widehat r$; for the binary norm gap it
returns $+1$ exactly when $\widehat r^{\,2}\ge(L^2+U^2)/2$. The
responder uses only $A,s,\mu,\Gamma$ and the task parameters. It needs
neither the signal law nor the full noise law. For the lower bound, we
give the attacker the more informative observation $o$ itself; the
task answer is determined by $o$ and the announcement.

For Gaussian noise independent of the signal, $T$ is the scalar
sufficient statistic used in our attacks. In particular, this responder
belongs to the sufficient-statistic class on that query family. We do
not require statistical sufficiency for other noise laws.

\subsection{Lower-bound guarantees}
\label{sec:lb-results}

Fix an error fraction $\delta\in(0,1)$ and a minimum interval length
$\ell k$, where $\ell>0$ is constant. We consider either relative norm
estimation with tolerance $\kappa\alpha$, for a fixed $\kappa>0$, or a
binary norm gap with fixed thresholds satisfying $U^2-L^2=\alpha$.
Taking $\kappa=c_{\rm err}$ recovers our norm-estimation loss;
the stated norm-gap width includes our thresholds. Write $e_i\in\{0,1\}$
for the corresponding task-error indicator. An observed-error success
means returning a consecutive interval $J$ of already queried inputs
with $|J|\ge\ell k$ and $\sum_{i\in J}e_i\ge\delta|J|$. In particular,
this includes returning a round with a constant fraction of errors.

\begin{theorem}[Lower bound for announced signal--noise attacks]
\label{thm:covariance-lower-bound}
Fix the norm task, $\delta$, and $\ell$ as above. There are positive
constants $C,C_1,c_0,c_{\rm ad}$ and $k_0$ such that the following holds
for every $k\ge k_0$, $n>k$, and
\[
 0<\alpha\le1,\qquad \alpha^2 k\ge C.
\]
Let $A$ have independent $N(0,1/k)$ entries. Use the responder
\eqref{eq:lb-noisy-response}--\eqref{eq:lb-norm-estimate} with
\[
 \sigma=\alpha/C_1,\qquad
 Q_0=\lfloor c_0\alpha^2k^2\rfloor.
\]
Against every attack satisfying \Cref{def:covariance-attack} and
\eqref{eq:lb-energy}, with probability at least $0.9$ we have
$\nu_i<4e/k$ for all $i\le Q_0$. With probability at least $0.8$, every
consecutive interval $J\subseteq\{1,\ldots,Q_0\}$ with $|J|\ge\ell k$
satisfies
\begin{equation}\label{eq:lb-empirical-guarantee}
 \frac1{|J|}\sum_{i\in J}e_i<\delta.
\end{equation}
The probabilities include the matrix, attacker, and responder randomness.
\end{theorem}

The empirical guarantee controls the fraction of errors in every long
enough interval. It does not assert that all answers are correct.
It holds even when the attacker chooses the returned interval after
seeing all the answers, and does not require the queries within that
interval to share a distribution.

\begin{corollary}[Query complexity]
\label{cor:covariance-optimal}
In the setting of \Cref{thm:covariance-lower-bound}, an attack achieving
observed-error success with probability at least $0.9$ for every fixed
matrix and every informed responder requires
$\Omega(k+\alpha^2k^2)$ queries in the worst case. The same lower bound
holds for its worst-case expected query count.
\end{corollary}

\begin{proof}
Average the attack's success guarantee over Gaussian $A$. The theorem
bounds the probability of success using at most $Q_0$ queries by $0.2$.
Thus, if $Q$ is the total query count, $\Pr[Q>Q_0]\ge0.7$ and
$\E Q\ge0.7Q_0$. Also $\E Q\ge0.9\ell k$, since a successful execution
returns at least $\ell k$ queries. The average over $A$ lower bounds the
worst case over fixed matrices. An attack that terminates early may be
padded with arbitrary admissible queries for the purpose of applying
the theorem.
\end{proof}

\paragraph{Comparison with our attacks.}
For our centered Gaussian query laws,
$\Var(\|c\,\iota_Mu\|_2^2)=2c^4\tr(\Sig^2)$.
The general attack has
$\tr(\Sig^2)\le\lambda^2+1/(m-1)$, and the sufficient-statistic attack
has $\tr(\Sig^2)\le1/m+14\bar\rho^{\,2}$ on its success event.
Since $m=40k$ and $\lambda,\bar\rho=\Theta(\alpha^2)$, both meet
\eqref{eq:lb-energy} after decreasing the fixed upper accuracy cutoff
and increasing the fixed lower bound on $\alpha^2k$, if necessary.
The sufficient-statistic attack can enforce its covariance bound on
every execution by aborting before a query that would violate it,
without changing its success guarantee. These constant adjustments do
not change either query bound. Consequently, the sufficient-statistic
attack is optimal within this family, and the general attack is optimal
up to its logarithmic factor. The lower bound is witnessed by a
distribution over matrices; it does not claim hardness for every fixed
matrix.

\subsection{Proof ingredients}
\label{sec:lb-proof-ingredients}

\begin{lemma}[A fixed law rarely has large variance]
\label{lem:lb-fixed-law}
For any fixed unit vector $s$ and covariance $\Gamma\succeq0$ with
$\Gamma s=0$ and $\tr\Gamma\le1$, and every $v>0$,
\begin{equation}\label{eq:lb-fixed-tail}
 \Pr_A[\nu\ge v]\le\left(\frac{e}{kv}\right)^{k/2}.
\end{equation}
In particular, $\Pr_A[\nu\ge4e/k]\le2^{-k}$.
\end{lemma}

\begin{proof}
Gaussian rotational invariance and $\Gamma s=0$ imply that
$x=As\sim N(0,I_k/k)$ is independent of $K=A\Gamma A^\top$.
If $K$ is singular, then $x\notin\im(K)$ almost surely, so $\nu=0$.
For $K\succ0$ and $a>0$, the Gaussian quadratic-form identity gives
\begin{align*}
 \Pr[\nu\ge v\mid K]
 &\le e^{a/v}\E[e^{-a x^\top K^{-1}x}\mid K]\\
 &=e^{a/v}\det\!\left(I_k+\frac{2a}{k}K^{-1}\right)^{-1/2}\\
 &\le e^{a/v}\left(\frac{k}{2a}\right)^{k/2}\sqrt{\det K}.
\end{align*}
Moreover, $\E\det K\le k^{-k}$. To see this, the rows of
$A\Gamma^{1/2}$ are independent Gaussian vectors with covariance
$\Gamma/k$. Gram--Schmidt expresses their Gram determinant as the
product of their squared distances from the preceding row spans.
Each distance has conditional expectation at most
$\tr\Gamma/k\le1/k$, so successive conditioning gives the bound.
Thus $\E\sqrt{\det K}\le k^{-k/2}$. Averaging the preceding display
and taking $a=kv/2$ proves \eqref{eq:lb-fixed-tail}.
\end{proof}

The next lemma converts this small prior probability into a bound for
adaptively selected laws. It conditions on the attacker's complete fresh
draw, not just its query.
All logarithms below are natural.

\begin{lemma}[Information needed to reach a rare query law]
\label{lem:lb-first-passage}
For each announcement $\mathcal D$, let
$\mathsf{Bad}(A,\mathcal D)$ be an event with probability at most $q$
under the matrix prior, for every fixed $\mathcal D$. At query $i$, let
$Z_i$ denote the complete fresh draw and suppose the observation is
\[
 F_i(A,Z_i)+\xi_i,\qquad \xi_i\sim N(0,\sigma^2),
\]
with independent fresh $\xi_i$. The functions below may depend on the
announcement. Suppose an ideal response $f_i(Z_i)$ does not otherwise
depend on $A$, and that whenever the current announcement is not bad,
\begin{equation}\label{eq:lb-comparison-mse}
 \E[(F_i(A,Z_i)-f_i(Z_i))^2
       \mid A,\text{past and announcement}]\le b.
\end{equation}
Let $\tau$ be the first query whose announcement is bad. For every
horizon $N$ with $0<Nq<1$,
\begin{equation}\label{eq:lb-first-passage}
 \Pr[\tau\le N]
 \le\frac{Nb/(2\sigma^2)+\log2}{\log(1/(Nq))}.
\end{equation}
\end{lemma}

\begin{proof}
Let $\mathbb P$ be the joint law of $A$ and the attacker's execution
when the responses are real before $\tau$ and ideal,
$f_i(Z_i)+\xi_i$, from query $\tau$ onward. Let $\mathbb Q$ use ideal
responses throughout. These modified interactions are used only in
the proof. The real interaction and $\mathbb P$ agree through the
announcement at the first bad query, and hence have the same crossing
probability. Under $\mathbb Q$, the attacker's execution is independent
of $A$. A union bound therefore gives
$\mathbb Q(\tau\le N)\le Nq$.

The two interactions use the same matrix prior, the same adaptive
announcement rule, and the same fresh-draw kernels. Conditional on $A$,
the past, and $Z_i$, their observation kernels are Gaussians with common
variance $\sigma^2$. The KL chain rule gives
\begin{align*}
 D_{\rm KL}(\mathbb P\|\mathbb Q)
 &=\frac1{2\sigma^2}\sum_{i=1}^N
   \E_{\mathbb P}\!\left[
    \ind{i<\tau}(F_i(A,Z_i)-f_i(Z_i))^2\right]\\
 &\le\frac{Nb}{2\sigma^2}.
\end{align*}
The event $i<\tau$ is determined before the current fresh draw, so
\eqref{eq:lb-comparison-mse} applies. For
$p=\mathbb P(\tau\le N)$, data processing for the indicator of this
event gives
\[
 D_{\rm KL}(\mathbb P\|\mathbb Q)
 \ge p\log\frac1{Nq}-\log2.
\]
Combining the bounds proves the claim. If the event has zero
$\mathbb Q$-probability, it also has zero $\mathbb P$-probability by
absolute continuity.
\end{proof}

\subsection{Proof of the lower bound}
\label{sec:lb-proof}

\begin{proof}[Proof of \Cref{thm:covariance-lower-bound}]
\textbf{Finding a law with large variance takes many queries.}
Call an announcement bad when $\nu_i\ge4e/k$, including $As_i=0$.
By \Cref{lem:lb-fixed-law}, its probability is at most $2^{-k}$ for
every fixed announcement. Apply \Cref{lem:lb-first-passage} with
\[
 F_i=T_i,\qquad f_i=w_i,\qquad b=4e/k,\qquad q=2^{-k}.
\]
The mean-square bound is \eqref{eq:lb-error}. If
$\log N\le(k\log2)/2$, then
\begin{equation}\label{eq:lb-crossing}
 \Pr[\tau\le N]
 \le\frac{4e}{\log2}\frac{N}{\sigma^2k^2}+\frac2k.
\end{equation}
With $N=Q_0$, $\sigma=\alpha/C_1$, and $c_0$ sufficiently small
relative to $C_1$, this is at most $0.1$ for all sufficiently large $k$.
The logarithmic restriction holds because $Q_0\le c_0k^2$.

\smallskip\noindent
\textbf{Below this variance, the responder is accurate.}
Put $\varepsilon=\delta/4$. Conditional on $A$, the past, and any
announcement with $\nu_i<4e/k$,
\begin{equation}\label{eq:lb-response-mse}
 \E[(o_i-w_i)^2]\le4e/k+\sigma^2.
\end{equation}
Choose constants $a,d>0$ depending only on the norm task, with $a\le1$.
Taking $C,C_1$ sufficiently large ensures
\[
 \Pr[|o_i-w_i|>a\alpha]\le\varepsilon/2
\]
by \eqref{eq:lb-response-mse} and $\alpha^2k\ge C$.
On the complementary event, $w_i\le2$ and $\alpha\le1$ give
$|o_i^2-w_i^2|\le5a\alpha$. Also, \eqref{eq:lb-energy} gives
\[
 \Pr\!\left[
  \left|\|u_i\|_2^2-\tr\Gamma_i-\|\mu_i\|_2^2\right|>d\alpha
 \right]\le c_{\rm ad}/d^2\le\varepsilon/2,
\]
after choosing $c_{\rm ad}$ sufficiently small. The probabilities in
these displays have the same conditioning as \eqref{eq:lb-response-mse}.
Since $u_i\perp s_i$, outside the two exceptional events,
\begin{equation}\label{eq:lb-squared-error}
 \left|\widehat r_i^{\,2}-\|v_i\|_2^2\right|
 \le(5a+d)\alpha.
\end{equation}

For the norm gap take $a=1/40$ and $d=1/8$. The squared-norm error is
at most $\alpha/4$, less than the $\alpha/2$ margin between the decision
threshold $(L^2+U^2)/2$ and either endpoint. For relative norm estimation,
take $a=\min\{1,\kappa/40\}$ and $d=\min\{1/8,\kappa/8\}$.
Then $5a+d\le\kappa/4$. Since $\|v_i\|_2^2\ge w_i^2\ge1/4$,
\[
 \frac{|\widehat r_i-\|v_i\|_2|}{\|v_i\|_2}
 \le\frac{|\widehat r_i^{\,2}-\|v_i\|_2^2|}{\|v_i\|_2^2}
 \le\kappa\alpha.
\]
Thus the conditional task-error probability is at most $\varepsilon$
whenever $\nu_i<4e/k$.

\smallskip\noindent
\textbf{No long interval has many errors.}
Define $X_i=e_i\ind{i<\tau}$. Its conditional mean given $A$ and the
past is at most $\varepsilon$, averaging also over the next announcement
if it is randomized. For every $\theta>0$,
\[
 \E[e^{\theta X_i}\mid A,\text{past}]
 \le1-\varepsilon+\varepsilon e^\theta.
\]
Iterating this bound and optimizing over $\theta$ gives, for every
fixed consecutive interval $J$,
\[
 \Pr\!\left[\sum_{i\in J}X_i\ge\delta|J|\right]
 \le\exp\{-|J|\operatorname{kl}(\delta\|\varepsilon)\},
\]
where $\operatorname{kl}(p\|q)=p\log(p/q)+(1-p)\log((1-p)/(1-q))$.
A union bound over at most $Q_0^2$ intervals of length at least $\ell k$
bounds their total failure probability by
\[
 Q_0^2\exp\{-\ell k\operatorname{kl}(\delta\|\varepsilon)\}\le0.1
\]
for sufficiently large $k$. On $\tau>Q_0$, $X_i=e_i$ throughout these
queries. Combining this estimate with \eqref{eq:lb-crossing} proves
\eqref{eq:lb-empirical-guarantee}.
\end{proof}

\section{Proof of the gain lemma}
\label{app:gain-proof}

We restate the lemma for reference.  The proof uses the scalar Gaussian
channel and integration by parts against the signal density.
\gainrestatetrue
\gainlemma*
\gainrestatefalse

\paragraph{Signal density.}
Recall $\ell=10\alpha/c$ from \Cref{def:signal-density}, and put
\[
 a=1-\ell,\qquad
 b=1+\alpha+\ell,\qquad C_W=(\ell+\alpha)^{-1}.
\]
Convolving the two uniform densities in \Cref{def:signal-density} gives
\begin{equation}\label{eq:signal-density}
\nu(w)=
\begin{cases}
0, & w<a\text{ or }w>b,\\
C_W(w-a)/\ell, & a\le w\le1,\\
C_W, & 1<w<1+\alpha,\\
C_W(b-w)/\ell, & 1+\alpha\le w\le b.
\end{cases}
\end{equation}

\begin{figure}[htbp]
\centering
\begin{tikzpicture}[x=1.45cm,y=1.15cm,font=\small]
  \fill[black!13] (0.6,0) -- (3,1.65) -- (3,0) -- cycle;
  \fill[black!4] (3,0) rectangle (4,1.65);
  \fill[black!13] (4,0) -- (4,1.65) -- (6.4,0) -- cycle;
  \draw[densely dotted,black!55] (0,1.65) -- (3,1.65);
  \draw[dashed,black!55] (3,0) -- (3,1.65);
  \draw[dashed,black!55] (4,0) -- (4,1.65);
  \draw[thick] (0.25,0) -- (0.6,0) -- (3,1.65) -- (4,1.65)
    -- (6.4,0) -- (6.75,0);
  \draw[->] (0,-0.06) -- (0,2.12) node[above] {$\nu(w)$};
  \draw[->] (0,0) -- (7,0) node[right] {$w$};
  \node[left] at (0,1.65) {$C_W$};
  \foreach \x/\lab in {0.6/a,3/1,4/{1+\alpha},6.4/b} {
    \draw (\x,0.045) -- (\x,-0.045);
    \node[below] at (\x,-0.05) {$\lab$};
  }
  \node at (1.8,1.94) {label $-1$};
  \node at (3.5,1.94) {no loss};
  \node at (5.2,1.94) {label $+1$};
  \draw[<->] (0.6,-0.56) -- (3,-0.56)
    node[midway,fill=white,inner sep=2pt] {$\ell$};
  \draw[<->] (3,-0.56) -- (4,-0.56)
    node[midway,fill=white,inner sep=2pt] {$\alpha$};
  \draw[<->] (4,-0.56) -- (6.4,-0.56)
    node[midway,fill=white,inner sep=2pt] {$\ell$};
\end{tikzpicture}
\caption{The signal density and the two labeled regimes.  Each ramp has
width $\ell=10\alpha/c$; the plateau is the unlabeled gap of width
$\alpha$.  The gap is widened schematically for legibility; horizontal
widths are not to scale.}
\label{fig:signal-density}
\end{figure}
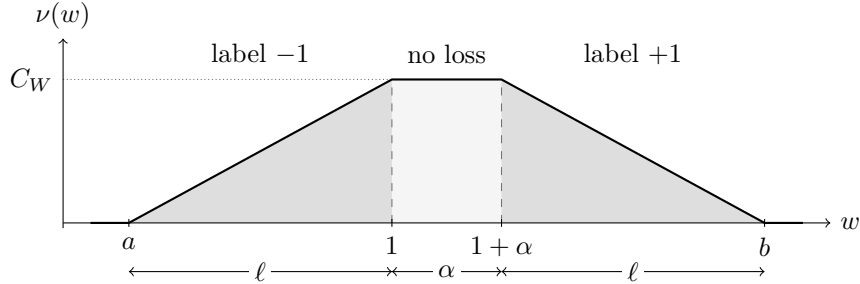

\begin{proof}
Fix $\Sig$ and $\oracle$ as in the statement, and write
$\delta:=\loss_\Sig(\oracle)\le\delta_0$.

\paragraph{Reduction to a scalar channel.}
By \eqref{eq:estimator-noise},
\[
 T_\Sigma(y)=w+X,\qquad
 X\sim N(0,s^2),\qquad s^2=c^2\Vinfo(\Sig)>0,
\]
with $X$ independent of $w$.
Sufficiency of $T_\Sigma$ gives a measurable function
\begin{equation}\label{eq:Psi-def}
 \Psi(t):=\E[\oracle(y)\mid T_\Sigma(y)=t],\qquad |\Psi(t)|\le1,
\end{equation}
such that
$\E[\oracle(y)\mid T_\Sigma(y),w]=\Psi(T_\Sigma(y))$.
Consequently, replacing the response by $\Psi(w+X)$ preserves both its
signal-gap loss and its gain:
\begin{equation}\label{eq:scalar-gain-start}
 G:=\Gamma_\oracle(\Sig)=\E[\Psi(w+X)X].
\end{equation}
This replacement is only for the proof; the responder may use the full
sketch.

\paragraph{An exact formula for the gain.}
Let $\phi_s$ be the $N(0,s^2)$ density, and define the mean response at
signal value $w$ by
\[
 f(w):=\E_X[\Psi(w+X)]
      =\int_\R\Psi(t)\phi_s(t-w)\,dt.
\]
Gaussian convolution makes $f$ continuously differentiable, even when
$\Psi$ is only measurable.  Differentiating the density gives
\begin{equation}\label{eq:gain-gaussian-derivative}
 f'(w)
 =\int_\R\Psi(t)\frac{t-w}{s^2}\phi_s(t-w)\,dt
 =\frac1{s^2}\E_X[\Psi(w+X)X].
\end{equation}
Thus $G=s^2\int_\R\nu(w)f'(w)\,dw$.
The trapezoidal density $\nu$ is continuous, vanishes at $a,b$, and has
derivative $C_W/\ell$ on $[a,1]$, $-C_W/\ell$ on
$[1+\alpha,b]$, and zero elsewhere, apart from its breakpoints.
Integration by parts therefore yields the exact identity
\begin{equation}\label{eq:gain-ramp-identity}
 G=-s^2\int_\R\nu'(w)f(w)\,dw
  =\frac{C_Ws^2}{\ell}
    \left(\int_{1+\alpha}^b f(w)\,dw-\int_a^1 f(w)\,dw\right).
\end{equation}
This is the role of the ramps: the gain depends on the difference between
the average responses in the two labeled regimes.
Continuity of $\nu$ eliminates boundary terms at the internal breakpoints.

\paragraph{Using the loss bound.}
Put $H_-(w):=(1+f(w))/2$ and $H_+(w):=(1-f(w))/2$, both in $[0,1]$.
The loss condition is
\begin{equation}\label{eq:error-Psi}
 \int_a^1\nu(w)H_-(w)\,dw
 +\int_{1+\alpha}^b\nu(w)H_+(w)\,dw
 =\delta.
\end{equation}
On $[a+\ell/4,1]$ and $[1+\alpha,b-\ell/4]$, the density is at least
$C_W/4$.  The two omitted intervals have total length $\ell/2$.
Using \eqref{eq:error-Psi} on the former intervals and $H_\pm\le1$ on
the latter gives
\[
 \int_a^1H_-(w)\,dw+\int_{1+\alpha}^bH_+(w)\,dw
 \le\frac\ell2+\frac{4\delta}{C_W}.
\]
Since each ramp has length $\ell$, it follows that
\[
 \int_{1+\alpha}^b f(w)\,dw-\int_a^1 f(w)\,dw
 =2\ell-2\left(\int_a^1H_-(w)\,dw
                  +\int_{1+\alpha}^bH_+(w)\,dw\right)
 \ge\ell-\frac{8\delta}{C_W}
 \ge\frac\ell2.
\]
The last inequality uses $0<c\le1$, which gives
$C_W\ell=10/(10+c)\ge10/11$, and
$\delta\le\delta_0=1/20\le C_W\ell/16$.
Substituting into \eqref{eq:gain-ramp-identity} and using
$C_W=c/((10+c)\alpha)\ge c/(11\alpha)$, we obtain
\[
 G\ge\frac{C_Ws^2}{2}
   \ge\frac{c\,s^2}{22\alpha}
   =c_g\frac{c^3\Vinfo(\Sig)}{\alpha}
   =c\Vinfo(\Sig)g,
\]
where $s^2=c^2\Vinfo(\Sig)$ and $c_g=1/22$.
\end{proof}

\section{Universal hardness and the isotropic endpoint}
\label{app:intrinsic-hardness}

We first identify the Bayes-optimal signal-gap response, justifying
the sufficient-statistic responder class. Its risk gives a direct
hardness bound whenever the scalar-channel standard deviation is a
fixed positive multiple of the signal gap. For the high-accuracy
isotropic round, a direct norm-hardness argument complements this
signal bound, avoiding any change of signal width or norm thresholds.
The covariance-update analyses use the separate hardness threshold at
loss level $\delta_0$ from \Cref{cor:gain-hardness}.

\subsection{Optimal signal-gap prediction}

\begin{proposition}[The Bayes signal-gap response]
\label{prop:sufficient-bayes}
Let $\alpha,c>0$ and $0<\Vinfo(\Sig)<\infty$. For the signal law
$W_{\alpha,c}$, the signal-gap risk has a unique minimizing response
rule, up to null sets:
\[
 \oracle_{\rm Bayes}(y)=
 \begin{cases}
 -1,&T_\Sig(y)<1+\alpha/2,\\
 +1,&T_\Sig(y)>1+\alpha/2.
 \end{cases}
\]
The response at equality is immaterial.  In particular, every exactly
loss-optimal signal-gap responder satisfies
\Cref{def:phasewise-sufficient} whenever the scalar channel has
positive finite variance.
\end{proposition}

\begin{proof}
Write $d=1+\alpha/2$ and $s^2=c^2\Vinfo(\Sig)>0$.
By sufficiency, the posterior law of $w$ depends on the sketch only
through $T=T_\Sig(y)$. Write $\nu$ for the signal density
\eqref{eq:signal-density}, with support endpoints
$a=1-10\alpha/c$ and $b=1+\alpha+10\alpha/c$.
Let $\phi_s$ be the $N(0,s^2)$ density and put
\[
 f_-(t)=\int_a^1\nu(w)\phi_s(t-w)\,dw,
 \qquad
 f_+(t)=\int_{1+\alpha}^b\nu(w)\phi_s(t-w)\,dw.
\]
These are the unnormalized posterior masses of the two labeled regimes.
The density $\nu$ is symmetric about $d$.  Pairing $w\le1$ with
$2d-w\ge1+\alpha$ shows that $f_+(t)>f_-(t)$ when $t>d$:
each paired high signal is closer to $t$ than its low counterpart.
The inequality reverses when $t<d$.  The conditional loss is affine
in the response, with coefficient proportional to $f_-(t)-f_+(t)$,
so its unique minimizer is the stated endpoint of $[-1,1]$.
Finally, $T$ has a continuous density, so $T=d$ is a null event.
\end{proof}

\subsection{A direct hardness bound}

\begin{lemma}[Large signal-estimation variance implies signal-gap hardness]
\label{lem:large-variance-hardness}
Fix constants $0<c\le5$ and $0<\tau\le1$, and put
\begin{equation}\label{eq:delta-hard}
  \delta_{\rm hard}(c,\tau)
  :=\frac{c\tau}{10+c}\,\Phi\!\left(-1-\frac1{2\tau}\right)>0,
\end{equation}
where $\Phi$ is the standard normal distribution function.  If
\begin{equation}\label{eq:large-variance-hardness}
 c\sqrt{\Vinfo(\Sig)}\ge\tau\alpha,
\end{equation}
then every measurable response rule $\oracle:\R^k\to[-1,1]$ satisfies
$\loss_\Sig(\oracle)\ge\delta_{\rm hard}(c,\tau)$; hence
$\mathcal D_\Sig$ is universally hard for the signal gap.
\end{lemma}

\begin{proof}
Let $\nu$, $\ell=10\alpha/c$, and $C_W=1/(\ell+\alpha)$ be as in
\eqref{eq:signal-density}. The two labeled regimes each have probability
$p_c=C_W\ell/2=5/(10+c)$.
If $\Vinfo(\Sig)=+\infty$, the sketch is independent of the signal,
so every rule has loss $p_c\ge\delta_{\rm hard}(c,\tau)$.
Otherwise put $\sigma=c\sqrt{\Vinfo(\Sig)}\ge\tau\alpha>0$.
By \Cref{prop:sufficient-bayes}, thresholding $T_\Sig(y)$ at
$1+\alpha/2$ minimizes the loss. Symmetry of the signal density about
this threshold gives the minimum risk
\[
 R^\star
 =2\int_0^\ell\nu(1-x)
      \Phi\!\left(-\frac{\alpha/2+x}{\sigma}\right)\,dx.
\]
Restrict the integral to $0\le x\le s:=\tau\alpha$.
Here $s/\ell=c\tau/10\le1/2$, so
$\nu(1-x)=C_W(1-x/\ell)\ge C_W/2$. Also $\sigma\ge s$ gives
\[
 \Phi\!\left(-\frac{\alpha/2+x}{\sigma}\right)
 \ge\Phi\!\left(-1-\frac1{2\tau}\right).
\]
Consequently every response rule has loss at least
\[
 R^\star
 \ge C_W\tau\alpha\,
       \Phi\!\left(-1-\frac1{2\tau}\right)
 =\delta_{\rm hard}(c,\tau).
\]
\end{proof}

\subsection{The very high accuracy regime}

For $\alpha=O(k^{-1/2})$, the initial isotropic query distribution
is already universally hard for both the signal-gap and norm tasks,
with high constant probability over random-coordinate initialization.
Thus the attack can return its first round, without any covariance
updates.

\begin{corollary}[An isotropic round at very high accuracy]
\label{cor:isotropic-small-accuracy}
Fix constants $0<c\le1$ and $C\ge1$. For every
$A\in\R^{k\times n}$ with $n>40k$ and
$0<\alpha\le\min\{c/20,C/\sqrt{k}\}$, use random-coordinate
initialization on any prescribed $40k+1$ coordinates and run the
first round with $m=40k$, $\Sig=I_m/m$, and $W=W_{\alpha,c}$.
Keep the thresholds and response conversion of \Cref{sec:norm2signal}.
For a sufficiently large round length $r=\Theta(k)$, returning this
round gives the following guarantee against every informed responder.

There is $\delta_{\rm iso}=\delta_{\rm iso}(c,C)>0$ such that,
with probability at least $0.9$, both the empirical signal-gap loss
and the fraction of norm-task errors are at least $\delta_{\rm iso}/2$.
For norm estimation, an error means relative error greater than
$\alpha/16$. On the same event, the distribution is universally hard
for the signal-gap task and the chosen norm task, with expected loss
at least $\delta_{\rm iso}$.
\end{corollary}

\begin{proof}
With probability at least $0.95$ over $h$, \Cref{cor:random} gives
$\Vinfo(I_m/m)\ge1/(40k)$, including infinite variance. Thus
$c\sqrt{\Vinfo(I_m/m)}\ge\tau\alpha$ for
$\tau=c/(\sqrt{40}\,C)$. By \Cref{lem:large-variance-hardness},
every response rule has signal-gap loss at least
$\delta_{\rm sg}:=\delta_{\rm hard}(c,\tau)>0$.

For the norm task, fix any $h$. Let $P$ project onto the row space of
$\Bmat$, put $q=\rank(\Bmat)\le k$, and set $d=m-q\ge39k$.
Even revealing $w$ and $Pu$, which together determine the sketch,
leaves an independent Gaussian component in $\ker\Bmat$. Hence
\[
 \|v\|_2^2=w^2+\frac{c^2}{m}(Q+X),\qquad
 Q=m\|Pu\|_2^2\sim\chi_q^2,\quad X\sim\chi_d^2,
\]
with $X$ independent of the revealed observation.
On $\mathcal A:=\{w\le1,\ q-\sqrt q\le Q\le q\}$, the event
$X\le d$ implies $\|v\|_2\le L$, whereas
$X\ge d+D\sqrt d$ implies $\|v\|_2\ge U$, for a constant
$D=D(c,C)$. Indeed, use $1-w^2\le20\alpha/c$,
$q-Q\le\sqrt k$, $\alpha\sqrt k\le C$, and the thresholds
\eqref{eq:norm-gap-thresholds}.
The probabilities of $q-\sqrt q\le Q\le q$, $X\le d$, and
$X\ge d+D\sqrt d$ have positive constant lower bounds uniformly
in $q,d$: the central limit theorem gives positive limits, and each
of the finitely many remaining probabilities is positive
(take $Q=0$ when $q=0$).
Since $\Pr(w\le1)=5/(10+c)$ independently of $Q$, $\mathcal A$
also has constant probability. Conditional on each observation in
$\mathcal A$, both norm labels therefore have positive constant
probability. Every norm-gap responder consequently has expected loss
at least some $\delta_{\rm ng}=\delta_{\rm ng}(c,C)>0$.
Thresholding as in \eqref{eq:threshold-response} maps every norm
estimate with relative error at most $\alpha/16$ to a correct gap
answer, so the same bound holds for norm-estimation errors.

Set $\delta_{\rm iso}=\min\{\delta_{\rm sg},\delta_{\rm ng}\}$.
On successful initialization, both expected-loss bounds hold
conditional on every preceding history. Martingale Hoeffding and a
union bound show that both empirical losses are at least
$\delta_{\rm iso}/2$, except with probability at most
$2\exp(-r\delta_{\rm iso}^2/2)\le0.05$ for a sufficiently large
absolute factor in $r=\Theta(k)$. The overall success probability
is at least $0.95^2>0.9$. No test against $\delta_0$ is needed.
\end{proof}

\section{Implementation and resource bounds}
\label{app:implementation}

We give the eigenvalue formulas and compact covariance representations
used to implement the attacks.  The bounds count arithmetic operations
and words of working memory; they do not include a finite-precision
analysis.

\subsection{The general attack}
\label{sec:general-computation}

The iterates have a compact representation: each update adds at most one
direction to the subspace on which the covariance differs from an
isotropic matrix.

\paragraph{Entropy projection.}
Diagonalize the unprojected iterate as
$\widehat\Sig=U\diag(\widehat s_1,\ldots,\widehat s_m)U^\top$.
The search space is orthogonally invariant.  For fixed eigenvalues of
$\Sig'$, the trace inequality maximizes
$\tr(\Sig'\log\widehat\Sig)$ by aligning its eigenvectors with those of
$\widehat\Sig$ in the same order.  The remaining scalar minimization
of $\sum_i s_i\log(s_i/\widehat s_i)$ gives
\begin{equation}\label{eq:entropy-projection-clipping}
 s_i=\max\{q_0,\tau\widehat s_i\},\qquad
 \sum_i s_i=1.
\end{equation}
No upper clipping is needed: the floor and trace constraint imply
$s_i\le1-(m-1)q_0=\lambda$.
The scale $\tau>0$ can be found by sorting the breakpoints
$q_0/\widehat s_i$ and solving the linear equation on the appropriate
interval. Equal eigenvalues may be grouped with their multiplicities.

\begin{proposition}[Structure of the iterates]\label{prop:iterate-structure}
For every $t\ge0$, $\Sig_t=\beta_tI_m+U_t\Lambda_tU_t^\top$ for some
$\beta_t>0$, an $m\times d_t$ matrix $U_t$ with orthonormal columns,
$d_t\le\min\{t,m\}$, and a diagonal matrix $\Lambda_t$.
\end{proposition}

\begin{proof}
Initially $\Sig_0=I_m/m$ and $d_0=0$.  For the inductive step, let
$V=\im(U_t)+\operatorname{span}(a_t)$, of dimension at most $d_t+1$.
Both $\Sig_t$ and $Z_t$ preserve $V$; on $V^\perp$ they act as
$\beta_tI$ and zero.  Hence the exponential update also preserves $V$
and is isotropic on $V^\perp$.  The projection keeps eigenvectors and
maps equal eigenvalues to equal eigenvalues, by
\eqref{eq:entropy-projection-clipping}.  The projected iterate
therefore has the asserted form with $d_{t+1}\le\min\{d_t+1,m\}$.
\end{proof}

\paragraph{Arithmetic and working memory.}
Store $\Sig_t$ as $(\beta_t,U_t,\Lambda_t)$.  Its square root is
\[
 R_t=\sqrt{\beta_t}I_m+
 U_t\bigl((\beta_tI+\Lambda_t)^{1/2}-\sqrt{\beta_t}I\bigr)U_t^\top,
\]
and its inverse square root has the analogous form.  Applying either
matrix to a vector costs $O(m(1+d_t))$ operations.  An update expands
the basis by at most one vector, diagonalizes a matrix of order at most
$d_t+1$, and applies \eqref{eq:entropy-projection-clipping} to its
eigenvalues and the repeated eigenvalue on $V^\perp$.  Using standard
dense linear-algebra costs, this takes
$O(m(1+d_t)^2+(1+d_t)^3)$ operations.
Since $r\ge48m$, the round cost is $O(rm(1+d_t))$.

Writing $Q=rT_{\max}$, the resulting bounds are
\begin{equation}\label{eq:general-arithmetic-cost}
 \begin{aligned}
 &O\bigl(Qm\min\{T_{\max},m\}\bigr)
       &&\text{arithmetic operations},\\
 &O\bigl(m\min\{T_{\max},m\}\bigr)
       &&\text{words of working memory}.
 \end{aligned}
\end{equation}
The working memory includes the current query, the running sums, and
the coordinate support.  A returned round is identified by its query indices in the attacker's
record of generated inputs and returned answers, which can be written
as the interaction proceeds.  This record does not include the
responder's private part of $\mathsf{Tr}$.  If an explicit copy of all round
vectors is required, its $O(rm)$ output storage is additional to the
working-memory bound.  No list of past covariances is stored.

For the norm-estimation parameters, these bounds become
$\widetilde O(k^2+\alpha^4k^4)$ operations and
$\widetilde O(k+\alpha^2k^2)$ working memory.
At $\alpha=\Theta(k^{-1/2})$, the cost per query is
$\widetilde O(k)$, comparable to writing its nonzero coordinates.
At constant $\alpha$, the rank bound permits a full-dimensional
representation.

\subsection{The sufficient-statistic attack}
\label{sec:sufficient-implementation}

\begin{proposition}[Compact iterates]
\label{prop:sufficient-structure}
Every iterate has the form
$\Sig_t=\beta_tI_m+U_t\Lambda_tU_t^\top$, with $\beta_t>0$,
an $m\times d_t$ matrix $U_t$ with orthonormal columns,
$d_t\le\min\{t,m\}$, and diagonal $\Lambda_t\succeq0$.
The update and the top eigenvalue needed for the stopping test
can be computed in a subspace of dimension at most $d_t+1$.
\end{proposition}

\begin{proof}
The representation holds initially. For the inductive step, put
$p=\Sig_t^{1/2}\widehat h_t$ and $s=\|p\|_2^2$.
The square root is isotropic on $\im(U_t)^\perp$, so $p$ lies in
$V=\im(U_t)+\operatorname{span}(\widehat h_t)$.
The rank-one update preserves $V$ and remains isotropic on $V^\perp$,
with floor $\beta_t/(1+s)>0$ and a positive semidefinite remainder.
Diagonalizing the restriction to $V$ therefore gives the next compact
representation and its top eigenvalue; eigenvalues outside $V$ equal
the floor.
\end{proof}

Applying the square root to a query costs $O(m(1+d_t))$ operations.
Updating the basis and diagonalizing the restricted covariance costs
$O(m(1+d_t)^2+(1+d_t)^3)$.
Summing over at most $T_{\max}$ updates gives
$O(QmT_{\max}+mT_{\max}^3)$ operations and $O(mT_{\max})$ working memory, using
$d_t\le\min\{t,m\}$.
As in \Cref{sec:general-computation}, the returned round is identified
by its query indices in the attacker's record; storing an additional
copy of all its vectors is output storage.

\end{document}